\documentclass[10pt]{article}
\usepackage[letterpaper]{geometry}

\usepackage{amsmath}
\usepackage{amsthm}
\usepackage{mathtools,amssymb,latexsym} 
\usepackage{amsxtra} 
\usepackage[mathscr]{eucal}
\usepackage{turnstile}
\usepackage{subfigure}
\usepackage{rotating}
\usepackage[title]{appendix}

\usepackage{tikz}
\usetikzlibrary{shapes,shapes.multipart, calc,matrix,arrows,positioning,automata}
\tikzset{
	>=stealth',
	punkt/.style={
		circle,
		rounded corners,
		draw=black, thick, 
		text width=1.5em,
		minimum height=3em,
		text centered},
	punkts/.style={
		circle,
		rounded corners,
		draw=black, thick, 
		text width=1em,
		minimum height=1em,
		text centered},
	invisible/.style={
		draw=none,
		text width=1.5em,
		minimum height=0em,
		text centered},
	inv/.style={
		draw=none,
		text width=2.5em,
		minimum height=3em,
		text centered},
	pil/.style={
		->,
		thick,
		shorten <=2pt,
		shorten >=2pt,}
}

\newcommand{\yltsn}[1]{\mathscr{#1}}  
\newcommand{\yltsnS}{\yltsn{S}}  
\newcommand{\ylts}[4]{\langle #1, #2, #3, #4\rangle}  
\newcommand{\yltsS}{\ylts{S}{s_0}{L}{T}}  
\newcommand{\yltsNS}{\yltsnS=\yltsS}  
\newcommand{\yltsc}[1]{\yltsn{L}(#1)}  

\newcommand{\yfsan}[1]{\mathscr{#1}}  
\newcommand{\yfsanA}{\yfsan{A}}  
\newcommand{\yfsanAs}[1]{\yfsan{A}_{#1}}  
\newcommand{\yfsa}[5]{\langle #1, #2, #3, #4, #5\rangle}  
\newcommand{\yfsaA}{\yfsa{S}{s_0}{A}{\rho}{F}}  

\newcommand{\yioltsn}[1]{\mathscr{#1}} 
\newcommand{\yI}{\yioltsn{I}}  
\newcommand{\yS}{\yioltsn{S}}  
\newcommand{\yQ}{\yioltsn{Q}}  
\newcommand{\yA}{\yioltsn{A}}  
\newcommand{\yB}{\yioltsn{B}}  
\newcommand{\yT}{\yioltsn{T}}  
\newcommand{\yC}{\yioltsn{C}}  
\newcommand{\yD}{\yioltsn{D}}  
\newcommand{\yU}{\yioltsn{U}}  

\newcommand{\yiolts}[5]{\langle #1, #2, #3, #4,#5\rangle}  
\newcommand{\yioltsI}{\yiolts{S}{s_0}{L_I}{L_U}{T}}  

\newcommand{\yio}[5]{\langle #1, #2, L_{#3}, L_{#4}, #5\rangle}  
\newcommand{\yioS}{\yio{S}{s_0}{I}{U}{T}}  
\newcommand{\yioQ}{\yio{Q}{q_0}{I}{U}{R}}  

\newcommand{\yioNI}{\yltsn{I}=\yioQ}  

\newcommand{\yioini}[2]{\yltsn{#1}_{\!/#2\!}}  

\newcommand{\yioc}[2]{\yltsn{IO}(L_#1,L_#2)}  
\newcommand{\yiocq}[2]{\yltsn{IO}_\yde(L_#1,L_#2)}  
\newcommand{\yioce}[2]{\yltsn{IOE}(L_{#1},L_{#2})}  

\newcommand{\yioipname}{\yltsn{IOIP}}  
\newcommand{\yiocip}[2]{\yioipname(L_#1,L_#2)}  

\newcommand{\yioipmname}{\yltsn{IOMIN}}  
\newcommand{\yiocipm}[2]{\yioipmname(L_#1,L_#2)}  
\newcommand{\yiocipmk}[3]{\yioipmname(L_#1,L_#2,#3)}  

\newcommand{\ytr}[3]{#1\overset{#2}{\rightarrow}#3} 
\newcommand{\ytru}[4]{#1\overset{#2}{\underset{#3}{\rightarrow}}#4} 
\newcommand{\ytrut}[4]{#1\overset{#2}{\underset{#3}{\Rightarrow}}#4} 
\newcommand{\ytrt}[3]{#1\overset{#2}{\Rightarrow}#3} 

\newcommand{\yatrt}[3]{#1\overset{#2}{\mapsto}#3} 

\newcommand{\ypass}{\text{\bfseries\sffamily pass}} 
\newcommand{\yfail}{\text{\bfseries\sffamily fail}} 

\newcommand{\yconf}[2]{\,\text{\bf conf}_{#1,#2}\,\,} 
\newcommand{\yioco}{\text{\,\,\bf ioco}\,\,} 
\newcommand{\yafter}{\text{\,\,\bf after}\,\,} 
\newcommand{\yout}{\text{\,\bf out}} 
\newcommand{\yinp}{\text{\,\bf inp}} 
\newcommand{\yinit}{\text{\,\bf init}} 

\newtheorem{theo}{Theorem}[section] 
\newtheorem{lemm}[theo]{Lemma} 
\newtheorem{coro}[theo]{Corollary} 
\newtheorem{prop}[theo]{Proposition}
\newtheorem{defi}[theo]{Definition} 

\newtheorem{rema}[theo]{Remark}
\newtheorem{hypo}[theo]{Hypothesis}

\theoremstyle{definition}
\newtheorem{exam}[theo]{Example}

\newcommand{\yst}{\,|\,} 
\newcommand{\yoh}[1]{\mathcal{O}(#1)} 
\newcommand{\yomeg}[1]{\Omega(#1)} 

\newcommand{\yeps}{\varepsilon} 
\newcommand{\yfim}{\hfill$\Box$} 

\newcommand{\ysse}{\subseteq}
\newcommand{\yemp}{\emptyset}
\newcommand{\ypow}[1]{\mathcal{P}(#1)}
\newcommand{\ycomp}[1]{\overline{#1}}

\newcommand{\ysi}{\sigma}
\newcommand{\ySi}{\Sigma}
\newcommand{\yal}{\alpha}
\newcommand{\ybe}{\beta}

\newcommand{\yde}{\delta}
\newcommand{\yte}{\theta}

\newcommand{\yltscA}[1]{\mathscr{LTS}_A{(#1)}} 
\newcommand{\yltscT}[1]{\mathscr{LTS}_T{(#1)}} 
\newcommand{\yltscTd}[1]{\mathscr{LTS}_T^\yde{(#1)}} 
\newcommand{\yltscTR}[1]{\mathscr{LTSR}_T{(#1)}} 
\newcommand{\yltscTRd}[1]{\mathscr{LTSR}_T^\yde{(#1)}} 
\newcommand{\ytrT}[3]{#1\overset{#2}{\rightarrow}_T#3} 
\newcommand{\ytrA}[3]{#1\overset{#2}{\rightarrow}_A#3} 
\newcommand{\ytrtT}[3]{#1\overset{#2}{\Rightarrow}_T#3} 
\newcommand{\ytrtA}[3]{#1\overset{#2}{\Rightarrow}_A#3} 
\newcommand{\yiocA}[2]{\yltsn{IO}_A(L_#1,L_#2)}  
\newcommand{\yiocAd}[2]{\yltsn{IO}_A^\yde(L_#1,L_#2)}  
\newcommand{\yiocALUpd}[2]{\yltsn{IO}_A(L_#1,L^{+\yde}_#2)}  
\newcommand{\yiocALUpdLImd}[2]{\yltsn{IO}_A(L^{+\yde}_#1,L^{-\yde}_#2)}  %
\newcommand{\yiocALUmdLIpd}[2]{\yltsn{IO}_A(L^{-\yde}_#1,L^{+\yde}_#2)}  %
\newcommand{\yiocAdLUpd}[2]{\yltsn{IO}^\yde_A(L_#1,L^{+\yde}_#2)}  
\newcommand{\yiocT}[1]{\yltsn{IOTS}_T(#1)}  
\newcommand{\yiocTd}[1]{\mathscr{IOTS}_T^\yde{(#1)}} 
\newcommand{\yiocTR}[1]{\yltsn{IOTSR}_T(#1)}  
\newcommand{\yioceA}[2]{\yltsn{IOE}_A(L_{#1},L_{#2})}  
\newcommand{\yLImd}{L_I^{-\yde}} 
\newcommand{\yLUpd}{L_U^{+\yde}} 
\newcommand{\yLUmd}{L_U^{-\yde}} 
\newcommand{\yiocoA}{\text{\,\,\bf ioco}_A\,\,} 
\newcommand{\yiocoT}{\text{\,\,\bf ioco}_T\,\,} 

\newcommand{\yafterm}[1]{\text{\,\,\bf after${}^{#1}$}\,\,} 

\newcommand{\yLIptmd}{L_I^{+\yte,-\yde}} 
\newcommand{\yLUptmd}{L_U^{+\yte,-\yde}} 

\usepackage{fullpage}

\usepackage{color}

\begin{document}

\title{A conformance relation and complete test suites for I/O  systems\protect\thanks{In collaboration with Computing Institute - University of Campinas.}}


\author{Adilson Luiz Bonifacio\thanks{Computing Department, University of Londrina, Londrina, Brazil.} \and
	Arnaldo Vieira Moura\thanks{Computing Institute, University of Campinas, Campinas, Brazil}}

\date{} 

\maketitle

\begin{abstract}
	Model based testing is a well-established approach to verify implementations modeled by I/O labeled transition systems (IOLTSs). 
	One of the challenges stemming from model based testing is the conformance checking and the generation of test suites, specially when completeness is a required property. 
	In order to check whether an implementation under test is in compliance with its respective specification  one resorts to some form of conformance relation that
	guarantees the expected behavior of the implementations, given the behavior of
	the specification. 
	The {\bf ioco} conformance relation is an example of such a relation, specially suited for  asynchronous models. 
	In this work we study a more general conformance relation,  show how to generate finite and complete test suites, and discuss the complexity of the test generation mechanism under this more general conformance relation.   
	We also show that {\bf ioco} conformance is a special case of this new conformance relation, and we  investigate the complexity of classical {\bf ioco}-complete test suites.
	Further, we relate our contributions to more recent works, accommodating the restrictions of their classes of fault models within our more general approach as special cases, and expose the complexity of generating any complete test suite that must satisfy their restrictions.
\end{abstract}

\section{Introduction}

Software testing has been an important part in system development processes, with the goal of improving the quality of the final products.
The systematic use of  formal methods and techniques has taken prominence  when accuracy and critical guarantees are of paramount importance to the development process, such as when failures can cause severe damages.
Testing approaches based on formal models have the added advantage that test suite generation, with proven completeness guarantees, can be effectively and precisely automated, for certain classes of specification models.

Some formalisms, such as  Finite State Machines (FSMs)~\cite{DBLP:journals/ijfcs/bonifacio2012,cardelloliver00conformance,dorofeeva05,gonenc70,hierons06,Krichen04BBCT}, capture some aspects of  systems behaviors. 
However, in FSM models, input actions from the environment and output actions produced by the implementations under test are strongly related, and must occur synchronously.  
This may limit the designers' ability to model more complex asynchronous behaviors. 

This work studies aspects of  more powerful  formalisms  where the exchange of input and output stimuli can occur asynchronously, namely, the class of  Input/Output Labeled Transition Systems (IOLTSs)~\cite{tret-formal-1992,tre99a}.
In this context, model based testing~\cite{Gargantini05CTES,sidhu89,Tretmans96TGIO} has been widely used as a formal framework to verify whether an implementation  under test (IUT) is in conformance to desired behaviors of a given specification, according to a given fault model  and a given conformance relation~\cite{tret-formal-1992,tre99a,ananbcc-orchestrated-2013,testdiscrete,testcritical}.
In this framework, in particular, the {\bf ioco} conformance has been proposed as a suitable conformance relation for testing IOLTS  models~\cite{tret-testing-2004,tret-model-2008}.
We propose a  more general notion of conformance relation for testing  IOLTSs models.
Under this more general notion of conformance, it is possible to accommodate  much wider classes of IOLTS models, thus removing some of the structural restrictions imposed by other approaches. 
The main contributions of our work can be summarized as follows:
\begin{itemize}
	\item[---] The more general notion of conformance allows for the specification of arbitrary desired  behaviors that an IUT must  comply to, as well as undesired behaviors that must not be produced by the IUT.
	When these specifications are regular languages, we show that there is a regular language that can be used as a complete test suite to detect the presence 
	of all desired behaviors and the absence of any undesired behaviors.
	
	\item[---] In a ``white-box'' testing scenario, when one has access to the internal structure of the implementations, we prove the correctness of a polynomial time algorithm that can be used for checking conformance in this more general setting.
	If the specification model is fixed, then the algorithm runs in linear time in the number of states in the implementation. 
	
	\item[---] In a ``black-box'' testing environment, when the tester does not have access to the internal structure of the implementations, we show that the classical {\bf ioco} conformance relation~\cite{tret-model-2008} is a special case of this new conformance relation. 
	We also show how to generate complete test suites that can be used to verify {\bf ioco}-conformance under the same set of fault models considered by Tretmans~\cite{tret-model-2008}, and for implementations with any number of states, independently of the number of states in the specification. 
	Further, in this setting, we prove that $1.61^m$ is an asymptotic worst case lower bound on the size of any {\bf ioco}-complete test suite, where $m$ is the number of states of the largest implementation to be put under test.
	We also show that our approach attains such a lower bound.
	
	\item[---] In a more recent work, Sim\~ao and Petrenko~\cite{simap-generating-2014} discussed how to generate complete test suites for some classes of restricted IOLTSs.
	We prove that our method can construct a finite and complete test suite for such classes of models, when the same restrictions must be satisfied. 
	Also, in that work the complexity of the generated test suites was not studied.
	Here we prove an asymptotic exponential worst case lower bound for any {\bf ioco}-complete test suite that must satisfy  the same  restrictions as in~\cite{simap-generating-2014} and, further,  we prove that our method attains this lower bound.
\end{itemize}

We briefly comment on works that are more closely related to our study.  
See Section~\ref{related} for a more expanded view.
Tretmans~\cite{tret-model-2008} proposed the \textbf{ioco}-conformance relation for IOLTS models, and developed the foundations of an {\bf ioco}-based testing theory, where IUTs are treated as ``black-boxes''. 
In this testing architecture  the tester is seen as an artificial environment that drives the exchange of input and output symbols with the IUT during test runs. 
Some restrictions must be observed by the specification, implementation and tester models, such as input-completeness and output-determinism.
Sim\~ao and Petrenko~\cite{simap-generating-2014}, in a more recent work, also described an approach to generate finite complete test suites for IOLTSs. 
They, however, also imposed a number of restrictions upon the specification and the implementation models in order to obtain finite complete test suites. 
They assumed test purposes to be single-input and also output-complete. 
Moreover, specifications and implementations must be input-complete, progressive, and initially-connected, so further restricting the class of IOLTS models that can be tested according to their fault model. 
Here, we remove many of such restrictions.

This work is organized as follows.
Section~\ref{sec:models} establishes notations and preliminary results. 
Section~\ref{sec:conformance} defines a new notion of conformance relation.
A method for generating complete test suites for verifying adherence to the new conformance relation, and its complexity, is described in Section~\ref{sec:suites}. 
Section~\ref{sec:tretma-suites} visits the classical {\bf ioco}-conformance relation  and establishes  exponential lower bounds on the size of {\bf ioco}-complete test suites.
Section~\ref{sec:adepetre-suites} looks at another class of IOLTS models and show how to obtain complete test suites for this class, and also discusses performance issues when working with models of this class. 
Section~\ref{related} discusses some related works, and  
Section~\ref{sec:conclusion} offers some concluding remarks.

\section{Preliminaries and Notation}\label{sec:models}

In this section we define  Labeled Transition Systems (LTSs) and Input/Output Labeled Transition Systems (IOLTSs). 
For completeness, we also include standard definitions and properties of regular languages and finite state automata (FSA).
Some preliminary results associating LTSs and  FSA are given. 
We start with the formal models and some notation.

\subsection{Basic Notation}\label{subsec:notation}

Let $X$ and $Y$ be sets. We indicate by $\ypow{X}=\{Z\yst Z\ysse X\}$ the power set of $X$, and $X-Y=\{z\yst \text{$z\in X$ and $z\not\in Y$}\}$ indicates set difference.
We will let $X_Y=X\cup Y$. 
When no confusion can arise, we may write $X_y$ instead of $X_{\{y\}}$.
If $X$ is a finite set, the size of $X$ will be indicated by $|X|$.

An alphabet is any non-empty set of symbols.
Let $A$ be an alphabet.
A word over $A$ is any finite sequence $\ysi=x_1\ldots x_n$ of symbols in $A$, that is, $n\geq 0$ and $x_i\in A$, for  $i=1,2,\ldots, n$.
When $n=0$ we have the empty sequence, also indicated by $\yeps$.
The set of all finite sequences, or words, over $A$ is denoted by $A^{\star}$.
When we write $x_1x_2\ldots x_n\in A^{\star}$, it is implicitly assumed that $n\geq 0$ and that $x_i\in A$,
$1\leq i\leq n$, unless explicitly noted otherwise.
The length of a word $\alpha$ over $A$ is indicated by $|\alpha|$.
Hence, $|\yeps|=0$.
Let $\ysi=\ysi_1\ldots \ysi_n$ and $\rho=\rho_1\ldots \rho_m $ be words over $A$. The concatenation of $\ysi$ and $\rho$, indicated by $\ysi\rho$, is the word $\ysi_1\ldots\ysi_n\rho_1\ldots\rho_m$.
Clearly, $|\ysi\rho|=|\ysi|+|\rho|$.
A language $G$ over $A$ is any set $G\ysse A^\star$.
Let $G_1$, $G_2\ysse A^\star$ be languages over $A$. Their product, indicated by  $G_1G_2$, is the language 
$\{\ysi\rho\yst \ysi\in G_1, \rho\in G_2\}$. 
If $G\ysse A^\star$, then its complement is the language $\ycomp{G}=A^\star-G$.

\begin{defi}\label{def:morph}
	Let $A$, $B$ be alphabets.
	A \emph{homomorphism}, or just a \emph{morphism}, from $A$ to $B$ is any function $h:A\rightarrow B^\star$.
\end{defi}
A morphism $h:A\rightarrow B^\star$ can be extended to a function 
$\widehat{h}:A^\star\rightarrow B^\star$ where 
$\widehat{h}(\yeps) = \yeps$, and
$\widehat{h}(a\ysi) =  h(a)\widehat{h}(\ysi)$ when $a\in A$.
We can further lift $\widehat{h}$ to a function $\widetilde{h}: \ypow{A^\star}\rightarrow \ypow{B^\star}$, by letting $\widetilde{h}(G)=\cup_{\ysi\in G}\widehat{h}(\ysi)$,
for all $G\ysse A^\star$. 
To avoid cluttering the notation, we often write $h$ in place of $\widehat{h}$, or of $\widetilde{h}$, when no confusion can arise.
When $a\in A$ is any symbol, we define the simple morphism $h_a:A\rightarrow (A-\{a\})^\star$ by letting $h_a(a)=\yeps$, and $h_a(x)=x$ when $x\neq a$.  
So,  $h_a(\ysi)$ erases all occurrences of  $a$ in  $\ysi$.

The next simple result will be useful later. 
\begin{lemm}\label{lemm:finite}
	Let $A_1$, $A_2$, $B\ysse A^\star$ be languages over $A$.
	Then, there is a finite language $C\ysse B$ such that  
	$A_2\cap B\ysse A_1$ if and only if $A_2\cap C\ysse A_1$. 
\end{lemm}
\begin{proof}
	If $A_2\cap B=\yemp$, choose $C=\yemp$.
	Else, let $x\in A_2\cap B$ and choose $C=\{x\}$. 
\end{proof}

\subsection{Labeled Transition Systems}\label{subsec:lts}

A Labeled Transition System (LTS) is a formal model that makes it convenient to express asynchronous exchange of messages between participating entities, in the sense that outputs do not have to occur synchronously with inputs, but are generated as separated events.
It consists of a set of states and a transition relation between states. Each transition is guarded by an action symbol.  
We first give the finite syntactic description of LTSs. A semantic structure that attributes meaning to a syntactic description will be given shortly in the sequel.
\begin{defi}\label{def:lts}
	A \emph{Labeled Transition System}  over $L$ is a tuple $\yltsnS=\yltsS$, where:
	\begin{enumerate}
		\item $S$ is a finite set of \emph{states} or \emph{locations};
		\item $s_0$ is the \emph{initial state}, or \emph{initial location};
		\item $L$ is a finite set of \emph{labels}, or \emph{actions}, and $\tau\notin L$, is the \emph{internal action symbol}; 
		\item $T\ysse S\times L_\tau \times S$ is a set of \emph{transitions}. 
	\end{enumerate}
\end{defi}
The  symbol $\tau$ is used to model any actions that are not exchanged as messages, that is,  actions that cause only an internal change of states in the model.  
The class of all LTSs over an alphabet $L$ will be denoted by $\yltsc{L}$.

\begin{exam}
	Figure~\ref{fig:lts1} represents a LTS $\yltsNS$ where the set of states is $S=\{s_0,s_1,s_2,s_3,s_4\}$ and the set of labels is $\{b,c,t\}$.
\begin{figure}[htb]
	\center
	\begin{tikzpicture}[ font=\sffamily,node distance=1cm,inner sep=0pt,text width=1.0em,auto,scale=0.75,transform shape]
	\node[initial below, initial text={},  punkt] (s0) {$ s_0$};
	\node[punkt,  below right=1.5cm and 1.5cm  of s0] (s1) { $s_1$};
	\node[punkt, below left=1.5cm and 1.5cm of s0] (s2) { $s_2$};
	\path (s0)    edge [pil]   	node[anchor=right,below]{ $b$} (s1);
	\path (s0)    edge [pil]   	node[anchor=north,above]{ $b$} (s2);
	
	\path (s1)    edge [loop right] node   { $b$} (s1);
	\path (s2)    edge [loop left] node   { $b$} (s2);
	
	\node[punkt, right=1.5cm of s0] (tea) { $s_3$};
	\node[punkt,left=1.5cm of s0] (cof) { $s_4$};
	\path (s1)    edge [pil]   	node[anchor=east]{ $t$} (tea);
	\path (s2)    edge [pil]   	node[anchor=north,left]{ $c$} (cof);
	
	\path (cof)    edge [pil]   	node[anchor=south, above]{ $\tau$} (s0);
	\path (tea)    edge [pil]   	node[anchor=south, above]{ $\tau$} (s0);
	\end{tikzpicture}
	\caption{A LTS with $5$ states and $8$ transitions.}\label{fig:lts1}
\end{figure}
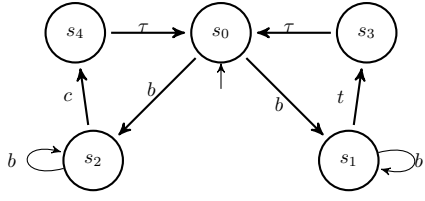
	The set of transitions is given by the arrows, so  $(s_0,b,s_1)$ and $(s_3,\tau,s_0)$ are  transitions.
	Then 
	$$T=\{(s_0,b,s_1),(s_0,b,s_2),(s_1,b,s_1),(s_1,t,s_3),$$$$(s_2,b,s_2),(s_2,c,s_4),(s_3,\tau,s_0),(s_4,\tau,s_0)
	\}$$
	is set of all transitions.
	\yfim
\end{exam} 

The semantics of a LTS is given by its traces, or behaviors. 
First, we need the notion of paths in a LTS.
\begin{defi}\label{def:path}
	Let $\yS=\yltsS$ be a LTS and $p,q \in S$. 
	Let $\ysi=\ysi_1,\ldots,\ysi_n\in L_\tau^\star$.
	We say that $\ysi$ is: 
	\begin{enumerate}
		\item a \emph{path} from $p$ to $q$ in $\yltsnS$ if there are states $r_i\in S$, $0\leq i\leq n$ and, additionally, we have $(r_{i-1},\ysi_i,r_i) \in T$, $1\leq i\leq n$, with $r_0=p$ and $r_n=q$; 
		\item an \emph{observable path}  from $p$ to $q$ in $\yltsnS$ if
		there is a path $\mu$ from $p$ to $q$ in $\yltsnS$ such that $\ysi=h_\tau(\mu)$.
	\end{enumerate}
	We say that the paths start at $p$ and end at $q$. 
\end{defi}

So a path $\ysi$ from $p$ to $q$ is just a sequence of symbols that allows one to follow the syntactic description of $\yS$ and move from state $p$ to state $q$.
Note that a path may include internal transitions, that is, transitions over the internal symbol $\tau$.
An observable path arises from any ordinary path from which internal symbols have been erased.   
The idea is that an external observer will not see the internal transitions, as the model moves from state $p$ to state $q$.
It is also clear that an observable path is not necessarily a path. 

\begin{exam}
	Consider the LTS depicted in Figure~\ref{fig:lts1}.
	The following sequences are paths starting at $s_2$: $\yeps$, $b$, $bc\tau$, $c\tau bbbt\tau b$.
	Among others, the following are observable paths starting at $s_2$: $\yeps$, $b$, $bc$, $cbbbtb$. 
	There are observable paths of any length from $s_2$ to $s_1$.
	Note that, starting at $s_2$, the path $c\tau b$ may lead to either $s_1$ or back to $s_2$, that is, a path may  lead to distinct target states.
	\yfim
\end{exam}

If $\ysi$ is a path from $p$ to $q$, this can also be indicated  by writing $\ytr{p}{\ysi}{q}$.
We may also write $\ytr{p}{\ysi}{}$ to indicate that there is some $q\in S$ such that  $\ytr{p}{\ysi}{q}$; likewise, $\ytr{p}{}{q}$ means that there is some $\ysi\in L_\tau^\star$ such that $\ytr{p}{\ysi}{q}$. 
Also $\ytr{p}{}{}$ means $\ytr{p}{\ysi}{q}$ for some $q\in S$ and some $\ysi\in L_\tau^\star$.
When we want to emphasize that the underlying LTS is $\yltsnS$, we write $\ytru{p}{\ysi}{\yltsnS}{q}$.
We will use the symbol  $\ytrut{}{}{}{}$ to denote observable paths and will use the same shorthand notations as those just indicated for the $\ytr{}{}{}$ relation.
Note that $\ytrt{s}{}{p}$ if and only if $\ytr{s}{}{p}$, for all states $s$ and $p$.

Paths starting at a  state $p$ are the traces of $p$.
The semantics of a LTS is related to traces starting at the initial state.
\begin{defi}\label{def:trace}
	Let $\yltsnS=\yltsS$ be a LTS and let $p\in S$. 
	The set of \emph{traces} of $p$ is   $tr(p)=\{\ysi\yst \ytr{p}{\ysi\,\,}{}\}$, and
	the set of \emph{observable traces} of $p$ is $otr(p)= \{\ysi\yst \ytrt{p}{\ysi\,\,}{}\}$.
	The \emph{semantics} of $\yltsnS$ is the set $tr(s_0)$, and the 
	\emph{observable semantics} of $\yltsnS$ is the set $otr(s_0)$.
\end{defi}
We may write $tr(\yS)$ for $tr(s_0)$ and  $otr(\yS)$ for $otr(s_0)$. 

\begin{rema}\label{rema:tau}
	If $\yS$ has no $\tau$-labeled transitions, then $otr(\yS)=tr(\yS)$.
\end{rema}

\begin{exam}
	Consider the LTS depicted in Figure~\ref{fig:lts1}.
	The semantics of $\yltsnS$ includes such sequences as: $\yeps$, $bbc$, $bt\tau bc\tau bbb$, among others.
	The observable semantics of $\yltsnS$ includes: $\yeps$, $bcbt$, $bbbtbcbtb$.
	\yfim
\end{exam}

A LTS $\yS$ where the set of states $S$ is infinite is better  modeled using internal variables that can take values in infinite domains. 
When the set of labels $L$ is infinite, it is better to use parameters that can exchange values from infinite domains with the environment. 
In Definition~\ref{def:lts} we will always assume that $S$ and $L$ are finite sets.
We can also restrict the syntactic  of LTS models somewhat,  without loosing any descriptive capability.
First,  we can assume that there are no $\tau$-moves that do not change states.
Further, we can do away with states that are not reachable from the initial state.

\begin{rema}\label{rema:lte-finite}
	In Definition~\ref{def:lts} we will always assume that $(s,\tau, s)\not\in T$ and that $\ytr{s_0}{}{s}$ holds, for any $s\in S$.
\end{rema}

The intended interpretation for  $\tau$-moves is that the LTS can autonomously move along such transitions, without consuming any observable labels, that is, 
from an  external perspective, a $\tau$ label may induce ``implicit nondeterminism''. 
On the other hand, 
internal actions can facilitate the specification of formal models.
For example, in a real specification saying that ``after delivering money, the ATM returns to the initial state'', if this behavior does not require exchanging messages with a user, then it can be specified by a $\tau$-move back to the initial state. 
Also note that Remark~\ref{rema:lte-finite} does not prevent the occurrence of cyclic subgraphs  whose edges are all labeled by $\tau$, a situation that would correspond to a livelock. 
In some situations, however, we may not want such behaviors, or  we might want that no observable behavior leads to two distinct states.
This motivates a special variant of LTSs.

\begin{defi}\label{def:lts-deterministic}
	We say that a LTS $\yS=\yltsS$ is \emph{deterministic} if $\ytrt{s_0}{\ysi}{s_1}$ and $\ytrt{s_0}{\ysi}{s_2}$ imply
	$s_1=s_2$, for all $s_1$, $s_2\in S$, and all $\ysi\in L^\star$.
\end{defi}

As a consequence, we have the following result.
\begin{prop}\label{prop:lts-deterministic}
	Let $\yS=\yltsS$ be a deterministic LTS.
	Then $\yS$ has no $\tau$-labeled transitions.
\end{prop}
\begin{proof}
	For the sake of contradiction, assume that $(p,\tau,q)\in T$ is a transition. 
	So, $\ytr{p}{\tau}{q}$.
	From Remark~\ref{rema:lte-finite} we get $p\neq q$ and $\ysi\in L^\star$ such that $\ytr{s_0}{\ysi}{p}$.
	Hence, $\ytr{s_0}{\ysi}{p}\ytr{}{\tau}{q}$. Using Definition~\ref{def:trace} we get $\ytrt{s_0}{\mu}{p}$ and $\ytrt{s_0}{\mu}{q}$, where $\mu=h_\tau(\ysi)=h_\tau(\ysi\tau)$. 
	Since $p\neq q$, this
	contradicts Definition~\ref{def:lts-deterministic}.
\end{proof}

\subsection{Finite State Automata}\label{subsec:fsa}

Any LTS  induces a finite automaton
with a number of simple properties, and for which there are simple decision algorithms. 
In particular, $\tau$-labeled transitions in the LTS will correspond to nondeterminism induced by $\yeps$-moves in the finite automaton. 
For completeness, we give here the basic definitions.
\begin{defi}\label{def:fsa}
	A \emph{Finite State Automaton} (FSA) over $A$ is a tuple $\yfsanA=\yfsaA$, where:
	\begin{enumerate}
		\item $S$ is a finite set of \emph{states};
		\item $s_0\in S$ is the \emph{initial state};
		\item  $A$ is a finite non-empty \emph{alphabet};
		\item  $\rho\ysse S\times (A\cup\{\yeps\}) \times S$ is the \emph{transition relation}; and
		\item $F\ysse S$ is the set of \emph{final states}. 
	\end{enumerate}
	A transition $(p,\yeps,q)\in \rho$ is called an \emph{$\yeps$-move} of $\yA$.
\end{defi}

The semantics of a FSA is given by the language it accepts.
\begin{defi}\label{def:fsa-move}\label{def:fsalang}
	Let $\yA=\yfsaA$ be a FSA. Inductively, define the relation $\yatrt{}{}{}$ as
	\begin{itemize}
		\item[---] $\yatrt{p}{\yeps}{p}$, and 
		\item[---] $\yatrt{p}{x\mu}{q}$ if $(p,x,r)\in\rho$ and $\yatrt{r}{\mu}{q}$ with $x\in\ySi\cup\{\yeps\}$.
	\end{itemize}
	The \emph{language} accepted by  $\yA$ is the set 
	$L(\yfsanA)=\{ \ysi \yst \yatrt{s_0}{\ysi}{p}\,\,\text{and $p\in F$}\}$.
	A language $G\ysse A^\star$ is \emph{regular} if there is a FSA $\yA$  such that $L(\yA)=G$.
\end{defi}

It is a well established fact that the class of regular languages enjoy the following properties.
\begin{prop}\label{prop:reg-closure}
	Let $h:A\rightarrow B^\star$ be a morphism, and let 
	$G$, $H\ysse A^\star$, and $L\ysse B^\star$ be regular languages.
	Then, the following are also regular languages: 
	$G H$, $G\cap H$,  $G\cup H$, $\ycomp{G}$, and $h(G)\ysse B^\star$.
	Further, given FSAs for the original languages we can  effectively construct a FSA for any of the languages indicated by these operations.
\end{prop}
\begin{proof}
	See \emph{e.g.}, \cite{hari-introduction-1978,hopcu-introduction-1979}.
\end{proof}

With respect to the operations of union and intersection we can put simple upper bounds in the number of states of the resulting FSA, as can be seen from the constructions in~\cite{hari-introduction-1978,hopcu-introduction-1979}. 
\begin{rema}\label{rem:bounds}
	Let $\yA$ and $\yB$ be  FSAs with $n_\yA$ and $n_\yB$ states, respectively.
	Then we can construct a  FSA $\yC$ with at most $n_\yA n_\yB$ states and such that $L(\yC)=L(\yA)\cap L(\yB)$.
	Likewise, from $\yA$ and $\yB$  we can  also construct  FSA $\yC$ with $n_\yA + n_\yB+1$ states and such that $L(\yC)=L(\yA)\cup L(\yB)$.
\end{rema} 

A deterministic FSA has no $\yeps$-moves and the transition relation reduces to a function. 
Completeness, in this case, means that from any state we can always find a transition on any symbol.
\begin{defi}\label{def:fsa-determinism}
	A FSA $\yA=\yfsaA$ is \emph{deterministic} if $\rho$ is a partial function from $S\times A$ into $S$, and $\yA$ is \emph{complete} if  for all $s\in S$ and all $a\in A$ there is some $p\in S$ such that $(s,a,p)\in\rho$.
\end{defi}

So,  a FSA $\yA$ is deterministic if it has no $\yeps$-moves and $(s,x,p_1)$ and $(s,x,p_2)$ in $\rho$ forces $p_1=p_2$.
Given any FSA $\yA$ we know how to construct an equivalent deterministic FSA $\yB$  using 
the subset construction ~\cite{hopcu-introduction-1979,rabis-finite-1959}.

\begin{prop}\label{prop:no-eps}
	For any FSA $\yA=\yfsaA$ there is a deterministic FSA $\yB$ with $L(\yA)=L(\yB)$.
	Moreover, if $F=S$, then all states in $\yB$ are final states.
	Further, $\yB$ can  be algorithmically constructed from $\yA$. 
\end{prop}
\begin{proof}
	First, we eliminate all $\yeps$-moves from $\yA$ obtaining an equivalent FSA $\yC$~\cite{hari-introduction-1978,hopcu-introduction-1979}.
	As an easy observation from those constructions, we can see that if $F=S$ then all states in $\yC$ are also final states.
	Next, given $\yC$, use the standard subset construction~\cite{rabis-finite-1959} to obtain a deterministic FSA $\yB$ with no $\yeps$-moves  equivalent to $\yA$.
	Here also, since all states in $\yC$ are final states, the construction forces all states in $\yB$ to be final states.
	Clearly, $L(\yA)=L(\yC)=L(\yB)$ as desired.   
\end{proof}

From a deterministic FSA we can easily get an equivalent complete FSA with at most one extra state.
We just add a transition to that extra state in the new FSA whenever a transition is missing in the original FSA.
\begin{prop}\label{prop:fsa-complete}
	Let $\yA=\yfsa{S_\yA}{s_0}{A}{\rho_\yA}{F}$ be a deterministic FSA. 
	We can effectively construct a complete FSA $\yB=\yfsa{S_\yB}{s_0}{A}{\rho_\yB}{F}$ such that 
	$L(\yB)=L(\yA)$ and with $\yst S_\yB|=|S_\yA|+1$ states and $\vert A\vert(\vert S_\yA\vert+1)$ transitions.
\end{prop}
\begin{proof}
	Let $S_\yB=S_A\cup \{e\}$, where $e\not\in S_\yA$, and
	$$\rho_\yB\,=\,\rho_\yA\cup\big\{(e,\ell,e)\yst \text{for any } \ell\in A \big\} \cup\big\{(s,\ell,e)\yst \text{if $\ell\in A$  and [  for all  $p\in S_A$ we have $(s,\ell,p)\not\in\, \rho_\yA$} ] \big\}.$$
	It is clear that $\yB$ is a complete FSA.
	A simple induction on $|\ysi|\geq 0$ shows that  $\ysi\in L(\yA)$ if and only if  
	$\ysi\in L(\yB)$, for all $\ysi\in A^\star$.
	Since $\yB$ is complete and has $\vert S_\yA\vert +1$ states, then is must have $\vert A\vert(\vert S_\yA\vert+1)$ transitions.
\end{proof}

We can now convert $\tau$-moves of a LTS into $\yeps$-moves of a FSA, and make any location in the LTS a final state of the  FSA, and vice-versa.
This association will allow us to develop algorithmic constructions involving LTSs by making use of known  efficient algorithmic constructions involving the associated FSAs.
Observe that in the induced FSAs all states are final.    
\begin{defi}\label{def:lts-fsa}
	We have the following two associations:
	\begin{enumerate}
		\item Let $\yltsnS=\yltsS$ be a LTS. 
		The FSA \emph{induced} by $\yltsnS$ is $\yA_\yS=\yfsa{S}{s_0}{L}{\rho}{S}$ where,
		for all $p$, $q\in S$ and all $\ell\in L$:\\
		\hspace*{5ex}
		$(p,\ell,q)\in\rho\,\, \quad \text{if and only if}\quad (p,\ell,q)\in T,$\\
		\hspace*{5ex}
		$(p,\yeps,q)\in\rho\,\,\quad \text{if and only if}\quad (p,\tau,q)\in T.$ 
		\item Let $\yA=\yfsa{S}{s_0}{A}{\rho}{S}$ be a FSA.
		The LTS \emph{induced} by $\yA$ is $\yS_\yA=\ylts{S}{s_0}{A}{T}$ where,
		for all $p$, $q\in S$ and all $a\in A$:\\
		\hspace*{5ex}
		$(p,a,q)\in T\quad \text{if and only if}\quad (p,a,q)\in\rho,$\\
		\hspace*{5ex}
		$(p,\tau,q)\in T\quad \text{if and only if}\quad (p,\yeps,q)\in\rho.$
	\end{enumerate}
\end{defi}

The observable semantics of $\mathcal{S}$ is just the language accepted by $A_\mathcal{S}$.
We note this as the next proposition.

\begin{prop}\label{prop:lts-fsa}
	Let $\yS=\yltsS$ be a LTS  with $\yfsanAs{\yltsnS}$ the FSA induced by $\yltsnS$, and
	let $s$, $q\in S$. Then we have:
	\begin{enumerate}
		\item If $\ysi\in L_\tau^\star$ and $\ytr{s}{\ysi}{q}$ in $\yS$, then  $\yatrt{s}{\mu}{q}$ in $\yA$, where $h_\tau(\ysi)=\mu$.
		\item If $\mu\in L^\star$ and $\yatrt{s}{\mu}{q}$ in $\yA$, then there is some  $\ysi\in L_\tau^\star$ such that $\ytr{s}{\ysi}{q}$ in $\yS$ and $\mu=h_\tau(\ysi)$.
	\end{enumerate}
	In particular, $otr(\yltsnS)=L(\yfsanAs{\yltsnS})$.
	Likewise when $\yA$ is a FSA and $\yS_\yA$ is the LTS induced by $\yA$.
\end{prop}
\begin{proof}
	An easy induction on the length of $\ysi\in L_\tau^\star$  shows that (1) holds, and a similar induction on the length of $\mu$ establishes (2).
	By Definition~\ref{def:trace}, we know that $\mu\in otr(\yS)$ if and only if there is some $\ysi\in tr(\yS)$ such that $\mu=h_\tau(\ysi)$.
	Then, by item (1), we get $\mu\in L(\yA_\yS)$, because all states in $\yA_\yS$ are final states.
	For the converse, if $\mu\in\yA_\yS$ then, by item (2),  we get some $\ysi\in tr(\yS)$ with $h_\tau(\ysi)=\mu$. Definition~\ref{def:trace} then gives $\mu\in otr(\yS)$.
	Hence,  $L(\yA_\yS)= otr(\yS)$.
	
	The case when $\yS_\yA$ is the LTS induced by a FSA $\yA$ follows by a similar reasoning.
\end{proof}

Proposition~\ref{prop:lts-fsa} says that  $otr(\yS)$ is  regular. This implies that  we can construct a FSA $\yA$ with $L(\yA)=otr(\yS)$, and from $\yA$ we can then obtain a deterministic LTS equivalent to $\yS$. 
\begin{prop}\label{prop:deterministic-lts}
	Let $\yS$ be a LTS. 
	Then we can effectively construct a deterministic LTS $\yT$  such that $otr(\yS)=tr(\yT)=otr(\yT)$.
\end{prop}
\begin{proof}
	Consider the FSA $\yA_\yS$ induced by $\yS$.
	From Proposition~\ref{prop:lts-fsa} we know that $L(\yA_\yS)=otr(\yS)$.
	Use Proposition~\ref{prop:no-eps} to get a deterministic 
	FSA $\yB$ such that $L(\yfsanAs{\yltsnS})=L(\yB)$.
	Hence $otr(\yS)=L(\yB)$.
	From Definition~\ref{def:lts-fsa} we know that all states of $\yA_\yS$ are final states and so, from
	Proposition~\ref{prop:no-eps} it follows that the same is true of $\yB$.
	Thus, from Definition~\ref{def:lts-fsa} again, we can get the LTS $\yS_{\yB}$, induced by $\yB$.
	By  Proposition~\ref{prop:lts-fsa} again we now have $L(\yB) =otr(\yS_{\yB})$, and so $otr(\yS)=otr(\yS_\yB)$.
	It is also clear from Definition~\ref{def:lts-fsa} that $\yS_\yB$ has no $\tau$-labeled transitions because $\yB$ has no $\yeps$-moves.
	So, Remark~\ref{rema:tau} says that $tr(\yS_\yB)=otr(\yS_\yB)$, and 
	we have  $otr(\yS)=tr(\yS_\yB)=otr(\yS_\yB)$.
	We now argue that $\yS_\yB$ is deterministic.
	We may write $\yB=\yfsa{S}{s_0}{A}{\rho}{S}$ and $\yS_\yB=\ylts{S}{s_0}{A}{T}$.
	For the sake of contradiction, assume that $\yS_\yB$ is not deterministic.
	Then, since $tr(\yS_\yB)=otr(\yS_\yB)$, Definition~\ref{def:lts-deterministic} gives some $\ysi\in A^+$, and some $p_1$, $p_2\in S$ such that $\ytr{s_0}{\ysi}{p_1}$ and $\ytr{s_0}{\ysi}{p_1}$, with $p_1\neq p_2$.
	Let $\vert \ysi\vert$ be minimum. Then we have $\ysi=\mu x$ with $x\in A$ and $q\in S$ such that  $\ytr{s_0}{\mu}{q}\ytr{}{x}p_1$ and $\ytr{s_0}{\mu}{q}\ytr{}{x}p_2$.
	Now, Definition~\ref{def:lts-fsa} says that $(q,x,p_1)\in \rho$ and  $(q,x,p_2)\in \rho$ are moves of $\yB$, contradicting Definition~\ref{def:fsa-determinism} and the fact that $\yB$ is deterministic. 
	
	Since the construction at Proposition~\ref{prop:no-eps} is effective, we see that $\yS_\yB$ can also be effectively constructed from $\yS$.
\end{proof} 
Of course, obtaining a deterministic FSA from a nondeterministic one can be exponentially expensive in terms os the number of states.
Proposition~\ref{prop:deterministic-lts} just says how one could go about obtaining deterministic LTSs from other nondeterministic models.

\subsection{Input Output Labeled Transition Systems}\label{subsec:iolts}

In many situations, we wish to treat some action labels as symbols that the LTS ``receives'' from the environment as input symbols, and another complementary set of action labels as symbols that the LTS ``sends back'' to the environment as output symbols.
The LTS variation that differentiates between such input  and output action symbols is called an 
Input/Output  Labeled Transition System (IOLTS).
These situations motivate the next definitions.

\begin{defi}\label{def:iolts}
	An Input/Output Labeled Transition System (IOLTS) is a tuple $\yI=\yioltsI$, where: 
	\begin{enumerate}
		\item $L_I$ is a finite set of \emph{input actions}; 
		\item $L_U$ is a finite set of \emph{output actions}; 
		\item $L_I\cap L_U=\emptyset$, and $L=L_I\cup L_U$ is the set of \emph{actions}; 
		\item $\yltsS$ is the  \emph{underlying LTS} associated to $\yI$.
	\end{enumerate}
\end{defi}
We indicate the underlying LTS associated to  $\yI$ by $\yS_{\yI}$, and
we denote the class of all IOLTSs with input alphabet $L_I$ and output alphabet $L_U$ by $\yioc{I}{U}$. 
Other works, impose additional restrictions to the basic IOLTS model, but we do not need to restrict the models at this point. 
In particular, we study some of the more restricted variations that appear in~\cite{tret-model-2008} in Subsection~\ref{subsec:ioco-test-cases}, and we look at some of the more restricted models  that appear in~\cite{simap-generating-2014} in Section~\ref{sec:adepetre-suites}.  

Several notions involving IOLTSs will be defined by a direct reference to their underlying LTSs.
In particular, the semantics of an IOLTS is the set of observable traces of its underlying LTS.  

\begin{defi}\label{def:iolts-semantics}
	Let $\yI\in\yioc{I}{U}$ be an IOLTS. The \emph{semantics} of $\yI$ is the set $otr(\yI)=otr(\yltsnS_\yI)$.
\end{defi}
When $\yI$ is an IOLTS the notation $\underset{\yI}{\rightarrow}$ and $\underset{\yI}{\Rightarrow}$ are to be understood as $\underset{\yS}{\rightarrow}$ and $\underset{\yS}{\Rightarrow}$, respectively, where $\yS$ is the underlying LTS associated to $\yI$. 
IOLTSs generalize the simpler formalism of Meally machines~\cite{gill62},
where communication is synchronous.
In an IOLTS model outputs can be performed separately from inputs, so that I/O is no longer an atomic action, which facilitates the specification of more complex behaviors such as those in reactive systems.

\begin{exam}
	Figure~\ref{fig:iots} represents an IOLTS, adapted from \cite{kric-model-2007}. 
	The model represents the lightening of a bulb.
	It has the two input actions $s$ and $d$ representing a single or a double click on the lamp switch, respectively.
	It has three output labels, $dim!$, $bri!$, and \emph{off!}, informing the environment that the illumination turned to dim, bright or off, respectively. 
	The initial state is $s_0$.
	Following the rightmost circuit, starting at $s_0$, when  the user --- that is, the environment --- hits a single click on the switch the system moves from state $s_0$ to state $sd$ on the input action $s?$. This is represented by the transition $\ytr{s_0}{s?}{sd}$. Then, following the transition $\ytr{sd}{dim!}{dm}$, the system reaches the state $dm$ and outputs the label $dim!$, informing the user that the illumination is now dimmed. At state $dm$ if the user double clicks at the switch, the system responds with \emph{off} and moves back to state $s_0$. This corresponds to the transitions $\ytr{dm}{d?}{dd}$ and $\ytr{dd}{\text{\emph{off!}}}{s_0}$.
	But, if at state $dm$ the user instead clicks only once at the switch then the transitions $\ytr{dm}{s?}{ds}$ and $\ytr{ds}{bri!}{br}$ are traversed, and the system moves to state $br$ issuing the output $bri!$ to signal that  the lamp is now in the bright mode.\yfim
\end{exam} 
\begin{figure}[htb]
	\centering
	\subfigure[An IOLTS.]{\label{fig:iots}
		\begin{tikzpicture}[font=\sffamily,node distance=1cm, auto,scale=0.75,transform shape]
		\node[ punkt] (q0) {$s_0$};
		\node[punkt, inner sep=3pt,below right=0.5cm and 2.5cm  of q0] (sd) {$sd$};
		\node[punkt, inner sep=3pt,below left=0.5cm and 2.5cm of q0] (db) {$db$};
		\path (q0)    edge [ pil]   node[anchor=north,above]{$s?$} (sd);
		\path (q0)    edge [pil]   	node[anchor=north,above]{$d?$} (db);
		\node[punkt, inner sep=3pt,below =1.5cm of sd] (dm) {$dm$};  
		\path (sd)    edge [ pil]  	node[anchor=north,right]{$dim!$} (dm);
		\node[punkt, inner sep=3pt,below =1.5cm of db] (br) {$br$};
		\path (db)    edge [ pil] 	node[anchor=north,left]{$bri!$} (br);
		\node[punkt, inner sep=3pt,above left =1cm and 1cm of dm] (dd) {$dd$};
		\path (dm)    edge [ pil]     node[anchor=north,right]{$d?$} (dd);
		\path (dd)    edge [ pil]     node[anchor=north,right]{$\text{\emph{off!}}$} (q0);
		\node[punkt, inner sep=3pt,above right =1cm and 1cm of br] (bs) {$bs$};
		\path (br)    edge [ pil]     node[anchor=north,right]{$s?$} (bs);
		\path (bs)    edge [ pil]     node[anchor=north,right]{$\text{\emph{off!}}$} (q0);
		
		\node[punkt, inner sep=3pt,below left =0.3cm and 2.5cm of dm] (ds) {$ds$};
		\path (dm)    edge [ pil]     node[anchor=north,below]{$s?$} (ds);
		\path (ds)    edge [ pil]     node[anchor=north,below]{$bri!$} (br);
		\node[punkt, inner sep=3pt,above left =0.3cm and 2.5cm of dm] (bd) {$bd$};
		\path (br)    edge [ pil]     node[anchor=north,below]{$d?$} (bd);
		\path (bd)    edge [ pil]     node[anchor=north,below]{$dim!$} (dm);
		\end{tikzpicture}
	}
	\subfigure[Another IOLTS.]{ \label{fig:iots2}
		\begin{tikzpicture}[font=\sffamily,node distance=1cm, auto,scale=0.75,transform shape]
		\node[ punkt] (s0) {$s_0$};
		\node[punkt, inner sep=3pt,below right=0.5cm and 2cm  of s0] (s1) {$s_1$};
		\node[punkt, inner sep=3pt,below left=0.5cm and 2cm of s0] (s2) {$s_2$};
		\path (s0)    edge [pil]   	node[anchor=north,above]{$but?$} (s1);
		\path (s0)    edge [pil]   	node[anchor=north,above]{$but?$} (s2);
		
		\path (s1)    edge [loop right] node   {$but?$} (s1);
		\path (s2)    edge [loop left] node   {$but?$} (s2);
		
		\node[punkt, inner sep=3pt,below =0.8cm of s1] (tea) {$tea$};
		\node[punkt, inner sep=3pt,below =0.8cm of s2] (cof) {$cof$};
		\path (s1)    edge [pil]   	node[anchor=north,right]{$tea!$} (tea);
		\path (s2)    edge [pil]   	node[anchor=north,left]{$\text{\emph{coffee!}}$} (cof);
		
		\path (cof)    edge [pil,bend right=30]   	node[anchor=north,right]{$\tau$} (s0);
		\path (tea)    edge [pil,bend left=30]   	node[anchor=north,right]{$\tau$} (s0);
		\end{tikzpicture}
	}
	\caption{IOLTS models.} 
\end{figure}
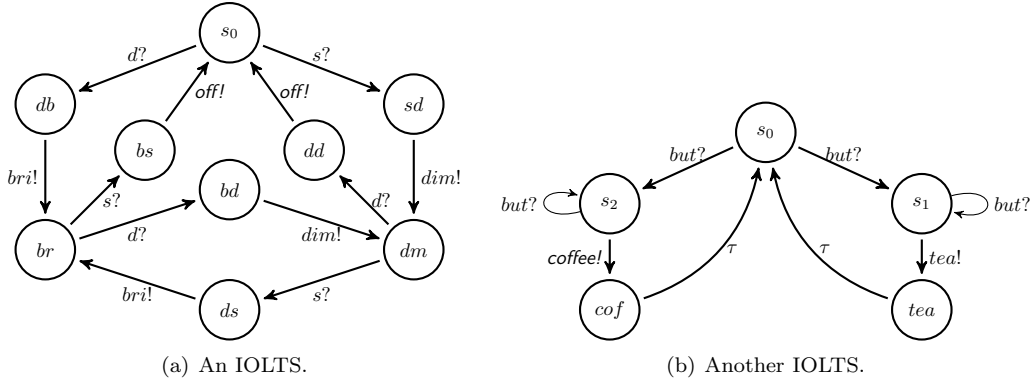
The next example illustrates internal actions.
\begin{exam}
	Figure~\ref{fig:iots2} represents another IOLTS, adapted from \cite{tret-model-2008}.
	It describes a strange coffee machine.
	When the user hits the start button --- represented by the input label $but?$ --- the machine
	chooses to go either to state $s_1$ or to state $s_2$.
	If it goes to state $s_1$, no matter how many extra times the user hits the button the loop labeled $but?$ at state $s_1$ is traversed, keeping the machine at state $s_1$;  then, asynchronously, the machine dispenses a cup of tea, signaled by the output label $tea!$, reaching state $tea$. Next, the machine performs an internal action, signaled by the internal action symbol $\tau$, and  returns to the start state.
	If it chooses to follow down the left branch from the start state the situation is similar, but now the machine will dispense a cup of coffee,  indicated by the output label \emph{coffee!}, and another internal action moves it back to the start state again.
	\yfim
\end{exam}

\section{Conformance testing}\label{sec:conformance}

This section defines a new and generalized notion of conformance relation, allowing for generic sets of desired and undesired behaviors to be checked for asynchronous IOLTS models. 
We study the relationship of this more general notion of conformance to the classical {\bf ioco}-conformance~\cite{tret-model-2008}.  
In particular, we show that the classical {\bf ioco}-conformance relation is a special case of this more general relation.

\begin{rema}\label{rema:quiescent}
	We deal with the notion of quiescent states in detail in the forthcoming Section~\ref{sec:tretma-suites}.
	For now we just note that the treatment of quiescent states will require the addition of a new symbol $\yde$ to the output alphabet, and will also result in the inclusion of some $\yde$-transition to the set of original transitions. 
	But, since the  results obtained in this section are valid for {\emph any} LTS model,  they will remain valid for those variations that treat quiescent states explicitly. 
\end{rema}

\subsection{The General Conformance Relation}\label{subsec:genconf}

Informally, we consider a language $D$, the set of ``desirable'', or ``allowed'', behaviors, and a language $F$, the set of ``forbidden'', or ``undesirable'', behaviors we want to verify of a system.
If we have a specification LTS $\yS$ and an implementation LTS $\yI$ we want to say that $\yI$ \emph{conforms} to $\yS$ according to $(D,F)$ if no undesired behavior in $F$ that is observable in $\yI$  is specified in $\yS$, and all desired behaviors in $D$ that are observable in $\yI$ are specified in $\yS$.
This leads to the following definition of conformance.
\begin{defi}\label{def:conf}
	Let $D, F\ysse L^\star$, 
	$\yS$ and $\yI$  be LTSs over $L$.
	Then \emph{$\yI$ $(D,F)$-conforms to $\yS$}, written $\yI\yconf{D}{F} \yS$, if and only if
	$\ysi\in otr(\yI)\cap F$ gives $\ysi\not\in otr(\yS)$, and 
	$\ysi\in otr(\yI)\cap D$ gives $\ysi\in otr(\yS)$.
\end{defi}

We note an equivalent way of expressing these conditions that may also be useful.
Write $\ycomp{otr}(\yS)$ for the complement of $otr(\yS)$, that is,  $\ycomp{otr}(\yS)=L^\star-otr(\yS)$.

\begin{prop}\label{prop:equiv-conf}
	Let $\yS$ and $\yI$ be LTSs over $L$, and $D, F\ysse L^\star$.
	Then $\yI\yconf{D}{F} \yS$ if and only if
	$$otr(\yI)\cap \big[(D\cap\ycomp{otr}(\yS))\cup(F\cap otr(\yS))\big]=\yemp.$$
\end{prop}
\begin{proof}
	From Definition~\ref{def:conf} we readily get $\yI\yconf{D}{F} \yS$ if and only if  $otr(\yI)\cap F\cap otr(\yS)=\yemp$ and 
	$otr(\yI)\cap D\cap \ycomp{otr}(\yS)=\yemp$.
	And the last two statements  hold if and only if 
	$\yemp=otr(\yI)\cap \big[(D\cap\ycomp{otr}(\yS))\cup(F\cap otr(\yS))\big]$,
	as desired.
\end{proof}


\begin{exam}
	\label{exnew}
	Let $\yS$ be the specification in Figure~\ref{fig:newspec} and $\yI$  the implementation in Figure~\ref{fig:newimpl}.
	Take the languages $D=(a+b)^\star ax$ and $F=ba^\star b$. 
	We  want to check if $\yI\yconf{D}{F} \yS$ holds. 
\begin{figure}[htb]
	\centering
	\subfigure[A LTS specification $\yS$.]{\label{fig:newspec}
		\begin{tikzpicture}[font=\sffamily,node distance=1cm, auto,scale=0.75,transform shape]
		\node[ initial by arrow, initial text={}, punkt] (s0) {$s_0$};
		\node[punkt, inner sep=3pt,right=3cm and 3cm  of s0] (s1) {$s_1$};
		\node[punkt, inner sep=3pt,below =2cm and 1.5cm of s1] (s2) {$s_2$};
		\node[punkt, inner sep=3pt,below =2cm and 1.5cm of s0] (s3) {$s_3$};
		
		\path (s0)    edge [pil]   	node[anchor=north,below]{a} (s1);
		\path (s1)    edge [pil]   	node[anchor=north,right]{b,x} (s2);
		\path (s1)    edge [pil]   	node[anchor=north,right]{a} (s3);
		
		\path (s0)    edge [pil]   	node[anchor=north,right]{b} (s3);
		\path (s3)    edge [loop below] node   {a} (s3);    
		\path (s2)    edge [loop below] node   {b} (s1);
		\path (s2)    edge [pil]   	node[anchor=north,below]{x} (s3);
		\path (s3)    edge [pil,bend left=25]   	node[anchor=north,left]{b} (s0);
		
		\end{tikzpicture} 
	}
	\subfigure[A LTS implementation $\yI$.]{ \label{fig:newimpl}
		\begin{tikzpicture}[font=\sffamily,node distance=1cm, auto,scale=0.75,transform shape]
		\node[ initial by arrow, initial text={}, punkt] (s0) {$q_0$};
		\node[punkt, inner sep=3pt,right=3cm and 3cm  of s0] (s1) {$q_1$};
		\node[punkt, inner sep=3pt,below =2cm and 1.5cm of s1] (s2) {$q_2$};
		\node[punkt, inner sep=3pt,below =2cm and 1.5cm of s0] (s3) {$q_3$};
		
		\path (s0)    edge [pil]   	node[anchor=north,below]{a} (s1);
		\path (s1)    edge [pil]   	node[anchor=north,right]{b,x} (s2);
		\path (s1)    edge [pil]   	node[anchor=north,right]{a} (s3);
		
		\path (s0)    edge [pil]   	node[anchor=north,right]{b} (s3);
		\path (s3)    edge [loop below] node   {a} (s3);    
		\path (s2)    edge [loop below] node   {b} (s1);
		\path (s2)    edge [pil]   	node[anchor=north,below]{a} (s3);
		\path (s3)    edge [pil,bend left=25]   	node[anchor=north,left]{b} (s0);
		
		\end{tikzpicture}
	}
	\caption{LTS models.} 
\end{figure}
	We see that $F\cap otr(\yS)\neq \yemp$ holds, for instance  $baab\in F\cap  otr(\yS)$.
	Since $baab\in otr(\yI)$, we conclude that $\yconf{D}{F} \yS$ does not hold.
	We also see that $\ycomp{otr}(\yS)$ is the language accepted by the FSA $\ycomp{\yS}$, in Figure~\ref{fig:newcompspec}.
	A simple inspection shows that $ababax$ is accepted by $\ycomp{\yS}$, and that $ababax\in D$. 
	Hence  $ababax\in D\cap\ycomp{otr}(\yS)$.  
	From Figure~\ref{fig:newimpl}, we get $\ytrt{q_0}{abab}{q_0}\ytrt{}{ax}{q_2}$, and so $ababax$ is an observable behavior of $\yI$. 
	So,  $otr(\yI) \cap D\cap\ycomp{otr}(\yS) \neq \yemp$ which also means that $\yI\yconf{D}{F} \yS$ does not hold.
	We see that, in this case, the condition specified in either of the regular languages $F$ or $D$ is enough to guarantee non-conformance.
	
	For the sake of the completeness, assume the same models $\yS$ and $\yI$, but now take $F=ab^+ x$ and  $D=aa^+b(bb)^\star ax$. 
	By inspection, for any $\ysi\in ab^{+}$ get $\ytrt{q_0}{\ysi}{q_2}$, so that $\ysi x\not\in otr(\yI)$. 
	Thus, $otr(\yI)\cap F\cap otr(\yS)=\yemp$. 
	Moreover, for any $\ysi\in D$ we get $\ytrt{s_0}{\ysi}{s_2}$ so that $D\ysse otr(\yS)$. 
	Hence, $D\cap\ycomp{otr}(\yS)=\yemp$. 
	Therefore, $otr(\yI)\cap \big[(D\cap\ycomp{otr}(\yS))\cup(F\cap otr(\yS))\big]=\yemp$, and we now see that $\yI\yconf{D}{F} \yS$ holds.
\begin{figure}[htb]
	\centering
	\subfigure[A FSA for $\ycomp{otr}(\yS)$ of Figure~\ref{fig:newspec}.]{\label{fig:newcompspec}
		\begin{tikzpicture}[font=\sffamily,node distance=1cm, auto,scale=0.75,transform shape]
		\node[ initial by arrow, initial text={}, punkt] (s0) {$\overline{s}_0$};
		\node[punkt, inner sep=3pt,right=3cm and 3cm  of s0] (s1) {$\overline{s}_1$};
		\node[punkt, inner sep=3pt,below =2cm and 1.5cm of s1] (s2) {$\overline{s}_2$};
		\node[punkt, inner sep=3pt,below =2cm and 1.5cm of s0] (s3) {$\overline{s}_3$};
		\node[punkt, accepting, inner sep=3pt,below left =1.5cm and 1.2cm of s3] (err) {$err$};

		\path (s0)    edge [pil]   	node[anchor=north,below]{a} (s1);
		\path (s1)    edge [pil]   	node[anchor=north,right]{b,x} (s2);
		\path (s1)    edge [pil]   	node[anchor=north,right]{a} (s3);
		
		\path (s0)    edge [pil]   	node[anchor=north,right]{b} (s3);
		\path (s3)    edge [loop below] node   {a} (s3);    
		\path (s2)    edge [loop below] node   {b} (s1);
		\path (s2)    edge [pil]   	node[anchor=north,below]{x} (s3);
		\path (s3)    edge [pil,bend left=25]   	node[anchor=north,left]{b} (s0);
		
		\path (s0)    edge [pil,bend right=25]   	node[anchor=south]{x} (err);
		\path (s3)    edge [pil,]   	node[anchor=south]{x} (err);
		\path (s2)    edge [pil,bend left=25]   	node[anchor=north]{a} (err);
		
		\path (err)    edge [loop below] node   {a,b,x} (err);
		\end{tikzpicture}
	}
	\subfigure[FSA $\yD$ for $D=otr(\yS)\cdot L_U$.]{ \label{fig:newspecD}
		\begin{tikzpicture}[font=\sffamily,node distance=1cm, auto,scale=0.75,transform shape]
		\node[ initial by arrow, initial text={}, punkt] (s0) {$d_0$};
		\node[punkt, inner sep=3pt,right=3cm and 3cm  of s0] (s1) {$d_1$};
		\node[punkt, inner sep=3pt,below =2cm and 1.5cm of s1] (s2) {$d_2$};
		\node[punkt, inner sep=3pt,below =2cm and 1.5cm of s0] (s3) {$d_3$};
		\node[punkt, accepting, inner sep=3pt,below right =2.5cm and 1cm of s3] (err) {$D$};

		\path (s0)    edge [pil]   	node[anchor=north,below]{a} (s1);
		\path (s1)    edge [pil]   	node[anchor=north,left]{b,x} (s2);
		\path (s1)    edge [pil]   	node[anchor=north,right]{a} (s3);
		
		\path (s0)    edge [pil]   	node[anchor=north,right]{b} (s3);
		\path (s3)    edge [loop below] node   {a} (s3);    
		\path (s2)    edge [loop below] node   {b} (s1);
		\path (s2)    edge [pil]   	node[anchor=north,below]{x} (s3);
		\path (s3)    edge [pil,bend left=25]   	node[anchor=north,left]{b} (s0);
		
		\path (s0)    edge [pil,bend right=60]   	node[anchor=south, left]{x} (err);
		\path (s3)    edge [pil,]   	node[anchor=south]{x} (err);
		\path (s2)    edge [pil]   	node[anchor=north, right]{x} (err);
		\path (s1)    edge [pil, bend left=60]   	node[anchor=north,right]{x} (err);
		
		\end{tikzpicture}
	}
	\caption{FSAs for languages  $\ycomp{otr}(\yS)$ and $D$.} 
\end{figure}

	In Subsection~\ref{subsec:ioco}  we will use this same example to compare the new conformance relation to the traditional  {\bf ioco}-conformance relation~\cite{tret-model-2008}.  \yfim
\end{exam}


%

We note that by varying $D$ and $F$ we can accommodate  several different notions of conformance, as illustrated next. 
This attests to the generality of the $\yconf{D}{F}{\!\!}$ relation.
\begin{enumerate}
	\item  All we want to check is that any observable behavior of $\yI$ must rest specified in $\yS$.
	Then, let $D=L^\star$ and $F=\yemp$.
	We get $\yI\yconf{D}{F} \yS$ if and only if $otr(\yI)\ysse otr(\yS)$.
	\item 
	Let $C$ and $E$ be disjoint subsets of locations of $\yI$. 
	Allowed observable behaviors of $\yI$ are all its observable traces that lead to a location in $C$, and no observable behavior of $\yI$ can lead to a location in $E$. 
	Let $H_C$, $H_E\ysse L^\star$ be the sets of observable behaviors of $\yI$ that end in locations in $C$ and in $E$, respectively.  
	Then, $\yI\yconf{H_C}{H_E} \yS$ if and only if any allowed behavior of $\yI$ is also specified in $\yS$ and no forbidden behavior of $\yI$ is  specified in $\yS$.
	\item Let $H\ysse L$.
	Desirable behaviors of $\yI$ are  its observable traces that end in a label in  $H$, and  undesirable behaviors are of no concern.
	Choose $F=\yemp$ and $D=L^\star H$.
\end{enumerate}
Clearly,  $D$ and $F$ are regular languages in all cases listed above.
More generally, if an implementation $\yI$ conforms to a specification $\yS$ according to a pair of  languages $(D,F)$, then we can assume that $D$ and $F$ are, in fact, finite languages.
\begin{coro}\label{coro:finite}
	Let $\yS$ and $\yI$ be  LTSs over $L$,  and
	let $D$ and $F$ be languages over  $L$.
	Then, there are finite languages $D'\ysse D$ and $F'\ysse F$ such that $\yI\yconf{D}{F} \yS$ if and only if $\yI\yconf{D'}{F'} \yS$.
\end{coro}
\begin{proof}
	By Proposition~\ref{prop:equiv-conf}, $\yI\yconf{D}{F} \yS$ if and only if  
	$otr(\yI)\cap F\ysse \ycomp{otr}(\yS)$ and $otr(\yI)\cap D\ysse otr(\yS)$.
	Lemma~\ref{lemm:finite} gives a
	finite language $D'\ysse D$  such that $otr(\yI)\cap D'\ysse otr(\yS)$ if and only if $otr(\yI)\cap D\ysse otr(\yS)$.
	Likewise, we can get a finite language $F'\ysse F$  such that $otr(\yI)\cap F'\ysse \ycomp{otr}(\yS)$ if and only if $otr(\yI)\cap F\ysse \ycomp{otr}(\yS)$. 
	By Proposition~\ref{prop:equiv-conf} again, this holds if and only if 
	$\yI\yconf{D'}{F'} \yS$, as desired.
\end{proof}

\begin{exam}
	Let $\yS$ and $\yI$ be a specification and an implementation, again as depicted in Figures~\ref{fig:newspec} and~\ref{fig:newimpl}, respectively. 
	Take the  languages $D=(a+b)^\star ax$ and $F=ab^+ x$. 
	From Example~\ref{exnew} we know that $\yI\yconf{D}{F} \yS$  does not hold, because $ababax\in otr(I)\cap D\cap \ycomp{otr}(S)$.
	Take the finite languages $D'=\{ax,ababax\}$ and $F'=\{abx\}$.
	Clearly, $D'\ysse D$ and $F'\ysse F$.
	It is clear that $ax,ababax\in otr(\yI)$, but $abx \notin otr(\yI)$, and $ax,abx \in otr(\yS)$, but $ababax\not\in otr(\yS)$, so that  $ababax\in otr(\yI)\cap D'\cap \ycomp{otr}(\yS)$ whereas $otr(\yI)\cap F'\cap otr(\yS)=\emptyset$.
	So, by Proposition~\ref{prop:equiv-conf}, $\yI\yconf{D'}{F'} \yS$ does not hold, as desired. \yfim
\end{exam}

\subsection{The ioco Conformance Relation}\label{subsec:ioco} 

Let $\yS$ be a specification IOLTS and  let $\yI$ be an implementation IOLTS.
The {\bf ioco}-conformance relation~\cite{tret-model-2008} essentially
requires that any observable trace $\ysi$ of $\yI$ is also an observable trace of $\yS$ and, further, if $\ysi$ leads $\yI$ to a location from which $\yI$ can emit the output label $\ell$, then $\ysi$  must also lead $\yS$ to a location from which the same label $\ell$ can also be output. 
That is, the implementation $\yI$ cannot emit a symbol $\ell$ that is not an output option already specified by the specification $\yS$, no matter to which locations any observable behavior of $\yS$ lead the two models to.

The preceding discussion motivates the following definitions. 
\begin{defi}[\cite{tret-model-2008}]\label{def:out-after}
	Let $L=L_I\cup L_U$, $\yS=\yioltsI$, and $\yioNI$.
	\begin{enumerate}
		\item Define $\yout\!: \ypow{S}\rightarrow L_I\cup L_U$ by $\yout(V)= \bigcup\limits_{s\in V}\{\ell\in L_U\yst \ytrt{s}{\ell}{}\}$.
		\item Define $\yafter\!\!\!: S\times  L^\star\rightarrow \ypow{S}$ by  $s \yafter \ysi = \{q \yst \ytrt{s}{\ysi}{q}\}$, for all $s\in S$ and all $\ysi\in L^\star$.
		\item Define $\yI \yioco \yS$ if and only if 
		$\yout(q_0 \yafter \ysi)\ysse \yout(s_0 \yafter \ysi)$, for all  $\ysi\in otr(\yS)$.
	\end{enumerate}   
\end{defi}
We may write $\yout(s)$ instead of $\yout(\{s\})$.
\begin{exam}
	Consider that the IOLTS depicted in Figure~\ref{fig:iots2} is an implementation $\yI=\yioltsI$ to be put under  test, where $L_I=\{but?\}$ and $L_U=\{co\!f\!\!f\!e!, tea!\}$.
	We then see that $\yout(\{s_1\})=\{but?,tea!\}$, and $\yout(tea!)=\{but?\}$.
	Also, $s_0\yafter but?tea!=\{tea,s_0\}$ and $cof \yafter but?=\{s_1,s_2\}$.
\end{exam}

Next we show that Definition~\ref{def:conf} subsumes the classical {\bf ioco}-conformance~\cite{tret-model-2008}. 
\begin{lemm}\label{lemm:ioco-reg}
	Let $\yS=\yioS$ be a specification IOLTS and let  $\yI=\yioQ$ be an implementation IOLTS.
	Then  $D=otr(\yS) L_U$ is a regular language over $L_I\cup L_U$, and we have that $\yI \yioco \yS$ if and only if
	$\yI \yconf{D}{\yemp} \yS$.
\end{lemm}
\begin{proof}
	From Definition~\ref{def:iolts-semantics} we know that the semantics of an IOLTS is given by the semantics of its underlying LTS.
	So, for the remainder of this proof, when we write $\yS$ and $\yI$ we will be referring to the underlying LTSs 
	of the given IOLTSs $\yS$ and $\yI$, respectively.
	
	By Proposition~\ref{prop:lts-fsa} 
	we see that $otr(\yS)$ and $L_U$ are regular languages.
	Hence, by Proposition~\ref{prop:reg-closure} we conclude that $D$ is also a regular language.
	
	Now, we show that $\yI\yioco \yS$ if and only if $\yI \yconf{D}{\yemp} \yS$.
	First assume that we have $\yI \yconf{D}{\yemp} \yS$. 
	Because $\yI\cap \yemp\cap \yS=\yemp$, it is clear from Definition~\ref{def:conf} that $\yI \yconf{D}{\yemp} \yS$ is equivalent  
	to $otr(\yI)\cap D\ysse otr(\yS)$.
	In order to prove that $\yI \yioco \yS$, let $\ysi\in otr(\yS)$ and let $\ell\in \yout(q_0 \yafter \ysi)$.
	We must show that $\ell\in \yout(s_0\yafter \ysi)$.
	Because  $\ell\in \yout(q_0\yafter \ysi)$ we get 
	$\ysi, \ysi\ell\in otr(\yI)$.
	Since $\ell\in L_U$, we get $\ysi\ell\in otr(\yS) L_U$ and so $\ysi\ell\in D$.
	We conclude that $\ysi\ell\in otr(\yI)\cap D$.
	Since we already know that $otr(\yI)\cap D\ysse otr(\yS)$, we now have 
	$\ysi\ell \in otr(\yS)$.
	So, $\ell\in \yout(s_0 \yafter \ysi)$, as desired.
	
	Next, assume that $\yI \yioco \yS$ and we want to show that $\yI \yconf{D}{\yemp} \yS$ holds.
	Since $otr(\yI)\cap \yemp\cap otr(\yS)=\yemp$, 
	the first condition of Definition~\ref{def:conf} is immediately verified.
	We now turn to the second condition of Definition~\ref{def:conf}.
	In order to show that $otr(\yI)\cap D\ysse otr(\yS)$,
	let $\ysi\in otr(\yI)\cap D$.
	Then, $\ysi\in D$ and so $\ysi=\yal\ell$ with $\ell\in L_U$ and $\yal\in otr(\yS)$, because $D=otr(\yS) L_U$.
	Also, $\ysi\in otr(\yI)$ gives $\yal\ell\in otr(\yI)$, and so $\yal\in otr(\yI)$.
	Then, because $\ell\in L_U$, we get $\ell\in \yout(q_0 \yafter \yal)$.
	Because we assumed $\yI \yioco \yS$ and we have $\yal\in otr(\yS)$, we also get $\ell\in \yout(s_0 \yafter \yal)$,
	and so $\yal\ell\in otr(\yS)$. Because $\ysi=\yal\ell$, we have $\ysi\in otr(\yS)$.
	We have, thus, showed that $otr(\yI)\cap D\ysse otr(\yS)$, as desired.
\end{proof}

The next example illustrates lemma~\ref{lemm:ioco-reg}.
\begin{exam}
	\label{exD-ioco}
	Let $\yS$ be the specification as depicted in Figure~\ref{fig:newspec}, and let the
	implementation $\yI$ be as depicted in Figure~\ref{fig:newimpl},
	with  the extra transition $\ytr{q_3}{x}{q_0}$.
	Recall that $L_I=\{a,b\}$, $L_U=\{x\}$. 
	
	Figure~\ref{fig:newspecD} shows a FSA $\yD$ such that $L(\yD)=D=otr(\yS) L_U$. 
	From Figure~\ref{fig:newspec} and the new $\yI$ it is apparent that 
	$\ytrt{s_0}{aa}{s_3}$ and also $\ytrt{q_0}{aa}{q_3}$.
	We also see that $x\in \yout(q_3)$, but $x\not\in \yout(s_3)$.
	So, by Definition~\ref{def:out-after}, $\yI\yioco \yS$ does not hold.
	Now take $\ysi=aax$.
	Since $aa\in otr(\yS)$, we get $\ysi\in otr(\yS)L_U=D$.
	Also, $\ysi\in otr(\yI)$ and $\ysi\not\in otr(\yS)$, so that
	$\ysi\in otr(\yI)\cap D\cap \ycomp{otr}(\yS)$.
	Using Proposition~\ref{prop:equiv-conf}, we conclude that $\yI\yconf{D}{\yemp} \yS$ does not hold too, as expected. \yfim
\end{exam} 

In the next example the new notion of conformance,  $\yconf{D}{F}{\!\!}$,  is seen to be able to capture non-conformance situations where the classical conformance relation, $\yioco{\!\!}$, would always yield positive results. 
\begin{exam}
	\label{exD-conf}
	Consider the models $\yS$ and $\yI$, as depicted in Figures~\ref{fig:newspec} and~\ref{fig:newimpl}, respectively, and assume $L_I=\{a,b\}$, $L_U=\{x\}$. 
	
	Figure~\ref{fig:newspecD} depicts a FSA $\yD$ such that $L(\yD)=otr(\yS) L_U$. 
	From Figures~\ref{fig:newspec} and~\ref{fig:newimpl} we can check that 
	there is no $\ysi \in (L_I \cup L_U)^\star$, $s\in S $ and $q\in Q$ such that $\ytrt{s_0}{\ysi}{s}$ and $\ytrt{q_0}{\ysi}{q}$ with $x\in \yout(q)$, but $x\not\in \yout(s)$.
	So, by Definition~\ref{def:out-after}, $\yI\yioco \yS$ holds.
	
	Now let $F=\yemp$, $D=(a+b)^\star ax$, and take $\ysi=ababax \in D$.
	By inspection we see that $\ysi \in \ycomp{otr}(\yS)$ and also that $\ysi\in otr(\yI)$, so that $\ysi\in otr(\yI)\cap D\cap \ycomp{otr}(\yS)$.
	Using Proposition~\ref{prop:equiv-conf}, we conclude that $\yI\yconf{D}{\yemp} \yS$ does not hold, whereas $\yI\yioco \yS$ would always hold. 
	
	We also note that by taking $F\neq \yemp$  gives the test designer even more freedom to check
	whether  some behaviors that occur in the specification are, or are not, also represented in the implementation.
	That is, assuming $D=\yemp$ for the moment, if a verdict of conformance is obtained, there is a guarantee that the behaviors specified in $F$ are not present in the implementation, whereas a verdict of non-conformance would say that some behavior of $F$ is also present in the implementation.  \yfim
\end{exam} 

In fact, we can take $D$ to be a finite language, which is a very useful characteristic in practical applications. 
\begin{coro}\label{coro:ioco-finite}
	Let $\yS=\yioS$ be a specification and $\yI=\yioQ$ be an implementation.
	Then there is a finite language $D\ysse (L_I\cup L_U)^\star$ such that $\yI \yioco \yS$ if and only if
	$\yI \yconf{D}{\yemp} \yS$.
\end{coro}
\begin{proof}
	Follows immediately from Lemma~\ref{lemm:ioco-reg} and Corollary~\ref{coro:finite}.
\end{proof}

The {\bf ioco} relation can also be characterized as follows.
\begin{coro}\label{coro:ioco-charac}
	Let $\yS$ be a specification and  $\yI$ an implementation.
	Then $\yI \yioco \yS$ if and only if $otr(\yI)\cap T= \yemp$, where $T=\ycomp{otr}(\yS)\cap \big[otr(\yS) L_U\big]$.
\end{coro}
\begin{proof}
	From Lemma~\ref{lemm:ioco-reg} we have that $\yI \yioco \yS$ if and only if $\yI \yconf{D}{\yemp} \yS$, where 
	$D=otr(\yS)L_U$.
	From Proposition~\ref{prop:equiv-conf} we know that the latter holds if and only if 
	$otr(\yI)\cap \big[D\cap\ycomp{otr}(\yS)\big]=\yemp$.
\end{proof}

\section{Test Generation for IOLTS models}\label{sec:suites}

In this section, we show how to generate finite and complete test suites that can be used to verify models according to the generalized conformance relation defined in Section~\ref{sec:conformance},
and for arbitrary sets $D$ of desired and arbitrary sets $F$ of undesired behaviors.
Moreover, we want the generated test suites to be sound and exhaustive. 
We also examine the complexity of the algorithm that constructs the test suite, and we show that, for deterministic models, the algorithm is linear on the number of states of the implementation, for a given specification. 
Also recall Remark~\ref{rema:quiescent}.

\subsection{Complete Test Suite Generation}\label{sec:test-suites}

We first define a test suite as an arbitrary language over the set of symbols exchanged by the models.
\begin{defi}\label{def:testsuite}
	A test suite $T$ over an alphabet $L$ is any language over $L$, \emph{i.e.}, $T\ysse L^\star$. 
	Each $\ysi\in T$ is called a \emph{test case}.
\end{defi}
If a test suite $T$ is a regular language geared to detect bad observable behaviors in an IUT, then it could be represented by a FSA $\yA$.
The final states in $\yA$~---~the ``fail'' states~---~could then specify the set of  undesirable behaviors.
We could then say that an implementation $\yI$ satisfies, or adheres, to a test suite $T$ when no observable behavior of $\yI$ is a harmful behavior present in $T$.
\begin{defi}\label{def:adherence}
	A LTS  $\yI$ over $L$ \emph{adheres} to a test suite $T$ if and only if for all $\ysi\in otr(\yI)$ we have $\ysi\not\in T$.
	An IOLTS $\yI=\yioS$ with $L=L_I\cup L_U$ adheres to $T$ if and only if $\yS_\yI$ adheres to $T$.
\end{defi}


Given a pair of languages $(D,F)$ and a specification $\yS$, we want to generate test suites $T$ that are sound, that is, checking that an implementation  $\yI$ adheres to $T$ should always guarantee that $\yI$ also  $(D,F)$-conforms to $\yS$.
Moreover, the converse is also desirable, that is, we want a guarantee that if any implementation $\yI$ indeed $(D,F)$-conforms to the given specification $\yS$ than the adherence test with respect to $T$ will also always be positive.
We also remark that,  since our notion of adherence is geared towards undesirable behaviors of the implementation, we found it more convenient to express the following notions of soundness and exhaustiveness with their meaning reversed when compared to the usual definitions found in the specialized literature~\cite{tret-model-2008,simap-generating-2014}. 
Of course, when requiring test suites to be complete, as asserted in our main results, this reversal is innocuous.
We now make these notions precise.
\begin{defi}\label{def:complete}
	Let $L$ be a set of symbols and let $T$ be a test suite over $L$.
	Let $\yS$ be a LTS over $L$, and let $D, F\ysse L ^\star$ be languages over $L$.
	We say that:
	\begin{enumerate}
		\item  $T$ is \emph{sound} for $\yS$ and $(D,F)$  if 
		$\yI$ adheres to $T$ implies $\yI \yconf{D}{F} \yS$, for all LTS $\yI$ over $L$.
		\item $T$ is \emph{exhaustive} for $\yS$ and $(D,F)$ if $\yI \yconf{D}{F} \yS$ implies that $\yI$ adheres to $T$, for all LTS $\yI$ over $L$.
		\item $T$ is \emph{complete} for $\yS$ and $(D,F)$ if it is both sound and exhaustive 
		for $\yS$ and $(D,F)$.
	\end{enumerate}
\end{defi}

Now we show that the test suite already hinted at Proposition~\ref{prop:equiv-conf} is, in fact, a complete test suite.
Furthermore, in a sense, it is also unique.

\begin{lemm}\label{lemm:always-complete}
	Let $\yS$ be a LTS over $L$, and $D$, $F\ysse L^\star$.
	Then, 
	$T=\big[(D\cap\ycomp{otr}(\yS))\cup(F\cap otr(\yS))\big]$ is the only complete test suite for $\yS$ and $(D,F)$. 
\end{lemm}
\begin{proof}
	Let $\yI$ be an arbitrary LTS over $L$.
	From Proposition~\ref{prop:equiv-conf} we get that $\yI\yconf{D}{F} \yS$ if and only if $otr(\yI)\cap T=\yemp$.
	But the latter holds if and only if  $\yI$ adheres to $T$. 
	Hence, $\yI$ adheres to $T$ if and only if  $\yI\yconf{D}{F} \yS$, that is, $T$ is complete.
	
	Now let $Z\ysse L^\star$, with  $Z\neq T$, and assume that $Z$ is complete for $\yS$ and $(D,F)$. 
	Fix any implementation $\yI$.
	Since $T$ is complete we have that $\yI$ adheres to $T$ if and only if  $\yI\yconf{D}{F} \yS$.
	Since $\yI$ adheres to $T$ if and only if $otr(\yI)\cap T=\yemp$, we get $\yI\yconf{D}{F} \yS$ if and only if  $otr(\yI)\cap T=\yemp$.
	Likewise. $\yI\yconf{D}{F} \yS$ if and only if  $otr(\yI)\cap Z=\yemp$.
	Hence,  $otr(\yI)\cap T=\yemp$ if and only if $otr(\yI)\cap Z=\yemp$.
	But $Z\neq T$ gives some $\ysi\in L^\star$ such that $\ysi\in T$ and $\ysi\not\in Z$, 
	the case $\ysi\not\in T$ and $\ysi\in Z$ being analogous. 
	We now have $\ysi\in T\cap \ycomp{Z}$.
	It is simple to construct an IUT $\yI$ with $\ysi\in otr(\yI)$ and thus reach a contradiction,
	because we would have $\ysi\in otr(\yI)\cap T$ and $\ysi\not\in otr(\yI)\cap Z$.
\end{proof}


Lemma~\ref{lemm:always-complete} says that the test suite $T=\big[(D\cap\ycomp{otr}(\yS))\cup(F\cap otr(\yS))\big]$ is complete for a specification $\yS$ and the pair of languages $(D,F)$.
So, given an implementation $\yI$, if one wants to check if it $(D,F)$-conforms to $\yS$ it suffices to check if
$\yI$ adheres to $T$, that is, by Definition~\ref{def:adherence},  it suffices to check that we have $otr(\yI)\cap T=\yemp$.
Using Proposition~\ref{prop:lts-fsa} we see that $otr(\yS)$ is a regular language, and so,  by Proposition~\ref{prop:reg-closure} so is $\ycomp{otr}(\yS)$.
Thus, if $D$ and $F$ are regular languages, then Proposition~\ref{prop:reg-closure} again says that $T$ is also a regular language.
Further, it also says that, given FSAs $\yA_D$ and $\yA_F$ specifying $D$ and $F$, respectively, we can effectively construct a FSA $\yA_T$ whose semantics is the test suite $T$. 
Using Propositions~\ref{prop:lts-fsa} and \ref{prop:reg-closure} again, we know that $otr(\yI)\cap T$ is also a regular language.
Moreover, if we are in a ``white-box'' testing scenario, that is, we have access to the syntactic  description of $\yI$, we can construct the FSA  $\yA_\yI$, associated to the implementation $\yI$.
Then, using Propositions~\ref{prop:lts-fsa} and \ref{prop:reg-closure} once more, we can effectively construct a FSA $\yA$ whose language is just  $otr(\yI)\cap T$.
Then, a simple breadth-first traversal algorithm applied to $\yA$ can check if  $otr(\yI)\cap T=\yemp$, so that we can effectively decide if  $\yI$ $(D,F)$-conforms to $\yS$. 
Proposition~\ref{prop:conf-poli} details the time complexity of such an algorithm. 

A simple strategy in a test run should extract input strings from the FSA $\yA_T$, i.e., strings (test cases) of the language $T$ to be applied to an implementation. 
Hence we would get that the IUT $(D,F)$-conforms to the specification if all strings did not run to completion in IUT. 
Otherwise, the IUT does not $(D,F)$-conform to the specification. 

\begin{exam}\label{example:conf}
	Let the IOLTS $\yS$ of Figure~\ref{fig:spec} be the specification,
	with $L_I=\{a,b\}$, and $L_U=\{x\}$. 
\begin{figure}[htb]
	\centering
	\subfigure[A LTS specification $\yS$.]{\label{fig:spec}
		\begin{tikzpicture}[font=\sffamily,node distance=1cm, auto,scale=0.75,transform shape]
		\node[ initial by arrow, initial text={}, punkt] (s0) {$s_0$};
		\node[punkt, inner sep=3pt,right=3cm and 3cm  of s0] (s1) {$s_1$};
		\node[punkt, inner sep=3pt,below =2cm and 1.5cm of s1] (s2) {$s_2$};
		\node[punkt, inner sep=3pt,below =2cm and 1.5cm of s0] (s3) {$s_3$};
		
		\path (s0)    edge [pil]   	node[anchor=north,above]{a} (s1);
		\path (s1)    edge [loop above] node   {a} (s1);
		\path (s1)    edge [pil]   	node[anchor=north,left]{x} (s2);
		\path (s1)    edge [pil]   	node[anchor=north,right]{b} (s3);
		
		\path (s0)    edge [pil]   	node[anchor=north,left]{b} (s3);
		\path (s3)    edge [loop below] node   {a} (s3);    
		\path (s2)    edge [pil,bend right=25]   	node[anchor=north,right]{a} (s1);
		\path (s2)    edge [pil,bend right=25]   	node[anchor=south]{b} (s3);
		\path (s3)    edge [pil,bend right=25]   	node[anchor=north]{b} (s2);
		
		\end{tikzpicture} 
	}
	\subfigure[A LTS implementation $\yI$.]{ \label{fig:impl}
		\begin{tikzpicture}[font=\sffamily,node distance=1cm, auto,scale=0.75,transform shape]
		\node[ initial by arrow, initial text={}, punkt] (s0) {$q_0$};
		\node[punkt, inner sep=3pt,right=3cm and 3cm  of s0] (s1) {$q_1$};
		\node[punkt, inner sep=3pt,below =2cm and 1.5cm of s1] (s2) {$q_2$};
		\node[punkt, inner sep=3pt,below =2cm and 1.5cm of s0] (s3) {$q_3$};
		
		\path (s0)    edge [pil]   	node[anchor=north,above]{a} (s1);
		\path (s1)    edge [loop above] node   {a} (s1);
		\path (s1)    edge [pil]   	node[anchor=north,left]{x} (s2);
		\path (s1)    edge [pil]   	node[anchor=north,right]{b} (s3);
		
		\path (s0)    edge [pil]   	node[anchor=north,left]{b} (s3);
		\path (s3)    edge [loop below] node   {a} (s3);    
		\path (s2)    edge [pil,bend right=25]   	node[anchor=north,right]{a} (s1);
		\path (s2)    edge [pil,bend right=25]   	node[anchor=south]{b} (s3);
		\path (s3)    edge [pil,bend right=25]   	node[anchor=north]{b,x} (s2);
		
		\end{tikzpicture}
	}
	\caption{LTS models.} 
\end{figure}
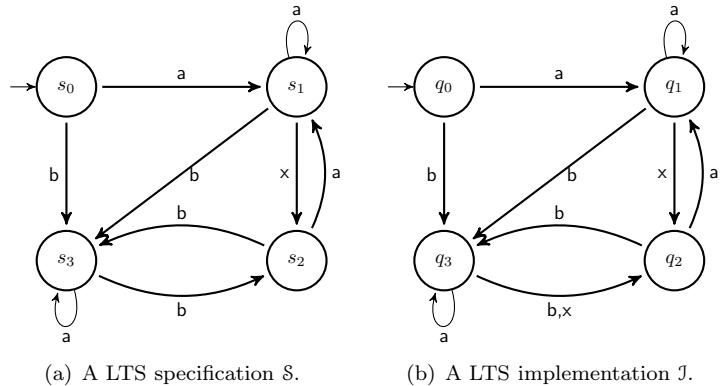
	The FSA $\overline{\yS}$, depicted in Figure~\ref{fig:compspec}, is such that $L(\overline{\yS})=\ycomp{otr}(\yS)$.
\begin{figure}[htb]
	\centering
	\subfigure[A FSA for $\ycomp{otr}(\yS)$ of Figure~\ref{fig:spec}.]{\label{fig:compspec}
		\begin{tikzpicture}[font=\sffamily,node distance=1cm, auto,scale=0.75,transform shape]
		\node[ initial by arrow, initial text={}, punkt] (s0) {$\overline{s}_0$};
		\node[punkt, inner sep=3pt,right=3cm and 3cm  of s0] (s1) {$\overline{s}_1$};
		\node[punkt, inner sep=3pt,below =2cm and 1.5cm of s1] (s2) {$\overline{s}_2$};
		\node[punkt, inner sep=3pt,below =2cm and 1.5cm of s0] (s3) {$\overline{s}_3$};
		\node[punkt, accepting, inner sep=3pt,below left =1.5cm and 1.2cm of s3] (err) {$err$};

		\path (s0)    edge [pil]   	node[anchor=north,above]{a} (s1);
		\path (s1)    edge [loop above] node   {a} (s1);
		\path (s1)    edge [pil]   	node[anchor=north,left]{x} (s2);
		\path (s1)    edge [pil]   	node[anchor=north,right]{b} (s3);
		
		\path (s0)    edge [pil]   	node[anchor=north,left]{b} (s3);
		\path (s3)    edge [loop below] node   {a} (s3);    
		\path (s2)    edge [pil,bend right=25]   	node[anchor=north,right]{a} (s1);
		\path (s2)    edge [pil,bend right=25]   	node[anchor=south]{b} (s3);
		\path (s3)    edge [pil,bend right=25]   	node[anchor=north]{b} (s2);
		
		\path (s0)    edge [pil,bend right=25]   	node[anchor=south]{x} (err);
		\path (s3)    edge [pil,]   	node[anchor=south]{x} (err);
		\path (s2)    edge [pil,bend left=25]   	node[anchor=north]{x} (err);
		
		\end{tikzpicture}
	}
	\subfigure[FSA $\yD$ for $D=otr(\yS)\cdot L_U$.]{ \label{fig:specD}
		\begin{tikzpicture}[font=\sffamily,node distance=1cm, auto,scale=0.75,transform shape]
		\node[ initial by arrow, initial text={}, punkt] (s0) {$d_0$};
		\node[punkt, inner sep=3pt,right=3cm and 3cm  of s0] (s1) {$d_1$};
		\node[punkt, inner sep=3pt,below =2cm and 1.5cm of s1] (s2) {$d_2$};
		\node[punkt, inner sep=3pt,below =2cm and 1.5cm of s0] (s3) {$d_3$};
		\node[punkt, accepting, inner sep=3pt,below right =2.5cm and 1cm of s3] (err) {$D$};

		\path (s0)    edge [pil]   	node[anchor=north,above]{a} (s1);
		\path (s1)    edge [loop above] node   {a} (s1);
		\path (s1)    edge [pil]   	node[anchor=north,left]{x} (s2);
		\path (s1)    edge [pil]   	node[anchor=north,right]{b} (s3);
		
		\path (s0)    edge [pil]   	node[anchor=north,left]{b} (s3);
		\path (s3)    edge [loop below] node   {a} (s3);    
		\path (s2)    edge [pil,bend right=25]   	node[anchor=north,right]{a} (s1);
		\path (s2)    edge [pil,bend right=25]   	node[anchor=south]{b} (s3);
		\path (s3)    edge [pil,bend right=25]   	node[anchor=north]{b} (s2);
		
		\path (s0)    edge [pil,bend right=60]   	node[anchor=south, left]{x} (err);
		\path (s3)    edge [pil,]   	node[anchor=south]{x} (err);
		\path (s2)    edge [pil,bend left=25]   	node[anchor=north, right]{x} (err);
		\path (s1)    edge [pil, bend left=60]   	node[anchor=north,right]{x} (err);
		
		\end{tikzpicture}
	}
	\caption{FSAs for languages  $\ycomp{otr}(\yS)$ and $D$.} 
\end{figure}
	The FSA $\yD$, shown in Figure~\ref{fig:specD}, is such that $L(\yD)=D$, where $D=otr(\yS)L_U$.
	The product $\yD \times \ycomp{\yS}$ is illustrated in Figure~\ref{fig:ts-spec}.
	The language accepted by $\yD \times \ycomp{\yS}$ is $TS=D\cap \ycomp{otr}(\yS)$.
	According to the Lemma~\ref{lemm:always-complete} $TS$ is a complete test suite for the specification $\yS$ and the pair of languages $(D,F)$ with $F=\yemp$. 
	Let $\yI$ be the IUT shown in Figure~\ref{fig:impl}, where $bax \in TS$. 
	Since $bax \in L(\yI)$ we conclude that $\yI$ does not $(D,F)$-conform to 
	$\yS$. \yfim
\end{exam}

We remark that, in order to construct $\yI$ with $\ysi\in otr(\yI)$, in Lemma~\ref{lemm:always-complete}, it was crucial that we had no restrictions on the size of $\yI$, because we have no control over the size of the witness  $\ysi$. 
This indicates that the size of the implementations to be put under test will affect the complexity of complete test suites that are generated
to verify $(D,F)$-conformance.
We investigate the complexity  of complete test suites in the next subsection. 

\begin{figure}[htb]
	\centering
	\subfigure[FSA $\yD \times \ycomp{\yS}$ for the language $TS$.]{\label{fig:ts-spec}
		\begin{tikzpicture}[font=\sffamily,node distance=1cm, auto,scale=0.65,transform shape]
		\node[ initial by arrow, initial text={}, punkt] (s0) {$\overline{s}_0d_0$};
		\node[punkt, inner sep=3pt,right=3cm and 3cm  of s0] (s1) {$\overline{s}_1d_1$};
		\node[punkt, inner sep=3pt,below =2cm and 1.5cm of s1] (s2) {$\overline{s}_2d_2$};
		\node[punkt, inner sep=3pt,below =2cm and 1.5cm of s0] (s3) {$\overline{s}_3d_3$};
		\node[punkt, accepting, inner sep=3pt,below right =2.5cm and 1cm of s3] (err) {$F$};
		
		\node[punkt, inner sep=3pt,below right =2cm and 1.2cm of s1] (s2D) {$\overline{s}_2 D$};
		%
		
		\path (s0)    edge [pil]   	node[anchor=north,above]{a} (s1);
		\path (s1)    edge [loop above] node   {a} (s1);
		\path (s1)    edge [pil]   	node[anchor=north,left]{x} (s2);
		\path (s1)    edge [pil]   	node[anchor=north,right]{b} (s3);
		
		\path (s0)    edge [pil]   	node[anchor=north,left]{b} (s3);
		\path (s3)    edge [loop below] node   {a} (s3);    
		\path (s2)    edge [pil,bend right=25]   	node[anchor=north,right]{a} (s1);
		\path (s2)    edge [pil,bend right=25]   	node[anchor=south]{b} (s3);
		\path (s3)    edge [pil,bend right=25]   	node[anchor=north]{b} (s2);
		
		\path (s0)    edge [pil,bend right=60]   	node[anchor=south, left]{x} (err);
		\path (s3)    edge [pil,]   	node[anchor=south]{x} (err);
		\path (s2)    edge [pil,bend left=25]   	node[anchor=north, right]{x} (err);
		
		\path (s1)    edge [pil,bend left=20]   	node[anchor=north,right]{x} (s2D);
		%
		%
		%
		\end{tikzpicture}
	}
	\subfigure[TP $\yT$ {\bf ioco}-complete for $\yS$.
	]{ \label{fig:tp-spec}
		\begin{tikzpicture}[font=\sffamily,node distance=1cm, auto,scale=0.65,transform shape]
		\node[ initial by arrow, initial text={}, punkt] (s0) {$t_0$};
		\node[punkt, inner sep=3pt,right=3cm and 3cm  of s0] (s1) {$t_1$};
		\node[punkt, inner sep=3pt,below =2cm and 1.5cm of s1] (s2) {$t_2$};
		\node[punkt, inner sep=3pt,below =2cm and 1.5cm of s0] (s3) {$t_3$};
		\node[punkt,  inner sep=3pt,below left =1.5cm and .6cm of s3] (err) {$\yfail$};

		\path (s0)    edge [pil]   	node[anchor=north,above]{$a$} (s1);
		\path (s1)    edge [loop above] node   {$a$} (s1);
		\path (s1)    edge [pil]   	node[anchor=north,left]{$x$} (s2);
		\path (s1)    edge [pil]   	node[anchor=north,right]{$b$} (s3);
		
		\path (s0)    edge [pil]   	node[anchor=north,left]{$b$} (s3);
		\path (s3)    edge [loop below] node   {$a$} (s3);    
		\path (s2)    edge [pil,bend right=25]   	node[anchor=north,right]{$a$} (s1);
		\path (s2)    edge [pil,bend right=25]   	node[anchor=south]{$b$} (s3);
		\path (s3)    edge [pil,bend right=25]   	node[anchor=north]{$b$} (s2);
		
		\path (s0)    edge [pil,bend right=25]   	node[anchor=south]{$x$} (err);
		\path (s3)    edge [pil,]   	node[anchor=south]{$x$} (err);
		\path (s2)    edge [pil,bend left=25]   	node[anchor=north]{$x$} (err);
		\path (err)    edge [loop below] node   {$x$} (err);
		\end{tikzpicture}
	}
	\caption{Complete test suites.} 
\end{figure}
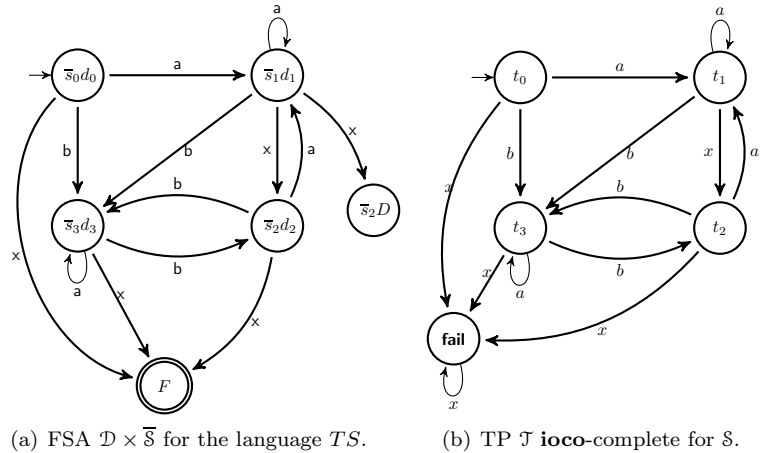

\subsection{On the Complexity of Test Suites}\label{sec:suites-complexity}

Another important issue is the size of test suites. If $\yS$ is an IOLTs with $n_\yS$ states and $t_\yS$ transitions, then $n_\yS -1 \leq t_\yS\leq n_\yS^2$, but usually $t_\yS$ is much larger than $n_\yS$.
Hence, we will take the number of transitions as an adequate measure of the size of an IOLTS model. 

Let  $D, F\ysse L^\star$, and $\yS=\yltsS$  deterministic with $n_S$ states. 
Lemma~\ref{lemm:always-complete} says that  $T=\big[(D\cap\ycomp{otr}(\yS))\cup(F\cap otr(\yS))\big]$ is complete for $\yS$ and $(D,F)$.
Assume that $L(\yA_D)=D$ and $L(\yA_F)=F$ where $\yA_D$ and $\yA_F$ are FSA.
By Definition~\ref{def:fsalang} and Propositions~\ref{prop:no-eps} and~\ref{prop:fsa-complete} we can assume that $\yA_D$ and $\yA_F$ are complete FSA with  $n_D$ and $n_F$ states, respectively.
By Propositions~\ref{prop:lts-deterministic} and~\ref{prop:lts-fsa},  the FSA $\yA_1$ induced by  $\yS$ will also be deterministic with $n_S$ states and $L(\yA_1)=otr(\yS)$.
Using Proposition~\ref{prop:fsa-complete} again we can  get a complete FSA $\yA_2$ with $n_S+1$ states, and such that $L(\yA_2)=L(\yA_1)=otr(\yS)$.
Hence, from Propositions~\ref{prop:reg-closure} and~\ref{prop:fsa-complete}, and Remark~\ref{rem:bounds}, we can get a complete FSA $\yA_3$ with at most $(n_S+1) n_F$ states and such that $L(\yA_3)=L(\yA_F)\cap L(\yA_2)=F\cap otr(\yS)$.
Consider the complete FSA $\yB_2$ obtained from $\yA_2$ by reversing its set of final states, so that
$L(\yB_2)=\ycomp{L(\yA_2)}=\ycomp{otr}(\yS)$.
With Proposition~\ref{prop:fsa-complete}  and Remark~\ref{rem:bounds} we  get a complete FSA $\yB_3$ with $(n_S+1) n_D$ states and such that $L(\yB_3)=L(\yA_D)\cap L(\yB_2)=D\cap \ycomp{otr}(\yS)$.
Proposition~\ref{prop:fsa-complete} yields a FSA $\yC$ with $(n_S+1)^2n_Dn_F$ states and such that $L(\yC)=L(\yA_3)\cup L(\yB_3)=T$.
So, if $D$ and $F$ are regular languages 
we can construct a complete FSA $\yT$ with $(n_S+1)^2n_Dn_F$ states and
$L(\yT)=T$ is the complete test suite.   
\begin{prop}\label{prop:conf-poli}
	Let $L$ be an alphabet with $|L|=n_L$.
	Let $\yS$ and $\yI$ be deterministic IOLTSs over $L$ with $n_S$ and $n_I$ states, respectively.
	Let $D$, $F\ysse L^\star$ be regular languages over $L$, and let $\yA_D$ and $\yA_F$ be complete FSA over $L$ with $n_D$ and $n_F$ states, respectively,
	and such that $L(\yA_D)=D$, $L(\yA_F)=F$.
	Then, we can construct a complete FSA $\yT$ with $(n_S+1)^2n_Dn_F$ states, and such that $L(\yT)$ is a complete test suite for $\yS$ and $(D,F)$.
	Moreover, there is an algorithm, with polynomial time complexity $\yoh{n_S^2n_In_Dn_Fn_L}$ that checks if $\yI \yconf{D}{F} \yS$. 
\end{prop}
\begin{proof}
	The preceding discussion, gives a complete FSA $\yT$ with at most  $(n_S+1)^2(n_Dn_F)$ states and such that $L(\yT)=T=\big[(D\cap\ycomp{otr}(\yS))\cup(F\cap otr(\yS))\big]$.
	By Lemma~\ref{lemm:always-complete}.
	$\yI \yconf{D}{F} \yS$ if and only if $otr(\yI)\cap T=\yemp$.
	Likewise, we can  get a complete FSA $\yA$ with $n_I+1$ states, and  such that $otr(\yI)=L(\yA)$. 
	From Remark~\ref{rem:bounds} we get a FSA $\yB$ with at most $(n_S+1)^2(n_I+1)n_Dn_F$ states and with $L(\yB)=L(\yA)\cap L(\yT)=otr(\yI)\cap T$.
	A simple observation reveals that $\yB$ is a complete FSA with $(n_S+1)^2(n_I+1)n_Dn_Fn_L$ transitions, because both $\yA$ and $\yT$ are complete FSAs. 
	A state $(q,t)$ of $\yB$ is final if and only if $q$ is a final state of $\yA$ and $t$ is a final state of $\yT$. 
	So, using a standard breadth-first algorithm we can traverse  $\yB$ and determine is one of its final states is reached from its initial state, indicating that  
	$\yemp\neq L(\yA)\cap L(\yT)=otr(\yI)\cap T$, that is, if and only if 
	$\yI$ does not $(D,F)$-conform to $\yS$.
	Otherwise, if this condition is not met, we say that 
	$\yI$ does  $(D,F)$-conform to $\yS$.
	The breadth-first algorithm run in time proportional to the number of transitions  in $\yB$. Since $\yB$ is complete, it has at most  $\yoh{n_S^2n_In_Dn_Fn_L}$ transitions.
\end{proof}

Since we already know how to verify if the {\bf ioco}-conformance relation holds, we can state a similar result for checking {\bf ioco}-conformance when a description of the IUT is available.
\begin{theo}\label{prop:ioco-poli}
	Let $\yS$ and $\yI$ be deterministic  IOLTSs over $L$ with $n_S$ and $n_I$ states, respectively.
	Let $L=L_I\cup L_U$, and $|L|=n_L$.
	Then, we can effectively construct an algorithm with time complexity $\yoh{n_Sn_In_L}$ that  checks
	if $\yI \yioco \yS$ holds. 
\end{theo}
\begin{proof}
	Let $\yA=\yfsa{S_A}{s_0}{L}{\rho_A}{S_A}$ be the deterministic FSA induced by $\yS$ with $n_s$ states.
	Proposition~\ref{prop:lts-fsa} gives $otr(\yS)=L(\yA)$.
	
	Let $e,f\not\in S_A$, with $e\neq f$, and define $S_T=S_A\cup\{e,f\}$. 
	Extend the transition relation $\rho_A$ as follows.
	Start with $\rho_T=\rho_A$.
	Next: 
	(i) For any $\ell\in L_U$ and any $s\in S_A$ such that $(s,\ell,p)$ is not in $\rho_A$ for any $p\in S_A$, add $(s,\ell,f)$ to $\rho_T$;  
	(ii) for any $\ell\in L_I$ and any $s\in S_A$ such that $(s,\ell,p)$ is not in $\rho_A$ for any $p\in S_A$, add $(s,\ell,e)$ to $\rho_T$.
	Let $\yT=\yfsa{S_T}{s_0}{L}{\rho_T}{\{f\}}$.
	Since $\yA$ is deterministic, the construction implies that $\yT$ is deterministic with $n_S+2$ states. 
	Also $f$ is the only final state in $\yT$, and  $f$ and $e$ are sink states in $\yT$.
	
	Let $\ysi\in L^\star$.
	By the construction, we get $\yatrt{s_0}{\ysi}{f}$ in $\yT$ if and only if $\yatrt{s_0}{\yal}{p}\yatrt{}{\ell}{f}$ in $\yT$ for some $p\in S_A$,  where $\ell\in L_U$ and $\ysi=\yal\ell$.
	We claim that, for any $\ell\in L_U$, we have $\yatrt{s_0}{\yal}{p}\yatrt{}{\ell}{f}$ in $\yT$ with $p\in S_A$ if and only if we have $\yal\in L(\yA)$ and $\yal\ell\not\in L(\yA)$.
	Assume that  $\yatrt{s_0}{\yal}{p}\yatrt{}{\ell}{f}$ in $\yT$ with $\ell\in L_U$ and $p\in S_A$.
	Since $p\neq f$ and $p\neq e$, we  get $\yatrt{s_0}{\yal}{p}$ in $\yA$, that is, $\yal\in L(\yA)$.
	Suppose that $\yal\ell\in L(\yA)$. Then, we would get 
	$\yatrt{s_0}{\yal}{q}\yatrt{}{\ell}{r}$ in $\yA$, for some $q$, $r\in S_A$.
	But $\yA$ is deterministic and we already have $\yatrt{s_0}{\yal}{p}$ in $\yA$,
	so that $q=p$.
	Now we have $(p,\ell,f)$ in $\rho_T$ and  $(p,\ell,r)$ in $\rho_A$ with $\ell\in L_U$, a
	contradiction to item (i).
	For the converse, assume that $\yal\in L(\yA)$ and   $\yal\ell\not\in L(\yA)$ with $\ell\in L_U$.
	We  get $\yatrt{s_0}{\yal}{p}$ in $\yA$, which  implies $\yatrt{s_0}{\yal}{p}$ in $\yT$.
	If we had $(p,\ell,r)$ in $\rho_A$ for some $r\in S_A$, then we would get $\yal\ell\in L(\yA)$, which cannot happen.
	Then, item (ii) implies that  $(p,\ell,f)$ is in $\rho_T$, so that 
	$\yatrt{s_0}{\yal}{p}\yatrt{}{\ell}{f}$ in $\yT$.
	Putting it together we conclude that $\yatrt{s_0}{\ysi}{f}$ in $\yT$ if and only if $\ysi=\yal\ell$, $\yal\in L(\yA)$ and $\ysi\not\in L(\yA)$, for some $\ell\in L_U$ and some $\yal\in L^\star$.
	This shows that $L(\yT)=\ycomp{L(\yA)}\cap (L(\yA)L_U)=\ycomp{otr}(\yS)\cap (otr(\yS)L_U)$.
	
	Let $\yB$ be the deterministic FSA induced by $\yI$ with $n_I$ states and $otr(\yI)=L(\yB)$.
	By Proposition~\ref{prop:fsa-complete}, we get a complete FSA $\yC$ with $n_I+1$ states and such that $otr(\yI)=L(\yB)=L(\yC)$, and  a complete FSA $\yU$ with $n_S+3$ states, such that $L(\yU)=L(\yT)$.
	With $\yC$ and $\yU$ at hand, proceed as in the proof of Proposition~\ref{prop:conf-poli}. The desired algorithm will run in asymptotic worst case time complexity $\yoh{n_Sn_I n_L}$. 
\end{proof}

Assume that the alphabet $L_I\cup L_U$ is fixed.
Usually, one has access to the internal structure of the specification, that is, an IOLTS
model $\yS$ is given for the specification. 
Assume $\yS$ has $n$ states.
If we also know the internal structure of the implementation to be put under test, that 
is if we are in a ``white-box'' testing scenario, then we also have an IOLTS model $\yI$
for the implementation. 
In this case,  Theorem~\ref{prop:ioco-poli} says that we can test the {\bf ioco}-conformance relation between $\yS$ and $\yI$ in polynomial $\yoh{nm}$ time, for \emph{any} implementation with $\yoh{m}$ states. 
If the initial specification $\yS$ is also fixed, and we have several implementations to be put under test, then the verifying algorithm from Theorem~\ref{prop:ioco-poli} will run in time proportional to the size of each implementation. 
\begin{coro}\label{coro:lineartime}
	Fix an alphabet $L$ and a deterministic specification $\yS$, and assume that we have access to the syntactic descriptions of IUT models. Then we have an algorithm to verify whether $\yI \yioco \yS$ holds, and that runs in time $\yoh{t}$, where $t$ is the size of the implementation $\yT$, that is, $t$ is the number of transitions in $\yT$.
\end{coro}
\begin{proof}
	From the preceding discussion.
\end{proof}

\section{Testing IOLTS with Test Purposes}\label{sec:tretma-suites}





In Section~\ref{sec:suites} the testing architecture presupposed that one has access to a syntactic description of the IUTs.
In a contrasting setting, where IUTs are ``black-boxes'', we do not have access to their syntactic structure.
In this case, we can imagine a test setting where there is a tester $\yT$, and an implementation being tested, $\yI$, 
which are linked by a ``zero-capacity'' bidirectional and lossless communication channel.
In this setting, at each step either of two movements may occur: (i) the tester, or the ``artificial environment'', $\yT$ issues one of its output action symbols $x$ to $\yI$, and change its state. At once, the implementation $\yI$ accepts $x$  as one of its input symbols and also changes its state; or (ii) the movement is reversed
with the implementation $\yI$ sending one of its output symbols $y$ and changing its state, while $\yT$ 
accepts $y$ at once as one of its input symbols, and also changes its state.
Clearly, a sequence of type (i) moves can occur before a type (ii) move occurs; and vice-versa.
In other words, the net result is that $\yT$ and $\yI$ move in lock step, but with the input and output sets of symbols interchanged in the $\yT$ and $\yI$ models.
We will always refer to a symbol $x$ as an input or an output symbol from the perspective of the implementation $\yI$, unless there is an explicit mention to the contrary. 
Hence, one should write $\yI=\yio{S}{s_0}{I}{U}{T}$ and $\yT=\yio{Q}{q_0}{U}{I}{R}$.  

However, as a result of accepting an input $x$ from the tester $\yT$, the implementation $\yI$ may reach a so called \emph{quiescent state}.
Informally, those are states from which there are no transitions labeled by some output action symbol~\cite{Tretmans96TGIO,tret-model-2008}. 
In a practical scenario, from this point on the implementation could no longer send responses back to the tester, and the latter will have no way of ``knowing'' whether the implementation is rather slow, has timed out, or will not ever respond.
If we want to reason about this situation, within the formalism, it will be necessary to somehow signal the tester that the implementation is in a quiescent state.
A usual mechanism~\cite{tret-model-2008} is to imagine that the implementation has a special output symbol $\yde\in L_U$, and that $\yde$ is then sent back to $\yT$ when $\yI$ reaches state $s$, and no longer responds.
Since $\yI$ is not changing states in this situation, we add the self-loop $\ytr{s}{\yde}{s}$ to the set of transitions of $\yI$ in order to formally describe the situation.
On the tester side, being on a state $q$ and upon receiving a $\yde$ symbol from the implementation
the tester may decide whether receiving such a signal in state $q$ is appropriate or not, depending on the fault model it was designed for. 
If that response from the implementation was an adequate one, the tester may then move to another state $q'$ to continue the test run.
In this case, on the tester side we  add all appropriate $\ytr{q}{\yde}{q'}$ transitions%
\footnote{In~\cite{tret-model-2008}, a different symbol, $\yte$, was used to signal the acceptance of quiescence on the tester side, but for our formalism, that makes little difference, if any.}.

In this section we want apply our results of Section~\ref{sec:suites} to test architecture studied by Tretmans~\cite{tret-model-2008}, where a tester model  drives a black-box IUT model. 
Since in~\cite{tret-model-2008} quiescence is indicated in the models that are put under test, we must also formally prepare our models to deal with quiescent states. This is done in 
Definition~\ref{def:ioltsq}.
Further, must ensure that we are applying the more general results of  Section~\ref{sec:conformance}  to the same class of models that are considered in~\cite{tret-model-2008}.
Moreover, we must also guarantee that, the general notion of {\bf ioco}-conformance defined in Section~\ref{sec:conformance} induces the same {\bf ioco} relation as the notion of {\bf ioco}-conformance studied in~\cite{tret-model-2008}.    
Results described in Appendix~\ref{subsec:equiv-ioco} will be used to settle such matters, as well 
as other similar issues that may arise.   

We proceed as follows:
\begin{enumerate}
	\item We define a variation of IOLTS models, where the special symbol $\yde$ is used to indicate quiescence.
	\item We formalize the notion of an external tester in order to reason precisely about test runs. For that, we define test purposes~\cite{tret-model-2008} in Subsection~\ref{subsec:fault-models}.
	\item In  testing architectures where the internal structure of  IUTs is unknown, it is customary to   impose a series of restrictions over the formal models that describe the specifications, the IUTs and the test purposes~\cite{tret-model-2008}, so that some guarantees about the exchange of messages can be stated.
	Although our methods impose  almost no restrictions on the formal models, except for regularity of the $D$ and $F$ sets, in Subsection~\ref{subsec:ioco-test-cases} we look at the extra model restrictions imposed by Tretmans~\cite{tret-model-2008}.
	\item In Subsection~\ref{sec:tretma-complexity} we investigate the complexity of the test purposes that can be generated under these restrictions, and we establish a \emph{new asymptotic worst case exponential time} lower bound of the size of the test suite, or fault model.
	Other works hinted at possible exponential upper bounds on the size of test suites when requiring such suites to be complete. 
	We are not aware of any precise \emph{lower bounds} on the size of complete test suites, when treating these exact same restrictions as mentioned here. 
\end{enumerate}

We start with the following variation of Definition~\ref{def:iolts} incorporating quiescence in IOLTS models, 
following the preceding informal discussion.
\begin{defi}\label{def:ioltsq}
	A $\yde$-Input/Output Labeled Transition System ($\yde$-IOLTS) is a tuple $\yI=\yioltsI$, where: 
	\begin{enumerate}
		\item $\yioltsI$ is an IOLTS;
		\item $\yde\in L_U$ is a \emph{distinguished symbol} that will be used to indicate quiescence.
		\item For all states $s$ and $p$ in $S$ we have 
		$(s,\yde,p)\in T$ if and only if (a) $s=p$, and (b)  for all $x\in L$, $\ytr{s}{x}{}$ implies $x\in L_I\cup\{\yde\}$.
	\end{enumerate}
	A state $s\in S$ is said to be \emph{quiescent} if $\ytr{s}{\yde}{}$.
\end{defi}
We indicate the class of all $\yde$-IOLTSs with input alphabet $L_I$ and output alphabet $L_U$ by $\yiocq{I}{U}$.
The following example illustrates the situation.
\begin{exam}
	Consider the $\yde$-IOLTS depicted in Figure~\ref{fig:iots2simple}, where we have $L_I=\{a\}$ and $L_U=\{b,\yde\}$. 
	States  $s_1$ and $s_3$ are quiescent states, according to Definition~\ref{def:ioltsq}.
	
	Note that we have $\ytr{s_0}{\tau}{s_1}$, so that $\ytr{s_0}{x}{}$ does not imply $x\in L_I\cup\{\yde\}$.
	Hence, we do not have a self-loop $\yde$ at $s_0$ and so, according to to Definition~\ref{def:ioltsq}, it is not a quiescent state.
	Although $s_0$ is not quiescent, it can not, nevertheless, emit any output symbol back to a tester, that is, if $\ytr{s_0}{y}{s_1}$ then we have $y\not\in L_U$.
	It is only after we make the internal transition $\ytr{s_0}{\tau}{s_1}$ that the model can, then, issue the output symbol $\yde$ signaling quiescence.
\begin{figure}[htb]
	\center
	\begin{tikzpicture}[font=\sffamily,node distance=1cm, auto,scale=0.75,transform shape]
	\node[ initial by arrow, initial text={}, punkt] (s0) {$s_0$};
	\node[punkt, inner sep=3pt,right=2.0cm of s0] (s1) {$s_1$};
	\node[punkt, inner sep=3pt,right=2.0cm of s1] (s2) {$s_2$};
	\node[punkt, inner sep=3pt,right=2.0cm of s2] (s3) {$s_3$};
	
	\path (s0)    edge [pil,bend left=35,line width=0.5pt]   	node[right,above]	{$a$} (s1);
	\path (s0)    edge [pil,bend right=35,line width=0.5pt]  node[right, below] 	{$\tau$} (s1);
	
	\path (s1)    edge [pil,line width=0.5pt]  node[right, above] {$a$} (s2);
	\path (s1)    edge [loop above] node   {$\yde$} (s1);
	
	\path (s2)    edge [pil,bend right=35,line width=0.5pt]   	node[right, below]{$\tau$} (s3);
	\path (s2)    edge [pil,bend left=35,line width=0.5pt]  node[right, above] {$b$} (s3);
	
	\path (s3)    edge [loop above] node   {$a,\yde$} (s3);
	
	\end{tikzpicture}
	\caption{A simple $\delta$-IOLTS.}\label{fig:iots2simple}
\end{figure}
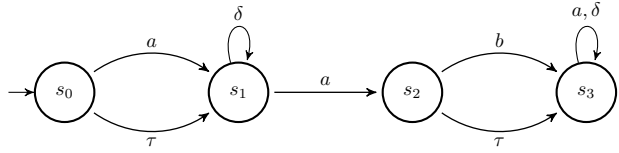
	\yfim
\end{exam}

In the test architecture studied in~\cite{tret-model-2008}, given an IOLTS $\yS=(S,s_0,L_I,L_U,T)\in\yltscT{L_I,L_U}$, a state $s\in S$ is said to be quiescent if, for all $x\in L_U\cup\{\tau\}$ we have $s\not\overset{\!\!\!x}{\rightarrow}$ in $\yS$. 
The fact that $s$ is quiescent is indicated by $\yde(s)$. 
Assuming that $\yde\not\in L_U$, the extended model $\yS_\yde=\yiolts{S}{s_0}{L_I}{L_U\cup \{\yde\}}{T\cup T_\yde}$ is defined, where $T_\yde = \{(s,\yde,s)\yst \yde(s)\}$.
That is, $\yS_\yde$ includes self-loops on the new output symbol $\yde$ at any quiescent state. 
Then, test runs the {\bf ioco}-relation, are all carried out and computed  referring to extended models.
Since in this section we are applying our results of  Sections~\ref{sec:conformance} and~\ref{sec:suites}
to the test architecture described in~\cite{tret-model-2008}, it is important to guarantee that the class of
$\yde$-IOLTS models from Definition~\ref{def:ioltsq} is coextensive with the class of extended models 
in~\cite{tret-model-2008}.
Moreover, we must also show that the {\bf ioco}-relation  used in~\cite{tret-model-2008} coincides with our Definition~\ref{def:out-after}.
Because the details of these considerations are not the focus of this section, we have grouped them in the appendix.
See Proposition~\ref{prop:class-delta}, for the guarantee that both classes of models are the same, and Proposition~\ref{prop:ioco2} that shows that both  {\bf ioco}-relations coincide.

This concludes the first step, of defining a class of models that includes the $\yde$ output symbol to indicate quiescence, as listed at the introduction to this section. 

\subsection{A Class of Fault Models}\label{subsec:fault-models}

One can formalize the external tester environment using the same notion of an  Input/Output Labeled Transition System to drive the test runs by blindly exchanging messages with the implementations in $\yiocq{I}{U}$ that are to be put under test~\cite{tret-model-2008}.
We will refer to these models here as \emph{test purposes}.
A test purpose $\yT$ has special \ypass\ and \yfail\  states.
Behaviors that reach a \yfail\  state are deemed harmful, and  those that reach a 
\ypass\ state are said to be acceptable. 
Once a test purpose has reached a {\bf fail} or a {\bf pass} state, then we have a verdict for the test, so that it is  reasonable to require  that there are no paths from a \ypass\ to a \yfail\ state, or vice-versa.
Note also that, if a test purpose $\yT$ sends symbols to an IUT $\yI=\yio{Q}{q_0}{I}{U}{R}$ and receives symbols from it, then the test purposes'  set of input symbols must be $L_U$  and its set of output symbols must be $L_I$.
A \emph{fault model} will be comprised by a finite set of test purposes, so that several conditions for acceptance and rejection of IUTs can be tested.

Of course, if we are dealing with implementations that treat quiescence, then we have $\yde\in L_U$, and so $\yde$ is a test purpose input symbol in $\yioc{U}{I}$.
The test designer is free to create test purposes whose $\yde$-transitions reflect the designer's intention when receiving $\yde$ symbols indicating quiescence in the implementation.  
But note that, from the perspective of the test designer, a test purpose is just an ordinary IOLTS in 
$\yioc{U}{I}$.

We now formalize these notions.
Since $\yiocq{U}{I}\ysse \yioc{U}{I}$, it is more profitable to express the following notions and results
using the full  $\yioc{U}{I}$ class, and specialize to more restricted classes of IOLTSs only when needed.

\begin{defi}\label{def:test-purpose}
	Let $L_I$ and $L_U$ be sets of input and output symbols, respectively, with $L=L_I\cup L_U$.
	A \emph{test purpose over $L$} is any  IOLTS $\yT\in\yioc{U}{I}$ such that
	for all $\ysi\in L^\star$ we have 
	neither $\ytrt{\yfail}{\ysi}{\ypass}$ nor  $\ytrt{\ypass}{\ysi}{\yfail}$.
	A \emph{fault model} over $L$ is 
	a finite collection of test purposes over $L$.
\end{defi}

Given  a test purpose $\yT$ and an IUT $\yI$, in order to formally describe the exchange of action symbols between $\yT$ and $\yI$, we define their cross-product LTS $\yT\times \yI$.%
\footnote{In~\cite{tret-model-2008} this is know as the \emph{parallel operator} $\vert [\cdot ]\vert$.}
\begin{defi}\label{defi:cross}
	Let $\yT=\yio{Q}{q_0}{U}{I}{R}\in\yioc{U}{I}$ and $\yI=\yioS\in\yioc{I}{U}$.
	Their cross-product is the LTS 
	$\yT\times \yI=\ylts{Q\times S}{(q_0,s_0)}{L}{P}$, where $L=L_I\cup L_U$ and $((q_1,s_1),x,(q_2,s_2))$ is a transition in $P$ if and only if either
	\begin{itemize}
		\item $x=\tau$, $s_1=s_2$ and $(q_1,\tau,q_2)$ is a transition in $R$, or 
		\item $x=\tau$, $q_1=q_2$ and $(s_1,\tau,s_2)$ is a transition in $T$, or
		\item $x\neq\tau$,  $(q_1,x,q_2)\in R$ and  $(s_1,x,s_2)\in T$.\yfim
	\end{itemize}
\end{defi}

We can now show that the collective behavior described in the cross-product always implies the same behavior about the two participating IOLTSs, and conversely.
\begin{prop}\label{prop:cross}
	Let $\yT=\yio{Q}{q_0}{U}{I}{R}\in\yioc{U}{I}$ and $\yI=\yioS\in\yioc{I}{U}$ be IOLTS, and let $\yT\times \yI$ be their cross-product.
	Then, we have
	\begin{enumerate}
		\item $\ytr{(t,q)}{\tau^k}{(p,r)}$ in $\yT\times \yI$, with $k\geq 0$, if and only if there are $n$, $m\geq 0$, with $n+m=k$, and  such that $\ytr{t}{\tau^n}{p}$ in $\yT$  and $\ytr{q}{\tau^m}{r}$ in $\yI$.
		\item $\ytrt{(t,q)}{\ysi}{(p,r)}$ in $\yT\times \yI$ if and only if $\ytrt{t}{\ysi}{p}$ in $\yT$  and $\ytrt{q}{\ysi}{r}$ in $\yI$, for all $\ysi\in L^\star$. 
	\end{enumerate}
\end{prop}
\begin{proof}
	The first item follows by an easy induction on $k\geq 0$.
	
	For the second item, first assume that $\ytrt{(t,q)}{\ysi}{(p,r)}$ in $\yT\times \yI$.
	According to Definition~\ref{def:path} we get $\ytr{(t,q)}{\mu}{(p,r)}$ with $h_\tau(\mu)=\ysi$.
	We proceed by induction on $|\mu|\geq 0$.
	When $|\mu|=0$, we get $\ysi=\yeps$, $t=p$, $q=r$ and the result follows easily.
	Next, assume  $\ytr{(t,q)}{\mu}{(t_1,q_1)}\ytr{}{x}{(p,r)}$ with $\vert x\vert=1$ and  $h_\tau(\mu x)=\ysi$.
	The induction hypothesis gives $\ytrt{t}{\rho}{t_1}$ in $\yT$ and $\ytrt{q}{\rho}{q_1}$ in $\yI$ with $h_\tau(\mu )=\rho$.
	If $x=\tau$, Definition~\ref{defi:cross} gives $\ytr{t_1}{\tau}{p}$ and $q_1=r$, or $\ytr{q_1}{\tau}{r}$ and $t_1=p$.
	We assume the first case, the reasoning for the other case being entirely analogous. 
	We now have $\ysi=h_\tau(\mu x)=h_\tau(\mu\tau)=h_\tau(\mu)=\rho$.
	We  have $q_1=r$, and $\ytrt{q}{\rho}{q_1}$ in $\yI$, so that now we can write $\ytrt{q}{\ysi}{r}$ in $\yI$.
	From $\ytrt{t}{\rho}{t_1}$ we get $\ytr{t}{\eta}{t_1}$ with $h_\tau(\eta)=\rho$.
	Hence, $\ytr{t}{\eta}{t_1}\ytr{}{\tau}{p}$, and since $h_\tau(\eta\tau)=h_\tau(\eta)=\rho=\ysi$ we also have $\ytrt{t}{\ysi}{p}$ in $\yS$, as desired.
	Next, we have the case when $x\neq \tau$.
	Then, $\ysi=h_\tau(\mu x)=h_\tau(\mu)h_\tau(x)=\rho x$.
	Now, Definition~\ref{defi:cross} gives $\ytr{t_1}{x}{p}$ and $\ytr{q_1}{x}{r}$.
	From $\ytrt{t}{\rho}{t_1}$ and $\ytrt{q}{\rho}{q_1}$ we obtain $\ytr{t}{\eta_1}{t_1}$ and $\ytr{q}{\eta_2}{q_1}$ with $h_\tau(\eta_1)=\rho=h_\tau(\eta_2)$.
	Hence, $\ytr{t}{\eta_1 x}{p}$ and $\ytr{q}{\eta_2 x}{r}$.
	Since $h_\tau(\eta_1 x)=h_\tau(\eta_1)h_\tau(x)=\rho x=\ysi=h_\tau(\eta_2)h_\tau(x)=h_\tau(\eta_2 x)$, we can write 
	$\ytrt{t}{\ysi}{p}$ in $\yT$  and $\ytrt{q}{\ysi}{r}$ in $\yI$, as desired.
	
	For the converse of item 2, we note that from $\ytrt{t}{\ysi}{p}$ in $\yT$  and $\ytrt{q}{\ysi}{r}$ in $\yI$, we get  $\ytr{t}{\mu_1}{p}$ and $\ytr{q}{\mu_2}{r}$, with $h_\tau(\mu_1)=\ysi=h_\tau(\mu_2)$.
	We induct on $|\mu_1|+|\mu_2|\geq 0$.
	If $\mu_1=\tau^n$, for some $n\geq 0$, we obtain $\ysi=h_\tau(\tau^n)=\yeps=h_\tau(\mu_2)$, and we must have $\mu_2=\tau^m$, for some $m \geq 0$.
	Using item 1, we obtain $\ytr{(t,q)}{\tau^{m+n}}{(p,r)}$, and since $h_\tau(\tau^{m+n})=\yeps=\ysi$ we arrive at $\ytrt{(t,q)}{\ysi}{(p,r)}$, as desired.
	Now, assume $\mu_1=\yal_1 x\ybe_1$ with $x\neq \tau$ and $\ybe_1=\tau^n$, for some $n\geq 0$. 
	Then, $\ysi=h_\tau(\mu_1)=h_\tau(\yal_1)x=h_\tau(\mu_2)$, and we conclude that
	$\mu_2=\yal_2 x\ybe_2$ with $h_\tau(\yal_2)=h_\tau(\yal_1)$ and $\ybe_2=\tau^m$, for some $m\geq 0$.
	So, we now have $\ytr{t}{\yal_1 }{t_1}\ytr{}{x}{t_2}\ytr{}{\ybe_1}{p}$ and $\ytr{q}{\yal_2 }{q_1}\ytr{}{x}{q_2}\ytr{}{\ybe_2}{r}$.
	Let $\rho=h_\tau(\yal_1)=h_\tau(\yal_2)$, so that $\ysi=\rho x$.
	The induction hypothesis gives $\ytrt{(t,q)}{\rho }{(t_1,q_1)}$, which means that
	$\ytr{(t,q)}{\eta}{(t_1,q_1)}$ and $h_\tau(\eta)=\rho$.
	From Definition~\ref{defi:cross} again we get $\ytr{(t_1,q_1)}{x }{(t_2,q_2)}$.
	From the proof of item 1 we get $\ytr{(t_2,q_2)}{\tau^{m+n}}{(p,r)}$.
	Collecting, we have 
	$\ytr{(t,q)}{\eta}{(t_1,q_1)}\ytr{}{x}{(t_2,q_2)}\ytr{}{\tau^{n+m}}{(p,r)}$.
	Since $h_\tau(\eta x\tau^{n+m})=h_\tau(\eta)x=\rho x=\ysi$,
	we can now write $\ytrt{(t,q)}{\ysi}{(p,r)}$, completing the proof.
\end{proof}

Having a tester in the form of a test purpose $\yT$ and an IUT $\yI$, we now need to say  
when a test run is successful with respect to a given specification $\yS$.
Recalling that a test run was formalized by the cross product $\yT\times \yI$, and that $\yT$ signals an unsuccessful run when it reaches a {\bf fail} state, the  following definition formalizes the verdict of a test run~\cite{tret-model-2008}.
Importantly, given a specification $\yS$, when the test run is successful  we need a guarantee that $\yI$ does  {\bf ioco}-conform to  $\yS$ and, conversely, that the test run surely fails when $\yI$ does not {\bf ioco}-conform to $\yS$, for any implementation $\yI$.
That is, we need  properties of soundness and exhaustiveness.
It is also customary to specify that  IUTs of interest are only taken from particular subclasses of models in $\yioc{I}{U}$.
\begin{defi}\label{def:passes}
	Let $\yI=\yio{S_\yI}{q_0}{I}{U}{T_\yI}\in\yioc{I}{U}$ be an IUT and let $\yT=\yio{S_\yT}{t_0}{U}{I}{T_\yT}\in\yioc{U}{I}$ be a test purpose.
	We say that $\yI$ \emph{passes} $\yT$
	if, for any  $\ysi\in (L_I\cup L_U)^\star$ and any state $q\in S_\yI$, we do not have $\ytrt{(t_0,q_0)}{\ysi}{(\yfail,q)}$ in $\yT\times \yI$.
	Let $TP$ be a fault model. We say that $\yI$ \emph{passes} $TP$ if $\yI$ passes all test purposes in $TP$.
	Let $\yS$ be an IOLTS, and let  $\yltsn{IMP}\ysse\yioc{I}{U}$ be a family of IOLTSs.
	We say that 
	$TP$  is \emph{{\bf ioco}-complete} for  $\yS$ relatively to the $\yltsn{IMP}$ class if, for all $\yI\in \yltsn{IMP}$, we have $\yI \yioco \yS$  if and only if $\yI$ passes $TP$.
\end{defi}
When $\yltsn{IMP}=\yioc{I}{U}$, the class of all IOLTSs over $L_I$ and $L_U$, we can also say that $TP$ is  {\bf ioco}-complete for  $\yS$,
instead of $TP$ is  {\bf ioco}-complete for  $\yS$ relatively to the full $\yioc{I}{U}$ class.

The following construction gives us a fault model comprised by a single test purpose, and which is complete for any given specification IOLTS $\yS$.
\begin{lemm}\label{lemm:ioco-complete}
	Let $\yS=\yio{S_\yS}{s_0}{I}{U}{R_\yS}\in\yioc{I}{U}$.
	We can effectively construct a fault model $TP=\{\yT\}$  which is {\bf ioco}-complete for $\yS$.
	Moreover, $\yT$ is deterministic  and  has a single \yfail\ and no \ypass\ states.
\end{lemm}
\begin{proof}
	We will devise a fault model $TP$ that is a singleton, that is, $TP=\{\yT\}$ where $\yT=\yio{S_\yT}{t_0}{U}{I}{R_\yT}$ is some test purpose to be constructed.
	In order for $TP$ to be {\bf ioco}-complete for $\yS$ we need that, for all implementations $\yI$, it holds that $\yI$ passes $\yT$ if and only if $\yI \yioco \yS$.
	From Lemma~\ref{lemm:ioco-reg} we know that $\yI \yioco \yS$ if and only if $\yI \yconf{D}{\yemp} \yS$, with $D=otr(\yS) L_U$.
	Recall that $\ycomp{otr}(\yS)$ indicates the complement of $otr(\yS)$, that is,  $\ycomp{otr}(\yS)=L^\star-otr(\yS)$.
	From  Proposition~\ref{prop:equiv-conf} we get $\yI \yconf{D}{\yemp} \yS$ if and only if  $otr(\yI)\cap T=\yemp$, where $T=\ycomp{otr}(\yS)\cap (otr(\yS) L_U)$.
	Putting it together, we see that we need a test purpose $\yT$ such that, for all implementations $\yI$, we have $\yI$ passes $\yT$ if and only if $otr(\yI)\cap T=\yemp$.
	
	Let $\yI=\yio{S_\yI}{q_0}{I}{U}{R_\yI}\in\yioc{I}{U}$ be an arbitrary implementation. 
	We have that $\yI$ passes $\yT$ if and only if in the cross-product LTS $\yT\times \yI$ we never have $\ytrt{(s_0,q_0)}{\ysi}{(\yfail,q)}$, for any $\ysi\in L^\star$ and any $q\in S_\yI$.
	Since we want $\yI$ passes $\yT$ if and only if $otr(\yI)\cap T=\yemp$,  
	we conclude that we need $otr(\yI)\cap T\neq \yemp$ if and only if  $\ytrt{(t_0,q_0)}{\ysi}{(\yfail,q)}$ in $\yT\times \yI$ for some $\ysi\in L^\star$ and some $q\in S_\yI$.
	
	We start by using Proposition~\ref{prop:deterministic-lts} and the underlying LTS of $\yS$ to effectively construct a deterministic LTS $\yB'$ such that $otr(\yS)=tr(\yB')=otr(\yB')$.   
	Let $\yB$ be the  IOLTS whose underlying LTS is $\yB'$.
	Clearly, $\yB$ can be effectively constructed, and it is deterministic with $otr(\yS)=tr(\yB)=otr(\yB)$.
	Let $\yB=\yio{S_\yB}{b_0}{I}{U}{R_\yB}$.
	We use $\yB$ to construct the desired test purpose $\yT=\yio{S_\yT}{t_0}{U}{I}{R_\yT}$ by extending the state set $S_\yB$ and the transition set $R_\yB$ as follows:
	we  add a $\yfail$ state to $S_\yB$ and add a transition $(s,\ell,\yfail)$ to $R_\yT$ whenever state 
	$s$ has no outgoing transition labeled  $\ell$ in $R_\yB$, where $\ell\in L_U $ is an output symbol.
	More precisely,  we construct $\yT$ by defining $t_0=b_0$, $S_\yT=S_\yB\cup \{\yfail\}$ where $\yfail\not\in S_\yB$, and letting 
	\begin{equation}\label{eq:testpurpose}
		\begin{split}
			R_\yT &=  R_\yB \cup \big\{(s,\ell,\yfail)\yst \ell\in L_U \text{ and } (s,\ell,p)\not\in R_\yB \\
			&  \text{ for any } p\in S_\yB\big\}\cup  \big\{(\yfail,\ell,\yfail)\yst \ell\in L_U\big\}.
		\end{split}
	\end{equation}
	Clearly, $\yT$ is a deterministic IOLTS with a single \yfail\ state and no \ypass\ states.
	In order to complete the proof we have to establish that $otr(\yI)\cap T\neq \yemp$ if and only if  $\ytrt{(t_0,q_0)}{\ysi}{(\yfail,q)}$ in $\yT\times \yI$ for some $\ysi\in L^\star$ and some $q\in S_\yI$.
	We start with the following claim.
	\begin{description}
		\item[\hspace*{0.3ex}]{\sc Claim:} 
		
		(i) If $t\neq \yfail$, then $\ytrt{t_0}{\ysi}{t}$ in $\yT$ if and only if 
		$\ytrt{b_0}{\ysi}{t}$ in $\yB$, for all $\ysi\in L^\star$.
		
		(ii) Let $p\neq \yfail$ with $p\in S_\yB$, and let $\ysi\in L^\star$, $\ell\in L$.
		Then, $\ytrt{t_0}{\ysi}{p}\ytrt{}{\ell}{\yfail}$ in $\yT$ if and only if $\ell \in L_U$, $\ytrt{b_0}{\ysi}{p}$ in $\yB$ and $(p,\ell,q)\not\in R_\yB$ for any $q\in S_\yB$.
	\end{description}
	\begin{description}
		\item[\hspace*{0.3ex}]{\sc Proof of the  Claim:} 
		Since $\yT$ is deterministic, it has no $\tau$-labeled transitions. Hence,  $\ytrt{t_0}{\ysi}{t}$ if and only if $\ytr{t_0}{\ysi}{t}$. 
		An easy induction on $|\ysi|\geq 0$, with $t\neq \yfail$, gives that $\ytr{t_0}{\ysi}{t}$ in $\yT$ if and only if $\ytr{b_0}{\ysi}{t}$ in $\yB$.
		Since $\yB$ is also deterministic, we immediately get $\ytr{b_0}{\ysi}{t}$ if and only if $\ytrt{b_0}{\ysi}{t}$ in $\yB$, thus establishing item (i).
		
		Using item (i) we get $\ytrt{t_0}{\ysi}{p}$ in $\yT$ if and only if $\ytrt{b_0}{\ysi}{p}$ in $\yB$.
		Since $\yT$ has no $\tau$-moves, we have $\ytrt{p}{\ell}{\yfail}$ in $\yT$ if and only if $\ytr{p}{\ell}{\yfail}$ in $\yT$.
		Because $p\neq \yfail$, the construction of $\yT$ we have $(p,\ell,\yfail)$ in $R_\yT$ if and only if 
		$\ell\in L_U$ and $(p,\ell,q)\not\in R_\yB$ for any $q\in S_\yB$.
		This verifies item (ii), and completes the proof of the claim.  
	\end{description} 
	
	Assume that $\ytrt{(t_0,q_0)}{\ysi}{(\yfail,q)}$ in $\yT\times \yI$ for some $\ysi\in L^\star$ and some $q\in S_\yI$.
	Using Proposition~\ref{prop:cross} we obtain $\ytrt{t_0}{\ysi}{\yfail}$ in $\yT$  and $\ytrt{q_0}{\ysi}{q}$ in $\yI$, so that $\ysi\in otr(\yI)$.
	Because $t_0\neq \yfail$ we get some $p\neq \yfail$ and some $\mu\in L^\star$, $\ell\in L$ such that $\ytrt{t_0}{\mu}{p}\ytrt{}{\ell}{\yfail}$ in $\yT$, and $\ysi=\mu\ell\rho$ for some $\rho\in L^\star$.
	The Claim gives $\ell\in L_U$, $\ytrt{b_0}{\mu}{p}$ in $\yB$, and $(p,\ell,q)\not\in R_\yB$, for all $q\in S_\yB$.
	Hence, $\mu\in otr(\yB)$ and, since we already have $otr(\yS)=otr(\yB)$, we get $\mu\in otr(\yS)$.
	Because $\ell\in L_U$ we conclude that $\mu\ell\in otr(\yS) L_U$.
	Because $\ysi\in otr(\yI)$ and $\ysi=\mu\ell\rho$, we also get $\mu\ell\in otr(\yI)$.
	This gives $\mu\ell\in otr(\yI) \cap (otr(\yS) L_U)$.
	Thus we must have  $\mu\ell\not\in otr(\yS)$, otherwise we would also have $\mu\ell\in otr(\yB)$, since $otr(\yS)=otr(\yB)$. 
	This would give $\ytrt{b_0}{\mu}{r}\ytrt{}{\ell}{q}$ in $\yB$, for some $r$, $q\in S_\yB$.
	Now, since $\yB$ is deterministic and we already have $\ytrt{b_0}{\mu}{p}$ in $\yB$, we conclude that $r=p$, so that $\ytrt{p}{\ell}{q}$ in $\yB$.
	Since $\yB$ is deterministic, $\ytrt{p}{\ell}{q}$ implies $\ytr{p}{\ell}{q}$ in $\yB$, which is a contradiction.
	So, $\mu\ell\not\in otr(\yS)$, that is, $\mu\ell\in \ycomp{otr}(\yS)$.
	We can now write $\mu\ell\in otr(\yI)\cap \big[\,\ycomp{otr}(\yS)\cap (otr(\yS) L_U)\big]$, that is,  $\ysi\in otr(\yI)\cap T$, showing that $otr(\yI)\cap T\neq \yemp$, as desired.
	
	For the other direction, assume that $otr(\yI)\cap T\neq \yemp$.
	We then get some $\ysi\in L^\star$ such that $\ysi\in otr(\yI)$, $\ysi\not\in otr(\yS)$ and $\ysi\in otr(\yS) L_U$.
	Then $\ysi\in otr(\yI)$ gives $\ytrt{q_0}{\ysi}{q}$ in $\yI$, for some $q\in S_\yI$.
	Also $\ysi\in otr(\yS) L_U$ implies that $\ysi=\yal\ell$ for some $\ell\in L_U$ and $\yal\in otr(\yS)$.
	Then, $\yal\in otr(\yB)$ because $otr(\yS)=otr(\yB)$, and then $\ytrt{b_0}{\yal}{b}$ in $\yB$ for some $b\neq \yfail$. 
	If $(b,\ell,q)\in R_\yB$ for some $q\in S_\yB$, then $\ytrt{b_0}{\yal\ell}{q}$ in $\yB$ because we already have $\ytrt{b_0}{\yal}{b}$ in $\yB$.
	This would give $\yal\ell\in otr(\yB)$, and so $\ysi\in otr(\yS)$ because $\ysi=\yal\ell$.
	But this is a contradiction since we already have $\ysi\not\in otr(\yS)$.
	Hence, $(b,\ell,q)\not\in R_\yB$ for all $q\in S_\yB$.
	The Claim, item (ii), then gives $\ytrt{t_0}{\ysi}{\yfail}$ in $\yT$.
	Using Proposition~\ref{prop:cross} we obtain $\ytrt{(t_0,q_0)}{\ysi}{(\yfail,q)}$ with $\ysi\in L^\star$, establishing the reverse direction.
	
	We have reached the desired conclusion, namely, $otr(\yI)\cap T\neq \yemp$ if and only if  $\ytrt{(t_0,q_0)}{\ysi}{(\yfail,q)}$ in $\yT\times \yI$ for some $\ysi\in L^\star$ and some $q\in S_\yI$.
	The proof is, thus, complete.
\end{proof}

We illustrate the construction of $\yT$ in Lemma~\ref{lemm:ioco-complete} using as specification the IOLTS in Figure~\ref{fig:spec}. 
\begin{exam}
	Let the IOLTS $\yS$ of Figure~\ref{fig:spec} be the specification. 
	Recall that $L_I=\{a,b\}$ and $L_U=\{x\}$.
	Since $\yS$ is already deterministic we have $\yS=\yB$ at Lemma~\ref{lemm:ioco-complete}, and so we start the construction of  $\yT$ at Eq. (\ref{eq:testpurpose}). 
	We add transitions $(t_0,x,\yfail)$, $(t_2,x,\yfail)$, $(t_3,x,\yfail)$, and $(\yfail,x,\yfail)$ to $R_\yB$ in order to get $R_\yT$, as shown in Figure~\ref{fig:tp-spec}.
	Notice that, among others, the observable behavior $\yal=axbx$ leads to the \yfail\ state in $\yT$, that is $\ytrt{t_0}{\yal}{\yfail}$ in $\yT$.
	Now let the IOLTS $\yI$ of Figure~\ref{fig:impl} be the IUT.
	We now have $\ytrt{q_0}{\yal}{q_2}$ in $\yI$.
	Hence, we get $\ytrt{(t_0,q_0)}{\yal}{(\yfail,q_2)}$ in $\yT\times \yI$ and so
	$\yI$ does not pass $\{\yT\}$.
	According to Lemma~\ref{lemm:ioco-complete}, $\{\yT\}$ is a complete fault model for the specification $\yS$.
	Therefore,  $\yI \yioco \yS$ should not hold too as we shall check.
	Take the  behavior $\ybe=ba$.
	We have $\ybe\in otr(\yS)$ because $\ytrt{s_0}{\ybe}{s_3}$ in $\yS$, but $x\not\in out(s_0 \yafter \ybe)$ because $x\not\in \yout(s_3)$ and $\yS$ is deterministic.
	We also have $\ytrt{q_0}{\ybe}{q_3}$ in $\yI$, and there is a transition $(q_3,x,q_2)$ in $\yI$.
	Then $x\in \yout(q_0 \yafter \ybe)$,  and so
	$out(q_0 \yafter \ybe)\not\ysse out(s_0 \yafter \ybe)$. 
	According to Definition~\ref{def:out-after},  $\yI \yioco \yS$ does not hold, as  expected. \yfim
\end{exam}

\subsection{A Specific Family of  Formal Models}\label{subsec:ioco-test-cases}

A number of restrictions are imposed by Tretmans~\cite{tret-model-2008} on the structure of the formal models in his approach, so that test runs can be adjusted to more practical situations.
First, from the perspective of the testers one would like test purposes to be acyclic, so that according to Definition~\ref{prop:cross} we have a guarantee that any test run is a finite process that runs to completion.
Secondly, since the tester cannot predict in advance which symbols a black-box IUT will be sending back in an exchange, it is convenient for the tester, at any state,  to be able to respond to any symbol that the IUT can possibly exchange with it. 
In this case, we say that the tester is input-enabled.  
Further, since a tester drives the IUT, in order to guarantee that  test runs continue to completion, at any state the tester must be able to drive the run by sending at least one symbol to the IUT.
Also, in order to avoid arbitrary choices and non-determinism, the tester is required to be  output-deterministic, that is, at any state it can emit only one of its output symbols back to the IUT. 
Moreover,  because {\bf fail} and {\bf pass} states already hold a verdict, in order to avoid moving out of these states, we it is required that a tester can only have self-loops at these special states.
The next definition precisely frames the notion of input-enabledness and output-determinism.
Recall Definition~\ref{def:out-after}.
\begin{defi}\label{def:in-compl-out-determ}\label{def:inic-tret}
	Let $\yS\in\yioc{I}{U}$.
	We say that $\yS$ is \emph{output-deterministic} if $|\yout(s)|=1$, for all $s\in S$, and that $\yS$ is \emph{input-enabled} if $\yinp(s)=L_I$, for all $s\in S$, where 
	the function $\yinp$ is defined from $\ypow{S}$ into $L_I$ as $ \yinp(V)= \bigcup\limits_{s\in V}\{\ell\in L_I\yst \ytrt{s}{\ell}{}\}$. 
	The class of all input-enabled IOLTSs over the alphabets $L_I$ and $L_U$ will be designated by $\yioce{I}{U}$.  
\end{defi}

We can also get {\bf ioco}-complete fault models whose test purposes are  input-enabled.
\begin{coro}\label{coro:input-enabled}
	For any specification $\yS$ we can effectively construct an {\bf ioco}-complete fault model $\{\yT\}$ where $\yT$ is deterministic, input-enabled  and  has a single \yfail\ and no \ypass\ states.
\end{coro}
\begin{proof}
	A simple observation reveals that the test purpose $\yT$ constructed in the proof of Lemma~\ref{lemm:ioco-complete} is already input-enabled.
\end{proof}

When we have testers that are input-enabled, output-deterministic and acyclic except for self-loops at {\bf fail} and {\bf pass} states, we have encountered all the restrictions imposed by Tretmans~\cite{tret-model-2008} to the testing architecture. 
In the remainder of this section we examine some consequences of these choices.
As a first consequence, it is no surprise that by requiring testers to be acyclic one has to impose some control on the size of the implementations.
\begin{prop}\label{prop:always-cycles}
	Consider the simple deterministic specification IOLTS $\yS=(\{s_0\},s_0,\{a\},\{x\},\{(s_o,a,s_0)\})$.
	Then there is no fault model $TP$, comprised only of acyclic test purposes, and which is {\bf ioco}-complete for $\yS$  even with respect to the subclass of all deterministic IOLTSs.
\end{prop}
\begin{proof}
	Clearly, $\yS$ is deterministic.
	Let $\yal=a^n$, with $n\geq 0$, so that
	$\yal x\in\ycomp{otr}(\yS)\cap(otr(\yS)\ L_U)$.
	Let  $\yI_n$ be described by the transitions $(q_{i-1},a,q_i)$ ($1\leq i\leq n)$,  $(q_n,x,q)$, and the self-loops $(q_n,a,q_n)$ and $(q,a,q)$, with $q_0$ the initial state.
	Clearly, $\yI_n$ is deterministic.
	Incidentally, since $\yinp(s_0)=\{a\}$ Definition~\ref{def:inic-tret}  says that $\yS$ is also input-enabled.
	We have $\yal x\in otr(\yI_n)$, so that
	$otr(\yI_n)\cap \big[\ycomp{otr}(\yS)\cap otr(\yS)\ L_U\big]\neq\yemp.$
	If $TP$ is {\bf ioco}-complete with all its test purposes acyclic, define
	$k=\max_{\yT\in TP}\big\{|\ysi|\yst \ysi\in otr(\yT)\big\}$.
	Since  $otr(\yI_k)\cap \big[\ycomp{otr}(\yS)\cap otr(\yS)\ L_U\big]\neq\yemp$,
	Proposition~\ref{prop:equiv-conf} and Lemma~\ref{lemm:ioco-reg} say that $\yI_k\yioco \yS$ does not hold.
	Hence, $\yI_k$ does not pass $TP$ and  there is $\yT\in TP$ with $\ytrt{(t_0,q_0)}{\ysi}{(\yfail,q)}$ in $\yT\times \yI_k$, for some $\ysi\in\{a,x\}^\star$.
	By Proposition~\ref{prop:cross}, $\ytrt{t_0}{\ysi}{\yfail}$ in $\yT$ and $\ytrt{q_0}{\ysi}{q}$ in $\yI_k$.
	By the maximality of $k$, $\vert\ysi\vert=m$ with $m\leq k$.
	Because $\ytrt{q_0}{\ysi}{q}$ in $\yI_k$, we get $\ysi=a^m$.
	Then, $\ytrt{t_0}{a^m}{\yfail}$ in $\yT$.
	Since, clearly, $\ytrt{s_0}{a^m}{s_0}$ in $\yS$, Proposition~\ref{prop:cross} gives
	$\ytrt{(t_0,s_0)}{a^m}{(\yfail,s_0)}$ in $\yT\times \yS$.
	Hence, $\yS$ does not pass $TP$.
	Taking $\yS$ as an IUT itself, and given that $TP$ was assumed an {\bf ioco}-complete fault model for $\yS$, we have that $\yS\yioco \yS$ does not hold, according to Definition~\ref{def:passes}.
	But this is a clear contradiction according to Definition~\ref{def:out-after}.
\end{proof}

Clearly, in the construction at the proof of Proposition~\ref{prop:always-cycles} we could have taken any IOLTS  $\yS$ with a state $s$ that is reachable from its initial state, and with a self-loop on any input symbol at $s$.
This  makes Proposition~\ref{prop:always-cycles} much more widely applicable.

Next, we investigate the situation where one does have an upper bound on the number of states in the IUTs, and show that the situation is more amenable in these cases.
\begin{defi}\label{def:m-complete}
	Let $\yltsn{IMP}\ysse \yioc{I}{U}$ be any class of IOLTSs, and let $m\geq 1$.
	We denote by $\yltsn{IMP}[m]$ the subfamily of $\yltsn{IMP}$ comprised by all models with at most $m$ states.
	Let $\yS\in \yioc{I}{U}$.
	We say that a fault model $TP$ over $L_I\cup L_U$ is $m$-{\bf ioco}-complete for  $\yS$
	relatively to $\yltsn{IMP}$ if and only if it is {\bf ioco}-complete for $\yS$ relatively to the class $\yltsn{IMP}[m]$.
\end{defi} 

Since it is not possible to construct fault models comprised only of  acyclic test purposes, and that are {\bf ioco}-complete in general, we turn to the problem of obtaining such fault models that are $m$-{\bf ioco}-complete, for a given $m$.
The construction is discussed  in Proposition~\ref{prop:m-complete}, and is illustrated in Example~\ref{exam:acyclic-TP}.
\begin{prop}\label{prop:m-complete}
	Let $\yS=\yioS\in\yioc{I}{U}$ be deterministic, and let $m\geq 1$.
	Then, there is a fault model $TP$ which is $m$-{\bf ioco}-complete for $\yS$, and such that all test purposes in $TP$ are deterministic and acyclic.
\end{prop}
\begin{proof}
	Let $s_0,s_1,\ldots,s_{n-1}$ be an enumeration of  $S$. 
	We construct a direct acyclic multi-graph $D$.
	Let $D$ have $mn+1$ levels, where at each level $i$, $0\leq i\leq mn$, we list all nodes in $S$, in the given order, from left to right, labeling them $s_{0,i}$, $s_{1,i}$,\ldots, $s_{n-1,i}$.
	Consider a node $s_{j,k}$, $0\leq j\leq n$, at level $k$, $0\leq k<mn$.
	For each transition $(s_j,\ell,s_i)$ of $\yS$:
	(i) if $i>j$, add a horizontal left to right arc to $D$ from node $s_{j,k}$ to node $s_{i,k}$;
	and (ii) if $i\leq j$ add a vertical arc to $D$ from node $s_{j,k}$ to node $s_{i,k+1}$. 
	Label the new arc with the symbol $\ell$.
	Complete the construction  by adding  an extra node labeled $\yfail$ to $D$.
	For any node $s_{i,k}$ in $D$, $s_{i,k}\neq \yfail$, and any $\ell\in L_U$, if $(s_i,\ell,p)\not\in T$, for every $p\in S$, add an arc from node $s_{i,k}$ to node $\yfail$ in $D$, 
	and label it $\ell$.
	Let $s_{0,0}$ be the root node.
	Finally, discard any node of $D$ not reachable from the root.
	Since all arcs in $D$ are directed top-down  or from left to right, it is clear that $D$ is acyclic.
	Hence, there is a finite number of distinct maximal paths starting at the root node of $D$.
	Also, since  at any node $s_{i,k}$ we add arcs to $D$ corresponding to all transitions in $\yS$ from node $s_i$, it is clear that, for all $\ysi\in L^\star$, with $|\ysi|\leq mn$, we have that $\ysi\in otr(\yS)$ if and only if $\ysi$ is a path from the root of $D$. 
	
	For each path $\ytr{s_{i_0}}{x_1}{s_{i_1}}\ytr{}{x_2}{}\cdots \ytr{}{x_{i_r}}{s_{i_r}}$  ($r\geq 1$ and $0\leq i \leq n-1$) in $D$ where $s_{i_0}=s_{0,0}$ is the root node and $s_{i_r}=\yfail$, we add the acyclic test purpose
	$\yT=(\{s_{i_0},\ldots, s_{i_r}\},s_{i_0},L_U,L_I,\{(s_{i_{j-1}},x_j,s_{i_j}\yst 1\leq j\leq r\})$ to $TP$.
	Clearly, $\yT$ is deterministic.
	Now we argue that $TP$ is $m$-{\bf ioco}-complete.
	Assume that $\yI\yioco \yS$ does not hold, where
	$\yI=(Q,q_0,L_I,L_U,R)$ is an IUT with $|Q|=h\leq m$ states.
	By Proposition~\ref{prop:equiv-conf} and Lemma~\ref{lemm:ioco-reg} we get $\ysi\in L^\star$ such that 
	$\ysi\in otr(\yI)$, $\ysi\not\in otr(\yS)$ and $\ysi\in otr(\yS) L_U$.
	Let $|\ysi|$ be minimum.
	Clearly, $\ysi=\yal \ell$, with $\ell\in L_U$ and $\yal\in otr(\yS)$ so that $\ytrt{s_0}{\yal}{s}$.
	From $\yal\ell\in otr(\yI)$ we get $\ytrt{q_0}{\yal}{q}\ytrt{}{\ell}{q'}$.
	By Proposition~\ref{prop:cross}, we have $\ytrt{(s_0,q_0)}{\yal}{(s,q)}$ in $\yS\times \yI$.
	With no loss of generality, let $\yal=x_1\ldots x_r$ ($r\geq 0$).
	Then, $\ytrt{(s_0,q_0)}{x_1}{(s_1,q_1)}\ytrt{}{x_2}{(s_2,q_2)}\ytrt{}{x_3}{\ldots}\ytrt{}{x_r}{(s_r,q_r)},$ with $s=s_r$ and $q=q_r$.
	If $r\geq mn\geq hn$ we get $(s_i,q_i)=(s_j,q_j)$ for some $0\leq i<j\leq r$.
	Then, $\ytrt{(s_0,q_0)}{\mu}{(s_r,q_r)}$  with $\mu=x_1\ldots x_ix_{j+1}\ldots x_r$.
	Again, $\ytrt{s_0}{\mu}{s_r}$ and $\ytrt{q_0}{\mu}{q_r}$, which  gives $\mu\in otr(\yS)$ and $\mu\in otr(\yI)$.
	Moreover, $(s_r,\ell,p)\not\in T$ for all $p\in S$, otherwise we get $\ysi=\yal \ell\in otr(\yS)$, a contradiction.
	Hence, $\mu\ell\not\in otr(\yS)$.
	Also, since $\mu\in otr(\yS)$, we get $\mu\ell\in otr(\yS)L_U$.
	Further, we also have $\ytrt{q_0}{\mu}{q_r}$ in $\yI$, $\ytrt{q_0}{\yal}{q}\ytrt{}{\ell}{q'}$, and $q=q_r$.
	Hence, $\ytrt{q_0}{\mu}{q}\ytrt{}{\ell}{q'}$, and so $\mu\ell\in otr(\yI)$.
	Thus, $\mu\ell\in otr(\yI)\cap \big[\ycomp{otr}(\yS)\cap otr(\yS) L_U\big]$.
	But $|\mu|<|\yal|$ and so $|\mu\ell|<|\yal\ell|=|\ysi|$, violating the minimality of $|\ysi|$.
	Hence, $r < mn$.
	Thus, from $\ytrt{s_0}{\yal}{s_r}$  we get a trace $\yal$ from the root $s_{0,0}$ to a node  $s_{i,k}$ in $D$.
	Because $(s_r,\ell,p)$ is not in $T$ for any $\ell\in L_U$ and any $p\in S$,  the construction gives an arc $\ell$ from $s_{i,k}$ to $\yfail$  in $D$.
	Hence, there is a trace $\yal\ell$ from $s_{0,0}$  to $\yfail$ in $D$.
	According to the construction, we get $\yT=(S_\yT,t_0,L_U,L_I,T_\yT)$ in $TP$ with $\ytrt{t_0}{\yal\ell}{\yfail}$ in $\yT$, and so $\ytrt{(t_0,q_0)}{\ysi}{(\yfail,q')}$ in $\yT\times \yI$, showing that $\yI$ does not pass $TP$.
	We have shown that if any IUT $\yI$ with at most $m$ states passes $TP$ then $\yI \yioco \yS$ holds.
	
	If an IUT $\yI=(Q,q_0,L_I,L_U,R)$ does not pass $TP$, we have $\ytrt{(t_0,q_0)}{\ysi}{(\yfail,q)}$ in $\yT\times \yI$, for some  $\yT=(S_\yT,s_{0,0},L_U,L_I,T_\yT)$ in $TP$.
	So, $\ytrt{s_{0,0}}{\ysi}{\yfail}$ in $\yT$ and $\ytrt{q_0}{\ysi}{q}$ in $\yI$.
	Thus, $\ysi\in otr(\yI)$.
	By the construction,  $\ytrt{s_{0,0}}{\yal}{s_{i,k}}\ytrt{}{\ell}{\yfail}$ in $\yT$, and $(s_{i},\ell,p)$ is not a transition in $\yS$, for all $p\in S$ and all $\ell\in L_U$.
	From $\ytrt{s_{0,0}}{\yal}{s_{i,k}}$ in $\yT$, we have a path $\yal$ from the root  to node $s_{i,k}$ in $D$.
	The construction of $D$ then gives, $\ytrt{s_0}{\yal}{s_i}$ in $\yS$, that is,
	we have  $\ysi=\yal\ell\in otr(\yS) L_U$.
	If $\yal\ell\in otr(\yS)$ we would have $\ytrt{s_0}{\yal}{s'}\ytrt{}{\ell}{s''}$ in $\yS$.
	Since $\yS$ is deterministic, we would get $s_i=s'$, and so $(s_i,\ell,s'')$ is a transition in $\yS$, a contradiction.
	Thus, $\ysi=\yal\ell\in \ycomp{otr}(\yS)$.
	Whence, $\ysi\in otr(\yI)\cap \ycomp{otr}(\yS)\cap (otr(\yS) L_U),$
	and, by Lemma~\ref{lemm:ioco-reg} and Proposition~\ref{prop:equiv-conf}, $\yI\yioco \yS$ does not hold.
	We have shown that if $\yI \yioco \yS$ holds then any IUT passes $TP$, and so does any IUT with at most $m$ states.
\end{proof}

Example~\ref{exam:acyclic-TP} illustrates the construction.
\begin{exam}\label{exam:acyclic-TP}
	Let the IOLTS $\yS$ of Figure~\ref{fig:spec} be the specification again. 
	Here we take $L_I=\{a,b\}$ and $L_U=\{x\}$.
	Following the construction exhibited at the Proof of Proposition~\ref{prop:m-complete},
	we construct the direct acyclic multi-graph $D$ 
	partially depicted in Figure~\ref{fig:tp-spec-nocycle}.
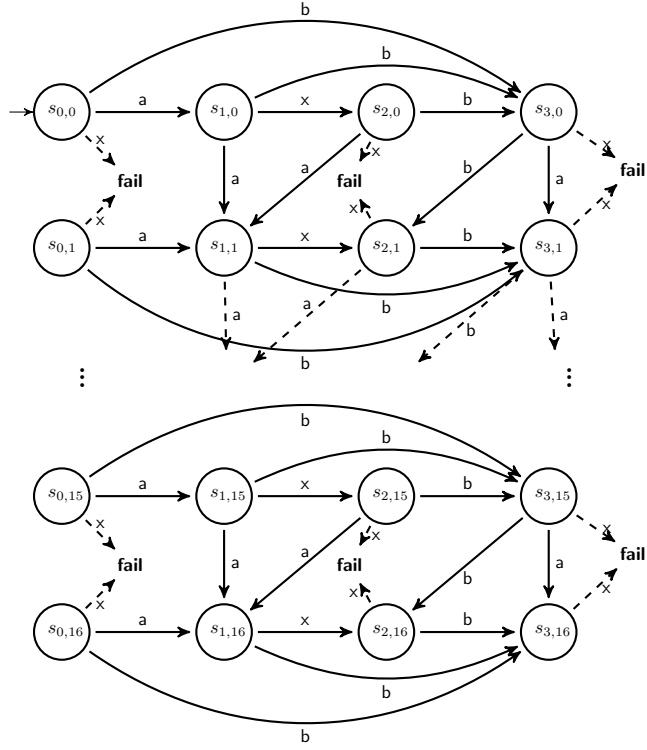
\begin{figure}[tb]
	\center
	\hspace*{-3ex}
	\begin{tikzpicture}[font=\sffamily,node distance=1cm, auto,
	scale=0.7,transform shape] 
	
	\node[ initial by arrow, initial text={}, punkt] (s00) {$s_{0,0}$};
	\node[punkt, inner sep=3pt,right=2cm  of s00] (s10) {$s_{1,0}$};
	\node[punkt, inner sep=3pt,right=2cm  of s10] (s20) {$s_{2,0}$};
	\node[punkt, inner sep=3pt,right=2cm  of s20] (s30) {$s_{3,0}$};

	\node[punkt, inner sep=3pt,below=1.5cm of s00] (s01) {$s_{0,1}$};
	\node[punkt, inner sep=3pt,right=2cm  of s01] (s11) {$s_{1,1}$};
	\node[punkt, inner sep=3pt,right=2cm  of s11] (s21) {$s_{2,1}$};
	\node[punkt, inner sep=3pt,right=2cm  of s21] (s31) {$s_{3,1}$};

	\node[punkt,inv,inner sep=3pt,below=1.5cm of s01] (s02) {};
	\node[punkt, inv,inner sep=3pt,right=2cm  of s02] (s12) {};
	\node[punkt, inv,inner sep=3pt,right=2cm  of s12] (s22) {};
	\node[punkt, inv,inner sep=3pt,right=2cm  of s22] (s32) {};
	
	\node[punkt,inner sep=3pt,below=1cm of s02] (s015) {$s_{0,15}$};
	\node[punkt, inner sep=3pt,right=2cm  of s015] (s115) {$s_{1,15}$};
	\node[punkt, inner sep=3pt,right=2cm  of s115] (s215) {$s_{2,15}$};
	\node[punkt, inner sep=3pt,right=2cm  of s215] (s315) {$s_{3,15}$};

	\node[punkt,inner sep=3pt,below=1.5cm of s015] (s016) {$s_{0,16}$};
	\node[punkt, inner sep=3pt,right=2cm  of s016] (s116) {$s_{1,16}$};
	\node[punkt, inner sep=3pt,right=2cm  of s116] (s216) {$s_{2,16}$};
	\node[punkt, inner sep=3pt,right=2cm  of s216] (s316) {$s_{3,16}$};
	
	
	\node[punkts, invisible,inner sep=0pt,below right=0.7cm and .7cm of s00] (fail0) {$\yfail$};
	\node[punkts, invisible,inner sep=0pt,below left=0.7cm and 0.1 of s20] (fail2) {$\yfail$};
	\node[punkts, invisible,inner sep=0pt,below right=0.5cm and 1cm of s30] (fail3) {$\yfail$};
	
	\node[punkts, invisible,inner sep=0pt,below right=0.7cm and .7cm of s015] (fail150) {$\yfail$};
	\node[punkts, invisible,inner sep=0pt,below left=0.7cm and 0.1 of s215] (fail152) {$\yfail$};
	\node[punkts, invisible,inner sep=0pt,below right=0.5cm and 1cm of s315] (fail153) {$\yfail$};
	
	
	\path (s00)    edge [pil]   	node[anchor=north,above]{a} (s10);
	\path (s00)    edge [pil,bend left=35]   	node[anchor=south]{b} (s30);
	\path (s10)    edge [pil]   	node[anchor=north, above]{x} (s20);
	\path (s10)    edge [pil]   	node[anchor=west]{a} (s11);
	\path (s10)    edge [pil,bend left=25]   	node[anchor=south]{b} (s30);
	\path (s20)    edge [pil]   	node[anchor=north, above]{b} (s30);
	\path (s20)    edge [pil]   	node[anchor=north, above]{a} (s11);
	\path (s30)    edge [pil]   	node[anchor=west]{a} (s31);
	\path (s30)    edge [pil]   	node[anchor=north, above]{b} (s21);
	
	
	\path (s01)    edge [pil]   	node[anchor=north,above]{a} (s11);
	\path (s01)    edge [pil,bend right=40]   	node[anchor=north]{b} (s31);
	\path (s11)    edge [pil]   	node[anchor=north, above]{x} (s21);
	\path (s11)    edge [pil,dashed]   	node[anchor=west]{a} (s12);
	\path (s11)    edge [pil,bend right=25]   	node[anchor=north]{b} (s31);
	\path (s21)    edge [pil]   	node[anchor=north, above]{b} (s31);
	\path (s21)    edge [pil,dashed]   	node[anchor=south]{a} (s12);
	\path (s31)    edge [pil,dashed]   	node[anchor=west]{a} (s32);
	\path (s31)    edge [pil,dashed]   	node[anchor=west,below]{b} (s22);

	
	\path (s015)    edge [pil]   	node[anchor=north,above]{a} (s115);
	\path (s015)    edge [pil,bend left=35]   	node[anchor=north]{b} (s315);
	\path (s115)    edge [pil]   	node[anchor=north, above]{x} (s215);
	\path (s115)    edge [pil]   	node[anchor=west]{a} (s116);
	\path (s115)    edge [pil,bend left=25]   	node[anchor=south]{b} (s315);
	\path (s215)    edge [pil]   	node[anchor=north, above]{b} (s315);
	\path (s215)    edge [pil]   	node[anchor=south]{a} (s116);
	\path (s315)    edge [pil]   	node[anchor=west]{a} (s316);
	\path (s315)    edge [pil]   	node[anchor=west,below]{b} (s216);

	
	\path (s016)    edge [pil]   	node[anchor=north,above]{a} (s116);
	\path (s016)    edge [pil,bend right=35]   	node[anchor=north]{b} (s316);
	\path (s116)    edge [pil]   	node[anchor=north, above]{x} (s216);
	\path (s116)    edge [pil,bend right=25]   	node[anchor=north]{b} (s316);
	\path (s216)    edge [pil]   	node[anchor=north, above]{b} (s316);
	
	
	\path (s00)    edge [pil,dashed]   	node[anchor=north,above]{x} (fail0);
	\path (s01)    edge [pil,dashed]   	node[anchor=west,below]{x} (fail0);
	\path (s20)    edge [pil,dashed]   	node[anchor=west]{x} (fail2);
	\path (s21)    edge [pil,dashed]   	node[anchor=east]{x} (fail2);
	\path (s30)    edge [pil,dashed]   	node[anchor=west]{x} (fail3);
	\path (s31)    edge [pil,dashed]   	node[anchor=west]{x} (fail3);
	
	\path (s015)    edge [pil,dashed]   	node[anchor=north,above]{x} (fail150);
	\path (s016)    edge [pil,dashed]   	node[anchor=west,below]{x} (fail150);
	\path (s215)    edge [pil,dashed]   	node[anchor=west]{x} (fail152);
	\path (s216)    edge [pil,dashed]   	node[anchor=east]{x} (fail152);
	\path (s315)    edge [pil,dashed]   	node[anchor=west]{x} (fail153);
	\path (s316)    edge [pil,dashed]   	node[anchor=west]{x} (fail153);

	\path (s01)    edge [pil,invisible]   	node[]{\Large\bf $\vdots$} (s015);
	\path (s31)    edge [pil,invisible]   	node[]{\Large\bf $\vdots$} (s315);
	\end{tikzpicture}
	\caption{A direct acyclic multi-graph D for Figure~\ref{fig:spec}.
	}\label{fig:tp-spec-nocycle}
\end{figure}
	As $\yS$ has $n=4$ states we  have four states at each level,  and all transitions connect nodes at the same level from left to right, or  must go to next level. 
	In this example we have considered IUTs with at most $m=n=4$ states, so that $mn+1=17$ levels are present in the multi-graph $D$. 
	The top at  Figure~\ref{fig:tp-spec-nocycle} shows the first two layers, and the bottom at that figure depicts the last two layers of $D$.
	To complete the construction we also add a \textbf{fail} state but, 
	in order to keep the figure uncluttered, we replicated the $\yfail$ label.
	
	A simple algorithm can extract test purposes $\yT$ by traversing the graph from the root $s_{0,0}$ to a $\yfail$ state. 
	For example, take the path $\yal=aaxabbbax$ with $\vert \yal\vert\leq 16$. 
	We can easily check that $\yal$ leads  $s_{0,0}$ to the \yfail\ state, that is $\ytrt{s_{00}}{\yal}{\yfail}$, by the sequence of states $s_{0,0}, s_{1,0}, s_{1,1}, s_{2,1}, s_{1,2}, s_{3,2}, s_{2,3}, s_{3,3}, s_{3,4}, \yfail$.
	Assume the IOLTS $\yI$ of Figure~\ref{fig:impl} as an IUT.
	By inspection,  $\yal$ leads  $q_0$ to $q_2$. 
	So $\yI$ does not pass the test purpose induced by $\alpha$  and so $\yI \yioco \yS$ does not hold, as expected again. \yfim
\end{exam}

A careful observation of the construction exhibited at the proof of Proposition~\ref{prop:m-complete}, with minor adjustments, also reveals that one can get $m$-{\bf ioco}-complete fault models whose test purposes are deterministic, output-deterministic, input-enabled and acyclic, except for self-loops at $\ypass$ and $\yfail$ states.
\begin{prop}\label{prop:in-out}
	Let $\yS\in\yioc{I}{U}$ be deterministic, and let $m\geq 1$.
	Then, there is a fault model $TP$ which is {\bf ioco}-complete for $\yS$ relatively to $\yioc{I}{U}[m]$, and such that all test purposes in $TP$ are deterministic, input-enabled, output-deterministic, and acyclic except for self-loops at
	$\yfail$ and $\ypass$ states.  
\end{prop}
\begin{proof}
	From Proposition~\ref{prop:m-complete} we get a fault model $TP$ that is $m$-{\bf ioco}-complete for $\yS$, and such that all test purposes in $TP$ are acyclic and deterministic.
	Consider any test purpose $\yT=\yio{S_\yT}{t_0}{U}{I}{T_\yT}$ in $TP$, and take any state $t$ of $\yT$.
	
	In order to secure input-enabledness, we add a new $\ypass$ state to $S_\yT$.
	For all $\ell\in L_U$ we proceed as follows.
	If we do not have a transition $(t,\ell,t')$ in $T_\yT$ for any $t'\in S_\yT$ and with $t\neq \yfail$, we add the 
	transition $(t,\ell,\ypass)$ to $T_\yT$.
	Also,  add transitions $(\ypass,\ell,\ypass)$ and $(\yfail,\ell,\yfail)$ to $T_\yT$, for all $\ell\in L_U$.
	This will transform $\yT$ into an input-enabled acyclic IOLTS $\yT'$,
	except for self-loops at $\ypass$ and $\yfail$ states.
	Moreover, an easy induction on $\vert\ysi\vert\geq 0$ can establish that we have $\ytrt{t_0}{\ysi}{\yfail}$ in $\yT$ if and only if we also have $\ytrt{t_0}{\ysi}{\yfail}$ in $\yT'$.
	We conclude that $\yI$ passes $\yT$ if and only if $\yI$ passes $\yT'$.
	Hence, after adjusting all test purposes in $TP$ in this manner we get a fault model $TP'$ that is $m$-{\bf ioco}-complete for $\yS$, and such that all test purposes in $TP$ are input-enabled and acyclic, except for the self-loops at states $\ypass$ and $\yfail$.
	
	We now argue for output-determinism.
	A simple re-examination of the proof of Proposition~\ref{prop:m-complete} reveals that we have at most one outgoing transition $(t,\ell,t')$ in $\yT$, for any $\ell\in L_I\cup L_U$.
	Since in the previous step we added no transitions on a symbol from $L_I$, if we have $\ell\in L_I$ then state $t$ is already output-deterministic.
	If $\ell\not\in L_I$ we choose any symbol $\ell'\in L_I$ and add a transition $(t,\ell',\ypass)$ to $\yT'$.
	This makes state $t$ output-deterministic in $\yT'$.
	After we apply  this transformation to all states in $\yT'$ we get a new test purpose $\yT''$ which is output-deterministic.
	Again, it is easy to see that for any implementation $\yI\in\yioc{I}{U}$, we have that $\yI$ passes $\yT'$ if and only if $\yI$ passes $\yT''$.
	We, thus, reach the conclusion that we can always get a fault model $TP''$ which is $m$-{\bf ioco}-complete for $\yS$, and such that all test purposes in $TP''$ are input-enabled, output-deterministic, and acyclic, except for the self-loops at the $\ypass$ and $\yfail$ states.
	
\end{proof}

In Example~\ref{exam:TP-det-inp-out-acy} we illustrate some test purposes that are extracted from the multi-graph $D$, following Propositions~\ref{prop:m-complete} and~\ref{prop:in-out}.
\begin{exam}\label{exam:TP-det-inp-out-acy}
	Example~\ref{exam:acyclic-TP} shows (part of) the direct acyclic multi-graph $D$ for the IOLTS $\yS$ of Figure~\ref{fig:spec}. 
	Using Propositions~\ref{prop:m-complete} and~\ref{prop:in-out} we can extract deterministic, acyclic, input-enabled and output-deterministic test purposes from $D$.
	First, take the  path $\ysi_1=axbaabbx$. 
	In Figure~\ref{fig:tp-det-inp-out-acy-1}, Proposition~\ref{prop:m-complete} gives the outer  path from $s_{0,0}$ to $\yfail$, giving rise to a deterministic and acyclic test purpose.
	Then, in order to secure input-enabledness of the test purpose we followed the construction at the first part in Proposition~\ref{prop:in-out} we added  a new $\ypass$ state.
	Proceeding, , for all states $s$ in the outer path from $s_{0,0}$ to $s_{3,3}$ we added new transitions $(s,\ell,\ypass)$ for all symbols $\ell\in L_U$ which were not yet in the set $\yout(s)$.
	This part of the construction was completed by adding transitions $(\ypass,\ell,\ypass)$ and $(\yfail,\ell,\yfail)$  for all $\ell\in L_U$, as shown in Figure~\ref{fig:tp-det-inp-out-acy-1}.
	Next, we guaranteed  output-determinism by following the second part of the construction in Proposition~\ref{prop:in-out}.
	We then added transitions $(s,\ell,\ypass)$ from all states $s$ in the outer path from state $s_{0,0}$ to state $s_{3,3}$  and for which we still had $\yinp(s)\cap L_I=\yemp$. 
	By a simple inspection one can check that test purpose depicted in Figure~\ref{fig:tp-det-inp-out-acy-1}  is, in fact,  deterministic, input-enabled, output-deterministic and  acyclic, except by the self-loops at $\yfail$ and $\ypass$ states. 
\begin{figure}[tb]
	\center
	\begin{tikzpicture}[font=\sffamily,node distance=1cm, auto,
	scale=0.6,transform shape] 
	
	\node[ initial by arrow, initial text={}, punkt] (s00) {$s_{0,0}$};
	\node[punkt, inner sep=3pt,right=2cm  of s00] (s10) {$s_{1,0}$};
	\node[punkt, inner sep=3pt,right=2cm  of s10] (s20) {$s_{2,0}$};
	\node[punkt, inner sep=3pt,right=2cm  of s20] (s30) {$s_{3,0}$};
	\node[punkt, inner sep=3pt,below=4cm  of s30] (s31) {$s_{3,1}$};
	\node[punkt, inner sep=3pt,left=2cm  of s31] (s32) {$s_{3,2}$};
	\node[punkt, inner sep=3pt,left=2cm  of s32] (s23) {$s_{2,3}$};
	\node[punkt, inner sep=3pt,left=2cm  of s23] (s33) {$s_{3,3}$};
	
	
	
	\node[punkt, inner sep=3pt,above=1.5cm  of s33] (fail3) {$\yfail$};
	
	\node[punkt, inner sep=3pt,right=3.5cm  of fail3] (pass) {$\ypass$};
	
	
	\path (s00)    edge [pil]   	node[anchor=north,above]{a} (s10);
	\path (s10)    edge [pil]   	node[anchor=north, above]{x} (s20);
	\path (s20)    edge [pil]   	node[anchor=north, above]{b} (s30);
	\path (s30)    edge [pil]   	node[anchor=west]{a} (s31);
	\path (s31)    edge [pil]   	node[anchor=north, above]{a} (s32);
	\path (s32)    edge [pil]   	node[anchor=north, above]{b} (s23);
	\path (s23)    edge [pil]   	node[anchor=north, above]{b} (s33);
	
	\path (s00)    edge [pil]   	node[anchor=south]{x} (pass);
	\path (s10)    edge [pil]   	node[anchor=east]{a} (pass);
	\path (s20)    edge [pil]   	node[anchor=west]{x} (pass);
	\path (s30)    edge [pil]   	node[anchor=west]{x} (pass);
	\path (s31)    edge [pil]   	node[anchor=west]{x} (pass);
	\path (s32)    edge [pil]   	node[anchor=west]{x} (pass);
	\path (s23)    edge [pil]   	node[anchor=west]{x} (pass);
	\path (s33)    edge [pil]   	node[anchor=west]{a} (pass);
	
	
	\path (s33)    edge [pil]   	node[anchor=west]{x} (fail3);

	\path (pass) 
	edge [loop right] node [] {$x$} () ;
	
	\path (fail3) 
	edge [loop left] node [] {$x$} () ;
	\end{tikzpicture}
	\caption{A test purpose extracted from Figure~\ref{fig:tp-spec-nocycle}.
	}\label{fig:tp-det-inp-out-acy-1}
\end{figure}
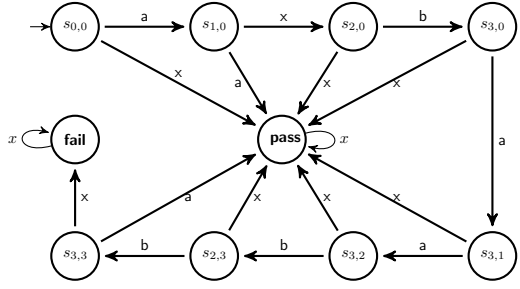
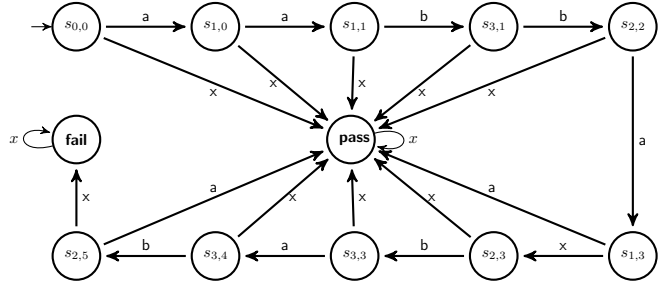
\begin{figure}[tb]
	\center
	\hspace*{-2ex}
	\begin{tikzpicture}[font=\sffamily,node distance=1cm, auto,
	scale=0.6,transform shape] 
	
	\node[ initial by arrow, initial text={}, punkt] (s00) {$s_{0,0}$};
	\node[punkt, inner sep=3pt,right=2cm  of s00] (s10) {$s_{1,0}$};
	\node[punkt, inner sep=3pt,right=2cm  of s10] (s11) {$s_{1,1}$};
	\node[punkt, inner sep=3pt,right=2cm  of s11] (s31) {$s_{3,1}$};
	\node[punkt, inner sep=3pt,right=2cm  of s31] (s22) {$s_{2,2}$};
	\node[punkt, inner sep=3pt,below=4cm  of s22] (s13) {$s_{1,3}$};
	\node[punkt, inner sep=3pt,left=2cm  of s13] (s23) {$s_{2,3}$};
	\node[punkt, inner sep=3pt,left=2cm  of s23] (s33) {$s_{3,3}$};
	\node[punkt, inner sep=3pt,left=2cm  of s33] (s34) {$s_{3,4}$};
	\node[punkt, inner sep=3pt,left=2cm  of s34] (s25) {$s_{2,5}$};

	
	
	\node[punkt, inner sep=3pt,above=1.5cm  of s25] (fail3) {$\yfail$};
	
	\node[punkt, inner sep=3pt,right=5cm  of fail3] (pass) {$\ypass$};
	
	
	\path (s00)    edge [pil]   	node[anchor=south]{a} (s10);
	\path (s10)    edge [pil]   	node[anchor=south]{a} (s11);
	\path (s11)    edge [pil]   	node[anchor=south]{b} (s31);
	\path (s31)    edge [pil]   	node[anchor=south]{b} (s22);
	\path (s22)    edge [pil]   	node[anchor=west]{a} (s13);
	\path (s13)    edge [pil]   	node[anchor=north, above]{x} (s23);
	\path (s23)    edge [pil]   	node[anchor=north, above]{b} (s33);
	\path (s33)    edge [pil]   	node[anchor=north, above]{a} (s34);
	\path (s34)    edge [pil]   	node[anchor=north, above]{b} (s25);
	
	\path (s00)    edge [pil]   	node[anchor=north]{x} (pass);
	\path (s10)    edge [pil]   	node[anchor=east]{x} (pass);
	\path (s11)    edge [pil]   	node[anchor=west]{x} (pass);
	\path (s31)    edge [pil]   	node[anchor=north]{x} (pass);
	\path (s22)    edge [pil]   	node[anchor=north]{x} (pass);
	\path (s13)    edge [pil]   	node[anchor=south]{a} (pass);
	\path (s23)    edge [pil]   	node[anchor=west]{x} (pass);
	\path (s33)    edge [pil]   	node[anchor=west]{x} (pass);
	\path (s34)    edge [pil]   	node[anchor=west]{x} (pass);
	\path (s25)    edge [pil]   	node[anchor=south]{a} (pass);
	
	
	\path (s25)    edge [pil]   	node[anchor=west]{x} (fail3);

	\path (pass) 
	edge [loop right] node [] {$x$} () ;
	
	\path (fail3) 
	edge [loop left] node [] {$x$} () ;
	\end{tikzpicture}
	\caption{A test purpose extracted from Figure~\ref{fig:tp-spec-nocycle}.
	}\label{fig:tp-det-inp-out-acy-2}
\end{figure}
	
	As another illustration, consider the path $\ysi_2=aabbaxbabx$.
	Using the same steps, it results in the test purpose shown at Figure~\ref{fig:tp-det-inp-out-acy-2}. 
	
	Now consider the IOLTS $\yI$ shown at Figure~\ref{fig:impl} as an implementation to be tested.
	First, take the test purpose $\yT_2$, as shown in Figure~\ref{fig:tp-det-inp-out-acy-2}, that was constructed as indicated above.
	We can easily check that $\ysi_2 \notin otr(\yI)$, so now we do not have $\ytrt{(s_{0,0},q_0)}{\ysi_2}{(\yfail,q_2)}$ in $\yT_2\times \yI$.
	So, at this point, running a test experiment using $\yT_2$ and the IUT $\yI$ results in an inconclusive verdict.
	Next, take the same IUT $\yI$ and now take test purpose $\yT_1$, as shown in Figure~\ref{fig:tp-det-inp-out-acy-1}.
	By inspection,  $\ysi_1$ leads  from $q_0$ to $q_2$ in $\yI$, so that we have $\ytrt{(s_{0,0},q_0)}{\ysi_1}{(\yfail,q_2)}$ in $\yT_1\times \yI$. 
	This indicates that $\yI$ does not pass $\yT_1$,  and so now we may conclude that $\yI \yioco \yS$, in fact, does not hold.
	\yfim
\end{exam}

Putting these partial results together we reach the following conclusion, showing that we can use our methods to effectively construct test purposes that satisfy all requirements listed by Tretmans~\cite{tret-model-2008} for a black-box testing  architecture.
\begin{theo}\label{theo:all-reqs}
	Let $\yS\in\yioc{I}{U}$ be a specification, and let $m\geq 1$.
	Then we can effectively construct a finite fault model $TP$ which is $m$-{\bf ioco}-complete for $\yS$, and such that all test purposes in $TP$ are deterministic, input-enabled, output-deterministic, and acyclic except for self-loops at
	special $\yfail$ and $\ypass$ states.
\end{theo}
\begin{proof}
	First, if $\yS$ is not already deterministic, use Proposition~\ref{prop:deterministic-lts}
	to transform $\yS$ into an equivalent deterministic IOLTS.
	Then use Proposition~\ref{prop:in-out} to get the desired fault model.
\end{proof}
It is not hard to see that all our models constructed according to Theorem~\ref{theo:all-reqs} satisfy all restrictions that must be obeyed by test cases as described in~\cite{tret-model-2008}, Definition~10.
A detailed, step by step argument can be seen in the Appendix, specially Proposition~\ref{prop:tps-are-tcs}.

Moreover, if one has a different set of characteristics, stemming  from another kind of testing architecture, and with somewhat different requirements to be satisfied by test purposes as compared to those proposed by Tretmans~\cite{tret-model-2008}, one might try to proceed as discussed here, and transform each basic test purpose so as to make it adhere to that specific set of new requirements.

\subsection{On the Complexity of Test Purposes}\label{sec:tretma-complexity}

We now look at the complexity of the specific family of test purposes constructed in Subsections~\ref{subsec:fault-models} and~\ref{subsec:ioco-test-cases}.
Let $\yS=\yioS$ be a deterministic specification IOLTS, with $|S|=n$. 
First, we return to the construction of the acyclic multi-graph $D$, described in the proof of  Proposition~\ref{prop:m-complete}, and that was used to obtain acyclic test purposes that are $m$-{\bf ioco}-complete for $\yS$.
Since $\yS$ is deterministic, it is clear that $D$ is also deterministic and acyclic.
Moreover, since $D$ has $nm+1$ levels with at most $n$ nodes per level, we conclude that $D$ has at most $n^2m+n$ nodes.

Although the number of nodes and levels in $D$ are polynomial on $n$ and $m$, the number of traces in $D$ might be super-polynomial on $n$ and $m$, in general.
Thus, given that we extract the test purposes from the traces in $D$, the fault model that is so generated, and which is complete for $\yS$, might also be of super-polynomial size on $n$ and $m$.
We argue now that, in general, this situation is unavoidable, even when we restrict specifications and implementations to  the smaller class of input-enabled IOLTS models. 
\begin{theo}\label{theo:ioco-specific-exp}
	Let $m\geq 3$, $L_I=\{0,1\}$, and $L_U=\{x\}$.
	Consider the deterministic, input-enabled specification $\yS=\yioS$ depicted in Figure~\ref{sec3:spec-complexity}, and let $TP$ be a fault model which is $m$-{\bf ioco}-complete for $\yS$, relatively to the class of all deterministic and input-enabled implementations.
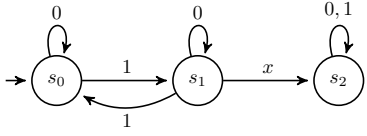
\begin{figure}
	\begin{center}
		\begin{tikzpicture}[->,>=stealth',shorten >=2pt,auto,node distance=2.5cm,
		semithick,scale=0.75,transform shape] 
		\node[state,initial by arrow, initial text={}] (s0) {$s_0$};
		\node[state] (s1) [right of=s0] {$s_1$};
		\node[state] (s2) [right of=s1] {$s_2$};
		\path (s0) edge [] node {$1$} (s1)
		edge [loop above] node [] {$0$} () 
		(s1) edge [] node {$x$} (s2)
		edge [bend left] node [] {$1$} (s0)
		edge [loop above] node [] {$0$} ()
		(s2) edge [loop above] node [] {$0,1$} ()  
		;
		\end{tikzpicture}
		\caption{A specification IOLTS.}\label{sec3:spec-complexity}
	\end{center}
\end{figure}
	If all test purposes in $TP$ are deterministic and output-deterministic, then $TP$ must be comprised of at least $\yomeg{\Phi^m}$ distinct test purposes, where $\Phi=(1+\sqrt{5})/2$.
\end{theo}
\begin{proof}
	Clearly, $\yS$ is  deterministic and input-enabled.
	Let $R=(0+11)^\star$.
	In order to ease the notation, we define  $\overline{\ysi}=1$ when $\ysi=0$, and let  $\overline{\ysi}= 0$ when $\ysi=1$.
	Let $\yal=y_1\ldots y_r\in R$, with $1\leq r\leq m-3$.
	It is clear that $\yal\in otr(\yS)$, $\yal x\in otr(\yS)L_U$, and $\yal x\not\in otr(\yS)$, and so $\yal x\in \ycomp{otr(\yS)}$.
	
	Now let $\yI_\yal=\yio{S_\yI}{q_o}{I}{U}{T_\yI}$ be an IUT given by the transitions $(q_{i-1},y_i,q_i)$ for $1\leq i\leq r$, and $(q_r,x,q_{r+1})$.
	To guarantee that $\yI_\yal$ is input-enabled,
	add a new  $\ypass$ state, and the transitions $(q_{i-1},\overline{y_i},\ypass)$ ($1\leq i\leq r$) to $\yI_\yal$, as well as the transitions  
	$(q_{r},\ysi,\ypass)$ and the self-loops $(q_{r+1},\ysi,q_{r+1})$, $(\ypass,\ysi,\ypass)$, for $\ysi\in\{0,1\}$. See Figure~\ref{sec3:spec-complexity-impl}.
	Clearly, $\yI_\yal$ has $r+3\leq m$ states, is deterministic and input-enabled.
	
	Assume that $TP$ is a fault model which is $m$-{\bf ioco}-complete for $\yS$, relatively to the class of all deterministic and input-enabled implementations.
	Since $\yal x\in otr(\yI)$, we get $otr(\yI_\yal)\cap \big[\ycomp{otr}(\yS)\cap (otr(\yS) L_U)\big]\neq \yemp$, so that from Lemma~\ref{lemm:ioco-reg} and Proposition~\ref{prop:equiv-conf} we have that $\yI_\yal \yioco \yS$ does not hold.
	So, we get  a test purpose $\yT_\yal=\yio{S_\yal}{t_0}{U}{I}{T_\yal}$ in $TP$ such that $\yI_\yal$ does not pass $\yT_\yal$.
	Thus, there is $\ysi$ such that $\ytrt{(t_0,q_0)}{\ysi}{(\yfail,q)}$ in $\yT_\yal \times\yI_\yal$, for some $q\in S_\yI$.
	We then get  $\ytrt{t_0}{\ysi}{\yfail}$ in $\yT_\yal$ and $\ytrt{q_0}{\ysi}{q}$ in $\yI_\yal$.
	We claim that $\ysi\not\in\{0,1\}^\star$. 
	To see this assume to the contrary that $\ysi\in\{0,1\}^\star$.
	Then $\ytrt{s_0}{\ysi}{s}$, so that we get $\ytrt{(t_0,s_0)}{\ysi}{(\yfail,s)}$ in $\yT_\yal \times\yS$. 
	This says  $\yS$ does not pass $TP$, that is, $\yS \yioco \yS$ does not hold, a contradiction.
	Thus, since $\ysi\in otr(\yI_\yal)$, we must have $\ysi=\yal x\yal'$ with $x\in L_I$ and 
	$\yal'\in\{0,1\}^\star$.
	
	Next, we look at other test purposes in $TP$.
	Let $\ybe\in R$, $|\ybe|=|\yal|$ and $\ybe\neq \yal$.
	By a similar reasoning, we get another test purpose $\yT_\ybe=\yio{S_\ybe}{t'_0}{U}{I}{T_\ybe}$ in $TP$, and 
	with $\ytrt{t'_0}{\ybe x\ybe'}{\yfail}$ in $\yT_\ybe$, for some $\ybe'\in\{0,1\}^\star$.
	We now claim that $\yT_\yal\neq \yT_\ybe$.
	Again, for the sake of contradiction, let $\yT_\yal=\yT_\ybe$.
	We now have $\ytrt{t_0}{\yal x\yal'}{\yfail}$ and $\ytrt{t_0}{\ybe x\ybe'}{\yfail}$ in $\yT_\yal$.
	Since prefixes of $\yal$ and of $\ybe$ are in $otr(\yS)\cap \{0,1\}^\star$, we can not reach $\yfail$ in $\yT_\yal$ with such prefixes, otherwise we would again reach the contradiction to the effect that $\yS \yioco \yS$ does not hold.
	We, therefore, must have
	$\ytrt{t_0}{\mu}{t_1}\ytrt{}{x_1}{t_2}\ytrt{}{\yal_1 x\yal'}{\yfail}$ and 
	$\ytrt{t_0}{\mu}{t'_1}\ytrt{}{x_2}{t'_2}\ytrt{}{\ybe_1 x\ybe'}{\yfail}$ in $\yT_\yal$, with $\yal=\mu x_1\yal_1$ and 
	$\ybe=\mu x_2\ybe_1$.
	Since $\yT_\yal$ is deterministic, we have $t_1=t'_1$.
	This gives  
	$(t_1,x_1,t_2)$ and $(t_1,x_2,t'_2)$ as transitions in $\yT_\yal$, with $x_1\neq x_2$ in $\{0,1\}$.
	But this contradicts $\yT_\yal$ being output-deterministic.
	We  conclude that $\yT_\yal\neq \yT_\ybe$ when $\yal\neq \ybe$ and $|\yal|=|\ybe|$, with $\yal$, $\ybe\in R$.
	
	Finally, we put a simple lower bound on the number of  words of length $m$ in $R$.
	Since the symbol $1$ occurs only in blocks of two in a word in $R$, there are $\binom{m-i}{i}$ distinct words of length $m$ with $i$ such blocks in $R$.
	So, we have 
	$F_m=\sum_{i=0}^{\lfloor m/2\rfloor} \binom{m-i}{i}$ words of length $m$ in $R$,
	where $F_m$ is the $m$th Fibonacci number, that is  
	$F_m=\frac{1}{\sqrt{5}}\left(\Phi^{m}+\frac{1}{\Phi^{m}}\right)\geq \frac{\Phi^{m}}{\sqrt{5}},$
	where $\Phi=(1+\sqrt{5})/2$.
	We then conclude that we must have 
	at least $\Phi^{m}/\sqrt{5}$ test cases in $TP$.
\begin{figure}
	\begin{center}
		\begin{tikzpicture}[->,>=stealth',shorten >=1pt,auto,node distance=2.0cm,
		semithick,scale=0.7,transform shape] 
		\node[state,initial by arrow, initial text={}] (q0) {$q_0$};
		\node[state] (q1) [right of=q0] {$q_1$};
		\node[state] (q2) [right of=q1] {$q_2$};
		\node[state] (qrm1) [right of=q2] {$q_{r-1}$};
		\node[state] (qr) [right of=qrm1] {$q_{r}$};
		\node[state] (qrp1) [right of=qr] {$q_{r+1}$};  
		\node[state] (pass) [below of=q2] {$\ypass$};  
		\path (q0) edge [] node {$y_1$} (q1)
		edge [] node [below] {$\overline{y_1}$} (pass)
		(q1) edge [] node {$y_2$} (q2)
		edge [] node [] {$\overline{y_2}$} (pass)          
		(qrm1) edge [] node {$y_r$} (qr)
		edge [] node [right] {$\overline{y_r}$} (pass)      
		(qr) edge [] node [above] {$x$} (qrp1)
		(qr) edge [] node [below] {$0,1$} (pass)
		(qrp1) edge [loop above] node [above] {$0,1$} ()
		(pass) edge [loop right] node [] {$0,1$} ()               
		;
		\draw[dashed]  (4.5,0)--(5.5,0);  
		\draw[dashed]  (4,-.5)--(4,-1.5);            
		\end{tikzpicture}
		\caption{Implementation $\yI_\yal$.}\label{sec3:spec-complexity-impl}
	\end{center}
\end{figure}

	For a later reference, we note that both transitions from $q_r$ to ${\ypass}$, as well as both self-loops at state $q_{r+1}$, were never necessary in the proof. Their only function here is to make states $q_r$ and $q_{r+1}$ also input-enabled.
\end{proof}

It is  clear that Theorem~\ref{theo:ioco-specific-exp} applies to any specification $\yS$ in which Figure~\ref{sec3:spec-complexity}
occurs, with $s_0$ being reachable from its initial state.

\section{Another Restricted Class of IOLTS Models}\label{sec:adepetre-suites}

As another illustration of the applicability of our approach, in this section we want to apply our proposal to a subclass of IOLTS models that were studied more recently~\cite{simap-generating-2014}.
In that work, Sim\~ao and Petrenko considered a more contrived subclass of IOLTS models, and showed  that it is possible to generate {\bf ioco}-complete test suites for specifications in that subclass. 
But in that work, they did not consider the size complexity of the test suites that were generated following their approach.
In this section we apply our method to that same subclass of IOLTS models and show how to construct {\bf ioco}-complete test suites for such models in a more unified and direct way. 
We also study  the size complexity of the test suites that are generated using our approach.
It comes as no surprise that the test suites that are generated using the approach described in this section will have a size that might be exponentially related to the size of the IOLTS models involved in the testing process, and when considering worst case scenarios.
Interestingly, however, as one of the main results of this section, we also establish a precise \emph{exponential lower bound} on the worst case asymptotic size of \emph{any test suite} that is required to be complete for the class of IOLTS models treated here. 

Since there are several restrictions that IOLTS models must satisfy in Sim\~ao and Petrenko's approach~\cite{simap-generating-2014}, we introduce them in stages, as needed.
Motivation for considering these restrictions can be found in their work~\cite{simap-generating-2014}. 
Recall Definitions~\ref{def:out-after} and~\ref{def:in-compl-out-determ}.
First we use the $\yinp$ and $\yout$ functions that collect inbound and outbound transitions, respectively, in order to characterize the notions of  input-complete and output-complete states, among others.
Also we need the notion of $\yinit$, where  $\yinit(V)=\yinp(V)\cup\yout(V)$, for all $V\ysse S$.
Now, the notion of an IOLTS model being progressive is formalized.
\begin{defi}[\cite{simap-generating-2014}]\label{def:inp-sink}
	Let $\yS=\yioltsI$ be an IOLTS, with $L=L_I\cup L_U$, and let $s\in S$.
	We say that $s$ is: (i) a \emph{sink state}  if $\yinit(\{s\})=\yemp$; (ii) a \emph{single-input} state when $\vert \yinp(\{s\})\vert\leq 1$; (iii) a \emph{input-state} when $\yinp(\{s\})\neq \yemp$; (iv) an \emph{input-complete} state if $\yinp(\{s\})=L_I$ or $\yinp(\{s\})=\yemp$; and (v) we say that $s$ is an
	\emph{output-complete} state when $\yout(\{s\})=L_U$ or $\yout(\{s\})=\yemp$.
	We say that the IOLTS $\yS$ is single-input, input-complete, or output-complete if all states in $S$ are, respectively,  single-input, input-complete, or output-complete.
	We also say that $\yS$ is \emph{initially-connected} if every state in $\yS$ is reachable from the initial state, and 
	we say that $\yS$ is \emph{progressive} if it has no sink state and for any cycle
	$\ytr{q_0}{x_1}{q_1}\ytr{}{x_2}{q_2}\ytr{}{x_3}{}\cdots\ytr{}{x_k}{q_k}$, with $q_0=q_k$ and $q_i\in S$ ($0\leq i\leq k$) we have $x_j\in L_I$ for at least one transition $\ytr{q_{j-1}}{x_j}{q_j}$,  for some $1\leq j\leq k$.
	Let $\yiocip{I}{U}\ysse\yioc{I}{U}$ denote the class of all IOLTSs which are deterministic, input-complete, progressive and initially-connected.
\end{defi}
Earlier, in Definition~\ref{def:inic-tret}, a state $s$ was said to be input-enabled when $\yinp(s)=L_I$.
Note the slight difference with the notion of $s$ being input-complete just given.

In Definition~\ref{def:testsuite}, the terms test case and test suite were already reserved to refer to words and languages, respectively, over $L_I\cup L_U$.
In order to avoid confusion about the use of these terms in this section, here we will employ the terms \emph{schemes} and \emph{scheme suites}, respectively,  when referring to the notions of  test cases and test suites as in Definition~3 of Sim\~ao and Petrenko~\cite{simap-generating-2014}.
Note that, contrary to the notion of a test purpose in Definition~\ref{def:test-purpose}, test schemes in Sim\~ao and Petrenko's approach~\cite{simap-generating-2014} do not have their sets of input and output symbols reversed with respect to the sets of input and output symbols of specifications and implementations. 
The next definition reflects this idea.
\begin{defi}\label{testcase}
	Let $L_I$ and $L_U$ be sets of symbols with $L_I\cap L_U=\yemp$ and $L=L_I\cup L_U$.
	A \emph{scheme over $L$} is an acyclic single-input and output-complete IOLTS $\yT\in\yioc{I}{U}$ which has a single sink state, designated $\yfail$. 
	A \emph{scheme suite} $SS$ over $L$ is  a finite set of schemes over $L$.
\end{defi}

Proceeding, recall Definition~\ref{defi:cross}, of the cross-product operator $\yS\times\yI$ for synchronous execution of two IOLTSs $\yS$ and $\yI$.
Sim\~ao and Petrenko~\cite{simap-generating-2014} denote the exactly same operator by $\yS\cap \yI$, with the proviso that, in that work, internal $\tau$-moves were not considered, by definition.
In this section we will continue to use the cross-product to denote synchronous execution. 
This being noted, we remark now that our Definition~\ref{def:passes}, for when an implementation IOLTS $\yI$ passes a test scheme $\yT$ and passes a scheme suite  $SS$, exactly matches Definition~4 of Sim\~ao and Petrenko and, as a consequence, we also have the very same notion of {\bf ioco}-completeness, as stated in our Definition~\ref{def:passes} and their Definition~4. 

In what follows we want to show that our approach can also be used to construct {\bf ioco}-complete scheme suites, but with the advantage that we do not need to further constrain specification and implementations models.  
\begin{theo}\label{theo:ade-petre}
	Let $\yS=\yio{S_\yS}{s_0}{I}{U}{R_\yS}\in\yioc{I}{U}$ a specification over $L=L_I\cup L_U$, and let $m\geq 1$.
	Then we can effectively construct a finite scheme suite $SS$ over $L$ which is $m$-{\bf ioco}-complete for $\yS$. 
\end{theo}
\begin{proof}
	We start with the Proposition~\ref{prop:deterministic-lts} to transform $\yS$ into an equivalent deterministic IOLTS, if $\yS$ is not already deterministic. 
	Using Proposition~\ref{prop:m-complete}, we get a scheme suite $SS$ which is {\bf ioco}-complete for $\yS$ relatively to $\yioc{I}{U}[m]$, and such that all schemes in $SS$ are deterministic and acyclic IOLTSs.
	From the proof of Proposition~\ref{prop:m-complete}, it is clear that all schemes in $SS$ have a single {\bf fail} state, which is also a sink state.
	
	Let $\yT=\yio{S_\yT}{t_0}{I}{U}{R_\yT}$ in $SS$ be any scheme constructed as in the proof of Proposition~\ref{prop:m-complete},  and let $s\in S_\yT$ be any state of $\yT$. 
	From that proof, we know that there is at most one transition $(s,\ell,p)$ in $R_\yT$, for any $\ell\in L_U\cup L_I$ and any $p\in S_\yT$.
	There are two cases for $\ell\in L_U\cup L_I$. 
	If $\ell\in L_I$, from Definition~\ref{def:inp-sink}, we immediately get that 
	$\vert \yinp(s)\vert\leq 1$ and $\yout(s)=\yemp$, that is, $s$ is single-input and output-complete.
	Now, assume $\ell\in L_U$ is an output symbol of $\yT$.
	Then, $s$ is already single-input.
	If $\vert L_U\vert=1$, then $s$ is already an output-complete state.
	Else, in order to turn $s$ into an output-complete, transform scheme $\yT$ to a scheme $\yT'=\yio{S'_\yT}{t_0}{I}{U}{R'_\yT}$ by adding a new $\ypass$ state to $S_\yT$ and, for any other $x\in L_U$ with $x\neq \ell$, add a new transition $(s,x,\ypass)$ to $R_\yT$.
	It is clear that $s$ is  now output-complete.
	
	Since $\ypass$ is a sink state in $\yT'$, for any implementation $\yI=\yio{S_\yI}{q_0}{I}{U}{R_\yI}\in\yioc{I}{U}$, we have that $\ytrt{(t_0,q_0)}{\star}{(\yfail,q)}$ in $\yT$ if and only if $\ytrt{(t_0,q_0)}{\star}{(\yfail,q)}$ in $\yT'$, for any $q\in S_\yI$.
	Therefore, we also get that $\yI$ passes $\yT$ if and only if $\yI$ passes $\yT'$.
	Let $SS'$ be the scheme suite obtained from $SS$ with the transformation just discussed,
	now applied to each scheme in $SS$.
	Since $SS$ is $m$-{\bf ioco}-complete for $\yS$ then it follows that $SS'$ is also $m$-{\bf ioco}-complete for $\yS$.
	
	Since $SS'$ satisfies all the requirements in Definition~\ref{testcase}, the proof is complete. 
\end{proof}
\begin{figure}[tb]
	\begin{center}
		\begin{tikzpicture}[->,>=stealth',shorten >=1pt,auto,node distance=2.0cm,
		semithick,scale=0.6,transform shape]
		\node[state,initial by arrow, initial text={}] (s0) {$s_0$};
		\node[state] (s1) [right of=s0] {$s_1$};
		\node[state] (s2) [right of=s1] {$s_2$};
		\node[state] (s3) [right of=s2] {$s_3$};
		\node[state] (skm1) at (9.5,0) {$s_{k-1}$};
		\node[state] (sk) [right of=skm1] {$s_{k}$};
		\node[state,dashed] (s0x) [below of=sk] {$s_0$};
		\path (s0) edge [bend right=-35] node {$1$} (s1)
		edge [loop above] node [] {$0$} () 
		edge  [bend right=45] node [below] {$a$} (s2)
		(s1) edge [] node {$x$} (s2)
		edge [bend left] node [anchor=south] {$1$} (s0)
		edge [loop above] node [] {$0$} ()
		edge [bend right=40] node [below] {$a$} (s3)
		(s2) edge  [] node [] {$0,1$} (s3)
		(skm1) edge  [] node [] {$0,1$} (sk)
		(sk) edge  [] node [] {$x$} (s0x);
		\draw[dashed,->] (6.5,0)--(7.5,0); \node at (6.8,0.3) {$0,1$};
		\draw[dashed,->] (8.0,0)--(9,0); \node at (8.3,0.3) {$0,1$};
		\draw[dashed] (-1,-1.5)--(2.8,-1.5)--(2.8,2.0)--(-1,2.0)--(-1,-1.5);
		\node at (1,2.3) {$B_1$};
		\draw[dashed] (3.2,-1.3)--(12.5,-1.3)--(12.5,1)--(3.2,1)--(3.2,-1.3);
		\node at (8.5,1.3) {$B_2$};
		\end{tikzpicture}
		\caption{Specification $\yS'$, modifying $\yS$ of Figure~\ref{sec3:spec-complexity}.}\label{sec3:spec-complexity-extd}
	\end{center}
\end{figure}
A similar result follows if we restrict all specifications and implementations to be members of the more restricted subclass $\yiocip{I}{U}$ of IOLTS models.
\begin{coro}\label{coro:ade-petre-input-states}
	Let $\yS\in\yiocip{I}{U}$ be a specification, with $L=L_I\cup L_U$, and let $m\geq 1$.
	Then we can effectively construct a finite scheme suite $SS$ over $L$ which is $m$-{\bf ioco}-complete for $\yS$ relatively to the sub-class of $\yiocip{I}{U}$ and in which every implementation has no more input-states then $\yS$.   
\end{coro}
\begin{proof}
	Since $\yiocip{I}{U}\ysse \yioc{I}{U}$, this follows immediately from Theorem~\ref{theo:ade-petre}.
\end{proof}

Sim\~ao and Petrenko~\cite{simap-generating-2014} consider specifications and implementations that are further restricted to be input-state-minimal, in the sense that any two distinct input-states are always distinguishable.
In their Corollary 1, two states $r$ and $p$ are said to be distinguishable when there are no sink state in the cross-product $\yioini{S}{r}\times\yioini{S}{p}\,$, where $\yioini{S}{p}$ stands for the same model as $\yS$, but now with $p$ being the initial state.
We formalize these notions next.
\begin{defi}[\cite{simap-generating-2014}]\label{def:distingui}
	Let $\yS=\yio{S_\yS}{s_0}{I}{U}{R_\yS}$ and $\yQ=\yio{S_\yQ}{q_0}{I}{U}{R_\yQ}$ be two IOLTS models, and let $s\in S_\yS$ and $q\in S_\yQ$.
	We say that $s$ and $q$ are \emph{distinguishable} if there is a sink state in the cross-product $\yioini{S}{s}\times\yioini{Q}{q}\,$.
	Otherwise, we say that $r$ and $s$ are \emph{compatible}.
	We say that an IOLTS $\yS$ is \emph{input-state-minimal} if any two distinct input-states $r$, $s\in S_\yS$ are distinguishable. 
	We denote by $\yiocipm{I}{U}\ysse\yiocip{I}{U}$ the subclass of all models  in
	$\yiocip{I}{U}$ which are also input-state-minimal.  
\end{defi}  
Recall Definition~\ref{def:inp-sink}.
Implementations are yet further constrained by Sim\~ao and Petrenko~\cite{simap-generating-2014} to have at most as many input-states as the specification model.
Let $k\geq 1$, and let $\yltsn{IMP}(L_I,L_U)\ysse\yioc{I}{U}$ be any family of IOLTS models.
We denote by $\yltsn{IMP}(L_I,L_U,k)$ the subclass of $\yltsn{IMP}(L_I,L_U)$ comprised by all models with at most $k$ input-states.
The main result of Sim\~ao and Petrenko~\cite{simap-generating-2014} is their Theorem 1, which shows that for any specification $\yS$ in the class $\yiocipm{I}{U}$ it is possible to construct scheme suites that are {\bf ioco}-complete for implementation models in the class $\yiocipmk{I}{U}{k}$, where $k$ is the number of input-states in $\yS$. 
This result also follows easily from Theorem~\ref{theo:ade-petre}.
\begin{coro}\label{coro:ade-petre-minimal}
	Let $m\geq 1$, and let $\yS\in\yiocipm{I}{U}$ be a specification with $k\geq 0$ input-states.
	Then we can effectively construct a finite scheme suite $SS$ over $L_I\cup L_U$ which is $m$-{\bf ioco}-complete for $\yS$ relatively to the sub-class of $\yiocipmk{I}{U}{k}$. 
\end{coro}
\begin{proof}
	Note that $\yiocipm{I}{U}\ysse \yioc{I}{U}$ and  $\yiocipmk{I}{U}{k}{[m]}\ysse \yioc{I}{U}[m]$. 
	Then the result follows applying Theorem~\ref{theo:ade-petre}.
\end{proof}

The complexity of the generated test suites were not analysed by Sim\~ao and Petrenko~\cite{simap-generating-2014}.
However, as we argued in Subsection~\ref{sec:tretma-complexity} and in Theorem~\ref{theo:ioco-specific-exp}, we cannot, in general, avoid scheme suites to asymptotically grow very large, even when specifications are confined to the subclass $\yiocipm{I}{U}$, and implementations are restricted to the subclass $\yiocipmk{I}{U}{k}$, where $k$ is the number of input-states in $\yS$.
The next result establishes a worst case exponential asymptotic lower bound on the size of the test schemes that can be generated using their Theorem 1~\cite{simap-generating-2014} or, equivalently, using our Corollary~\ref{coro:ade-petre-minimal}.
\begin{theo}\label{theo:ade-exp}
	Let $k\geq m\geq 3$, and let $L_I=\{0,1\}$ and $L_U=\{a,x\}$.
	There is a specification $\yS\in\yiocipm{I}{U}$ with $k$ input-states, and for which any scheme suite $SS$ that is  $m$-{\bf ioco}-complete for $\yS$, relatively to the class $\yiocipmk{I}{U}{k}$,
	must be of size  $\yomeg{\Phi^m}$, where $\Phi=(1+\sqrt{5})/2\approx 1.61803$.
\end{theo}
\begin{proof}
	We want to use an argument almost exactly as that used in the proof of Theorem~\ref{theo:ioco-specific-exp}, with a few adjustments to be considered later on.
	
	First, note that the specification $\yS$, used in the proof of Theorem~\ref{theo:ioco-specific-exp}, and depicted in Figure~\ref{sec3:spec-complexity}, is deterministic, input-complete, progressive and initially-connected, that is, $\yS\in\yiocip{I}{U}$.
	Also, the implementation $\yI_\yal$, constructed in that proof and illustrated in Figure~\ref{sec3:spec-complexity-impl}, is also deterministic and in the class $\yiocip{I}{U}$.
	Recall that we write $\overline{y}=0$ when $y=1$ and $\overline{y}=1$ when $y=0$.
	The  argument, then, proceeds just as in the proof of Theorem~\ref{theo:ioco-specific-exp}, and we postulate the existence of a scheme $\yT_\yal=\yio{S_\yal}{t_0}{I}{U}{T_\yal}$ in $SS$, and such that $\yI_\yal$ does not pass $\yT_\yal$.
	The proof continues by imitating the argument in the proof of Theorem~\ref{theo:ioco-specific-exp}. 
	In the present case, we will have 
	$(t_1,x_1,t_2)$ and $(t_1,x_2,t'_2)$ 
	in the scheme $\yT_\yal$, with $x_1\neq x_2$, and $x_1, x_2\in\{0,1\}=L_I$.
	Observe that, now, the input alphabet for $\yT_\yal$ is $L_I=\{0,1\}$.
	Since, according to Definition~\ref{testcase}, any scheme must be single-input, we need $x_1=x_2$, and so we reach a contradiction again, as in the proof of Theorem~\ref{theo:ioco-specific-exp}.
	As, before, this will force $\yT_\yal=\yT_\ybe$ when $\yal, \ybe\in R$, and $\yal\neq \ybe$, with $\vert \yal \vert=\vert \ybe \vert=r\leq m-3$.
	From this point on, the argument follows the one in the proof of  Theorem~\ref{theo:ioco-specific-exp}, establishing that $SS$ must be of size $\yomeg{\Phi^m}$.
	
	The proof would be complete if we had $\yS\in \yiocipm{I}{U}$ and $\yI_\yal\in\yiocipmk{I}{U}{k}[m]$, where $k$ is the number of input-states in $\yS$.
	We will now extend Figures~\ref{sec3:spec-complexity} and~\ref{sec3:spec-complexity-impl} in such a way that these conditions are met, while preserving the validity of the previous argument.
	First note that $\yS$ has 3 input-states, whereas the implementations $\yI_\yal\in\yiocipmk{I}{U}{k}$ has $r+3\leq (m-3)+3=m\leq k$ input-states.
	We then extend the specification in Figure~\ref{sec3:spec-complexity} as shown in Figure~\ref{sec3:spec-complexity-extd}, with states $s_3, \ldots, s_k$.
	States $s_0$ is repeated to avoid the clutter.
	Note that the important transitions on $0$ and $1$ out of states $s_0$ and $s_1$, as well as the transition on $x$ out of  state $s_1$ were not touched, so that the argument above
	is still valid when we consider this new specification.
	Call this new specification $\yS'$.
	States $s_i$, $0\leq i\leq k-1$, are the input-states, so $\yS'$ has $k$ input-states.
	It is also easy to check that $\yS'$ is deterministic, input-complete and initially-connected.
	Moreover, $\yS'$ is also progressive, since any cycle in $\yS'$ must go through a transition on an input.
	In order to assert that $\yS'$ is in the class $\yiocipm{I}{U}$, we need to verify that any two input-states in Figure~\ref{sec3:spec-complexity-extd} are distinguishable.
	As indicated in Figure~\ref{sec3:spec-complexity-extd}  states can be partitioned  into the two blocks $B_1$ and $B_2$.
	Consider two distinct states $s_i$, $s_j\in B_2$ with $2\leq i<j\leq k$, and let $w=0^{k-j}$.
	We see that $\ytrt{(s_i,s_j)}{w}{(s_{\ell},s_k)}$ where $\ell=k-(j-i)$, so that $2\leq i\leq\ell\leq k-1$.
	Since $\yinit(s_\ell)\cap \yinit(s_k)=\yemp$ and we conclude that any two distinct states in $B_2$ are distinguishable.
	Then, since $\ytrt{(s_0,s_1)}{a}{(s_{2},s_{3})}$, it follows that $s_0$ and $s_1$ are also distinguishable.
	We now argue that $s_0$ and $s_1$ are distinguishable from any state $s_i\in B_2$, $2\leq i\leq k$. 
	Let $w=0^{k-i}$, so that we get  $\ytrt{(s_0,s_i)}{w}{(s_{0},s_{k})}$.
	Since we already know that $s_0$ is distinguishable from $s_k$, we conclude that $s_0$ is distinguishable from any state in $B_2$.
	Likewise, with $w=0^{k-i}x$ we see that  $\ytrt{(s_1,s_i)}{w}{(s_{2},s_{k})}$, so that $s_1$ is also distinguishable from any state in $B_2$.
	Hence, any pair of states in $B_1\times B_2$ are distinguishable, and we can now 
	state that any two distinct states in $B_1\cup B_2$ are distinguishable, 
	that is, $\yS'$ is input-state-minimal.
	We conclude that $\yS'\in\yiocipm{I}{U}$ with $k$ input-states, as desired.
	
\begin{figure}
	\begin{center}
		\begin{tikzpicture}[->,>=stealth',shorten >=1pt,auto,node distance=2.0cm,
		semithick,scale=0.7,transform shape] 
		\node[state,initial by arrow, initial text={}] (q0) {$q_0$};
		\node[state] (q1) [right of=q0] {$q_1$};
		\node[state] (q2) [right of=q1] {$q_2$};
		\node[state] (qrm1) [right of=q2] {$q_{r-1}$};
		\node[state] (qr) [right of=qrm1] {$q_{r}$};
		\node[state] (qrp1) [right of=qr] {$q_{r+1}$};  
		\node[state] (pass) [below of=q2] {$\ypass$};  
		\path (q0) edge [] node {$y_1$} (q1)
		edge [] node [below] {$\overline{y_1}$} (pass)
		(q1) edge [] node {$y_2$} (q2)
		edge [] node [] {$\overline{y_2}$} (pass)          
		(qrm1) edge [] node {$y_r$} (qr)
		edge [] node [right] {$\overline{y_r}$} (pass)      
		(qr) edge [] node [above] {$x$} (qrp1)
		(pass) edge [loop right] node [] {$0,1$} ()               
		;
		\draw[dashed]  (4.5,0)--(5.5,0);  
		\draw[dashed]  (4,-.5)--(4,-1.5);            
		\end{tikzpicture}
		\caption{A modified implementation $\yI_\yal$.}\label{sec3:spec-complexity-impl-modif}
	\end{center}
\end{figure}
	We now turn to the implementation.
	In the proof of Theorem~\ref{theo:ioco-specific-exp}, we noted that  both transitions from state $q_r$ to state $\ypass$, together with the self-loops at state $q_{r+1}$ could have been removed, with no prejudice to the argument given therein.
	Thus, we are now looking at Figure~\ref{sec3:spec-complexity-impl-modif},  which here we also designate by $\yI_\yal$.
	Also, from the proof of Theorem~\ref{theo:ioco-specific-exp} we recall that $1\leq r\leq m-3$.
	By inspection, we see that $\yI_\yal$ is deterministic, input-complete, progressive, initially-connected, and has $r+3\leq (m-3)+3=m\leq k$ states.\
	Moreover, all states except for states $q_r$ and $q_{r+1}$, are input-states, so that $\yI$ has at most $k$ input-states, as we need.
	To complete the proof, we show that every pair of distinct input-states of $\yI_\yal$ are distinguishable.
	In order to show that the $\ypass$ state is  distinguishable from any other state in $\yI_\yal$, 
	fix some $q_j$, $0\leq j\leq r-1$, and define  $w=y_{j+1}\cdots y_r$.
	Clearly, $\ytrt{q_j}{w}{q_{r}}$ and, since $\ytrt{\ypass}{w}{\ypass}$, we get 
	$\ytrt{(q_j,\ypass)}{w}{(q_r,\ypass)}$.
	Since $\yinit(q_r)\cap \yinit(\ypass)=\yemp $ we conclude that $\ypass$ is distinguishable from any state $q_j$, $0\leq j\leq r-1$.
	Lastly, take  a state $q_i$ distinct from $q_j$, that is let $0\leq i<j\leq r-1$.
	Now we get $\ytrt{q_i}{w}{q_{\ell}}$ where $\ell=r-j+i=r-(j-i)\leq r-1$.
	Hence, $\ytrt{(q_i,q_j)}{w}{(q_\ell,q_r)}$ and, because $\ell\leq r-1$,  we see that $\yinit(q_\ell)\cap \yinit(q_r)=\yemp $, thus proving that $q_i$ and $q_j$ are also distinguishable.
	Putting it together, we conclude that any pair of distinct input-states of $\yI_\yal$ are distinguishable, that is, $\yI_\yal$ is also input-state-minimal.
\end{proof}

Theorem~\ref{theo:ade-exp} clearly also applies to any specification $\yS$ in which the model depicted in Figure~\ref{sec3:spec-complexity-extd} appears as a 
sub-model with state $s_0$ being reachable from the initial state of $\yS$. 
This is in contrast to Theorem~\ref{prop:ioco-poli} which says that, given a specification $\yS$ over an alphabet $L$, there is an algorithm of asymptotic time
complexity $\yoh{k m }$ for checking $m$-{\bf ioco}-completeness, where $k=n_\yS n_L$, $n_L=\vert L\vert$ and $n_\yS$ is the number of states in $\yS$.

We also remark that in Theorem~1 of Sim\~ao and Petrenko~\cite{simap-generating-2014},  implementations are further restricted to be ``input-eager'', although they do not precisely define this notion in that text. 
On the other hand, in none of their proofs is the input-eager hypothesis explicitly used, leading us to infer that constraining implementations
to also be input-eager is a practical consideration to render the testing process more controllable, from the point of view of a tester that is conducting the testing process.
Thus, given that the input-eager condition is not strictly necessary to establish their 
Theorem~1, we conclude that Theorem~\ref{theo:ade-exp}
expresses a valid worst case exponential asymptotic lower bound on the size of the test suites claimed by Theorem~1 in~\cite{simap-generating-2014}.

\section{Related Works}\label{related}

IOLTS models are largely used to describe the syntax and the semantics of  systems where input and output actions can occur asynchronously, thus capturing a wide class of systems and communication protocols. 
Several works have studied different aspects of (complete) test suite generation for families of IOLTS models, under various conformance relations.
We comment below on some works that are more closely related to our study.

de Vries and Tretmans~\cite{vriet-towards-2001} presented an {\bf ioco}-based testing theory to obtain $e$-complete test suites. 
This variant of test suite completeness  is based on specific test purposes, that share particular properties related to certain testing goals. 
In that case, they consider only the observable behaviors according to the observation objective when testing black-box implementations. 
It turns out that such specific test purposes, and their combinations, somewhat limit the fault coverage spectrum,  \emph{e.g.}, producing inconclusive verdicts. 
Large, or even infinite, test suites can be produced by their test generation method. 
Therefore,  test selection criteria need to  be put in place to avoid this problem, at least when applied in practical situations. 
On the other hand, our approach allows for a wider class of IOLTS models, 
and a low degree polynomial time algorithm was devised for efficiently testing {\bf ioco}-conformance in practical applications. 

Petrenko et al.~\cite{petryh-testing-2003} studied IOLTS-testing strategies considering implementations that cannot block inputs, and also testers that can never prevent an implementation from producing outputs. 
This scenario calls for input and output communication buffers that can be used for the exchange of symbols between the tester and the implementations.
But this effectively leads to another class of testing strategies, where arbitrarily large buffer memories (queues) are allowed.

Tretmans~\cite{tret-model-2008} surveyed the classic \textbf{ioco}-conformance relation for IOLTS models. 
He also developed the foundations of an {\bf ioco}-based testing theory for IOLTS models~\cite{tret-test-1996}, where implementations under test are treated as ``black-boxes'', a testing architecture where the tester, having no access to the internal structure of IUTs, is seen as an artificial environment that drives the  exchange of input and output symbols with the IUT. 
In this case, some restrictions must be observed by the specification, the implementation and the tester models, such as input-completeness and output-determinism.
The algorithms developed therein, however, may in general lead to infinite test suites, making it more difficult to devise solution for practical applications. 
In our work we described a method that, considering the exact same restrictions to the IOLTS models, does in fact generate finite sets of test purposes that can be used in practical situations.
In rare situations, the algorithm may lead to exponential sized testers. 
On the other hand, if the same restrictions are to be obeyed by the specification, the implementation and the tester IOLTS models, we established an exponential worst case asymptotic lower bound on the size of the testers.    
This shows that, if those restrictions are in order,  generating exponential sized testers is, in general, unavoidable, being rather an intrinsic characteristic of the problem of requiring {\bf ioco}-completeness. 

Sim\~ao and Petrenko~\cite{simap-generating-2014}, in a more recent work, also described an approach to generate finite {\bf ioco}-complete test suites for a class of IOLTS models. 
They, however, also imposed a number of restrictions on the specification and the implementation models in order to obtain {\bf ioco}-complete finite test suites. 
They assumed that the test purposes to be single-input and also output-complete. 
Moreover, specifications and implementations must be input-complete, progressive, and initially-connected, so further restricting the class of IOLTS models that can be tested according to their fault model. 
They also did not study the complexity of their method for generating {\bf ioco}-complete test suites under those restrictions the models must obey. 
In contrast, we applied our approach to a testing architecture that satisfies the same restrictions, and showed how to generate {\bf ioco}-complete test suites in a  more straightforward manner.
Further, we examined the complexity of the problem of generating {\bf ioco}-complete test suites under the
same restrictions on the IOLTS models, and  established an exponential   worst case asymptotic lower bound on the size of any {\bf ioco}-complete test suite that can be generated in this situation.

Noroozi et al.~\cite{noromw-complexity-2014} presented a polynomial time reduction from a variation of the SAT problem to the problem of checking \textbf{ioco}-completeness, thus establishing that, under very general assumptions about the IOLTS models~---~including non-determinism,~---~that checking {\bf ioco}-completeness is a PSPACE-complete problem.
In a more restricted scenario, treating only deterministic and input-enabled IOLTS models, they proposed a polynomial time algorithm, based on a simulation-like preorder relation. 
This is the same complexity bound that our method attains but, in contrast, our approach treats a  wider class of conformance relations not being restricted to {\bf ioco}-conformance only.
In another work, Noroozi et al.~\cite{norokmw-synchrony-2015} also  studied the problem of synchronous and asynchronous conformance testing, when allowing communication channels as auxiliary memories. 
They treated a more restricted class of IOLTS, the so-called Internal Choice IOLTSs, where quiescent states must be also input-enabled. 
The notion of  {\bf ioco}-conformance as well as the notion of traces of the models in the testing environment are also more restricted. 
In a white-box testing strategy, where the structure of IUTs were always accessible, algorithms to generate test cases are shown to be sound and exhaustive for testing completeness. 
However, in a setting where IUTs are black-boxes these algorithms are not applicable, thus limiting their practical use. 


In a recent work, Roehm et al.~\cite{roehowa-reachset-2016} introduced a variation of conformance testing, related  to safety properties. 
Despite being a weaker relation than  trace-inclusion conformance, it allows for  tunning a trade-off between accuracy and computational complexity, when checking conformance of hybrid systems. 
Instead of verifying the whole system, their approach searches for counter-examples. 
They also proposed a test selection algorithm that uses a coverage measure to reduce the number of test cases for conformance testing. 
However,  since the models are hybrid,  the continuous flow of time forces a discretization of the models in order to reduce the test generation problem to one that can be applied to discrete models. 
This, in turn, imposes a trade-off between accuracy and computational load, which must be tuned by  appropriate choices related to some over-approximations.

Other works have considered \textbf{ioco}-based testing for compositional systems, where components of a more complex system can be formally tested using composition operators to capture the resulting behavior of multiple components. 
Benes et al.~\cite{benedhk-complete-2015} have proposed merge and quotient operators in order to check consistency of more complex parts of systems under test, in an attempt to
reduce  the effort of model-based testing of more complex systems whose structures can be describe compositionally. 
Following a similar line, Daca et al.~\cite{dacahkn-compositional-2014} proposed compositional operators, friendly composition and hiding, applied to an \textbf{ioco}-testing theory in order  to minimize the integration of testing efforts. 
The result of the friendly composition is an overall specification that integrates the component specifications while pruning away any inputs that lead to incompatible interactions between the components. 
The friendly hiding operation can prune inputs that lead to states which are ambiguous with respect to underspecified parts of the system. 
In a similar vein, Frantzen and Tretmans~\cite{frant-model-2007}  presented a model-based testing method for components of a system where a complete behavior model is not available, and  used a parallel operator to obtain the resulting behavior when integrating different components. 
They proposed a specific conformance relation for the components and devised an algorithm that constructs complete test suites.

\section{Conclusions}\label{sec:conclusion}

Conformance between specification  and implementation IOLTS  models often need to be checked, in order to establish a mathematical guarantee of correctness of the implementations. 
The {\bf ioco} framework has been the conformance relation of choice for verifying IOLTS  models in several testing architectures. 

In this work we addressed the problem of conformance testing and test case generation  for asynchronous systems that can be described using IOLTS models as the base formalism.
A new notion of conformance relation was studied, one that is more general and encompasses the classic {\bf ioco}-conformance. 
This notion  opened the possibility for a much wider class of conformance relations, all uniformly treated under the same formalism. 
In particular, it allows for properties or fault models to be specified by formal languages, \emph{e.g.}, regular languages.
As a further advantage, very few restrictions over specification and implementation IOLTS models must be satisfied when generating finite and complete test suites under any notion of conformance that fits within the more general setting studied herein.
We also proved correct  a polynomial time algorithm to test general conformance in a ``white-box'' architecture.
Once a specification is fixed, our algorithm runs in linear time on the size of the implementations.

Equipped with the new notion of conformance relation, we specialized the test generation process in order to cover other special cases of conformance relations, such as the classical {\bf ioco}-conformance relation. 
In addition, complexity issues related to complete test suite generation for verifying {\bf ioco}-conformance in settings where the IOLTS models were under several specific restrictions were also discussed.
For some sets of such restrictions on the models, we showed that the state explosion problem cannot be avoided, in general, forcing {\bf ioco}-complete test suites to grow exponentially with the size of the implementation models. 
Further, in these cases, we proved correct general algorithms with time complexities that attained such lower bounds, while still generating complete test suites.
This indicates that other families of specialized IOLTS modes could be considered by our approach, leading to  similar results.

Other areas that might be inspired by these ideas are symbolic test case generation, where data variables and parameters are also present~\cite{rusubj-approach-2000,gastlrt-symbolic-2006}, as well as conformance relations and generation methods for models that can capture real-time~\cite{kric-model-2007,briob-test-2005}.

\begin{appendices}
\section{The equivalence of the ioco relations}\label{subsec:equiv-ioco}

In this appendix we discuss the relationship between the {\bf ioco} relation  described in this work and a classical {\bf ioco} relation used in the literature~\cite{tret-model-2008}. 
We want to establish that both of these variations describe the same {\bf ioco} relation. 
Since the precise definitions of several notions in this work and in the original proposal~\cite{tret-model-2008} differ slightly, we need to proceed step by step, comparing the same notions as defined in both texts.
The differences are most marked when the notion of quiescence is treated in both texts, and so, special care must be taken when comparing notions related to quiescence.

From now on, if $X$ denotes any object defined both in~\cite{tret-model-2008} and in this work, we let $X_T$ 
be the variation of $X$ as defined in~\cite{tret-model-2008}, and we use $X_A$ for the same object $X$ as defined in this work.
For instance, $\ytrT{}{}{}$ is the trace relation in LTSs as defined in~\cite{tret-model-2008}, and $\ytrA{}{}{}$ is the trace relation as defined in this work. 
In many cases they will be exactly the same, but in some other cases there might be a slight variation between the two relations.

\subsection*{The general model}

We first note that an LTS model as defined in~\cite{tret-model-2008} is denoted by $\ylts{S}{L}{T}{s_0}$, whereas an LTS model is here denoted as $\ylts{S}{s_0}{L}{T}$.
Further, in~\cite{tret-model-2008} LTS models can be infinite objects, whereas we deal only with finite models.
We will from now on restrict ourselves to finite models only.
\begin{hypo}\label{hypo:fin}
	We assume that all LTS models are finite. 
\end{hypo}

The notions of a path, $\ytr{}{}{}$, and of an observable path, $\ytrt{}{}{}$, appear as Definitions 3 and 4 in~\cite{tret-model-2008}. 
See also Definition~\ref{def:trace} here.
The following is immediate.
\begin{prop}\label{prop:derivas}
	$\ytrT{}{}{}\,=\,\ytrA{}{}{}$ and also $\ytrtT{}{}{}\,=\,\ytrtA{}{}{}$.
\end{prop}
\begin{proof} Follows from the definitions. \end{proof}
We now write $\ytr{}{}{}$ for both  $\ytrT{}{}{}$ and $\ytrA{}{}{}$.
Likewise, we write $\ytrt{}{}{}$ for both $\ytrtT{}{}{}$ and $\ytrtA{}{}{}$.

The function \yafter appears in Definition~5(3) in~\cite{tret-model-2008}. 
See our Definition~\ref{def:out-after}(2). 
\begin{prop}\label{prop:after}
	Let $\yS=\yltsS$ be an LTS. 
	For all $p\in S$, $\ysi\in L_\tau^\star$ we have
	$$p \yafter\!\!_T\,\ysi=\{q\yst \ytrt{p}{\ysi}{q}\}=p \yafter\!\!_A\,\ysi.$$
\end{prop}
\begin{proof} Immediate from Proposition~\ref{prop:derivas}. \end{proof}
From now on we may write $\yafter$ for both $\yafter\!\!_T$ and $\yafter\!\!_A$.

In order to leave no room for confusion we let  $\yltscA{L}$ be the class of all LTS models over the alphabet $L$ as defined here.
We designate by $\yltscT{L}$ the class of all LTS models according to Definition~5(12)~\cite{tret-model-2008}. 
As a refinement, the class of all models $\yS=\yltsS$ in $\yltscT{L}$ where all states are reachable from the initial state, that is, $\ytr{s_0}{}{s}$ for all $s\in S$, will be designated as $\yltscTR{L}$.

Each model in $\yltscT{L}$ is assumed to be  image finite and strongly converging (Definition~5(12)~\cite{tret-model-2008}), where an LTS $\yS=\yltsS$ is said to be 
\begin{enumerate}
	\item \textbf{image finite}  when $p\yafter \ysi$ is finite, for all $p\in S$ and all $\ysi\in L^\star_\tau$ (see its Definition 5(10)). 
	\item \textbf{strongly converging}  if there is no state  that can perform an infinite sequence of internal transitions (see its Definition 5(11)).
\end{enumerate}

We readily have the following result.
\begin{prop}\label{prop:class-lts}
	Assume hypothesis~\ref{hypo:fin}. Then
	\begin{enumerate}
		\item  $\yltscA{L}\subsetneq  \yltscT{L}$ (properly contained)
		\item  $\yltscTR{L}= \yltscA{L}$ 
	\end{enumerate}
\end{prop}
\begin{proof} 
	Let $\yS=\yltsS$ be an LTS in $\yltscA{L}$.
	Under Hypothesis~\ref{hypo:fin}, $\yS$ is image finite. 
	According to Remark~\ref{rema:lte-finite}, there is no transition $(s,\tau,s)$ in $T$, and so $\yS$ is also strongly converging.
	This proves (1).
	
	For (2), let $\yS=\yltsS$ be an LTS in $\yltscTR{L}$. 
	Since $\yS$ is strongly converging, there can be no transition $(s,\tau,s)$ in $T$.
	From the definition of the class $\yltscTR{L}$ we know that for all $s\in S$ we must have $\ytr{s_0}{}{s}$.
	We conclude that Remark~\ref{rema:lte-finite} is satisfied and so $\yS\in \yltscA{L}$.
	Hence, $\yltscTR{L}\ysse \yltscA{L}$.
	Using item (1) we conclude the proof. 
\end{proof}

\subsection*{Models with inputs and outputs} 

Let $L=L_I \cup L_U$ with $L_I\cap L_U=\yemp$ be alphabets.
Then,  Definition~\ref{def:iolts} says that $(S,s_0,L_I,L_U,T)$ is an IOLTS with input alphabet $L_I$
and output alphabet $L_U$ when $(S,s_0,L,T)\in \yltscA{L}$, that is, an IOLTS is a LTS where the alphabet has been partitioned into disjoint sets of input and output action symbols.
In this appendix we will designate the class of IOLTSs over $L_I$ and $L_U$ by $\yiocA{I}{U}$.
Likewise, in~\cite{tret-model-2008}, Definition 6 says that a labeled transition system with inputs and outputs is a system  $(S,L_I,L_U,T,s_0)$, where $(S,L,T,s_0)\in \yltscT{L}$, and we will here denote the class of all such models by $\yltscT{L_I,L_U}$.
\begin{prop}\label{prop:class-iolts}
	Again, assume hypothesis~\ref{hypo:fin}. Then
	\begin{enumerate}
		\item  $\yiocA{I}{U}\subsetneq \yltscT{L_I,L_U}$  (properly contained)
		\item  $\yltscTR{L_I,L_U}= \yiocA{I}{U}$ 
	\end{enumerate}
\end{prop}
\begin{proof} 
	Immediately from the definitions and from the Proposition~\ref{prop:class-lts}. 
\end{proof}

\subsection*{Quiescent states} 

First we introduce some notation. 
Let $A$ be any alphabet, and define $A^{+\yde}=A\cup\{\yde\}$
and $A^{-\yde}=A-\{\yde\}$.

Let $\yS=(S,s_0,L_I,L_U,T)\in\yltscT{L_I,L_U}$ be an IOLTS. 
In~\cite{tret-model-2008}, Definition~8.1 says that a state $s\in S$ is quiescent (in $\yS$), denoted  $\yde_T^\yS(s)$, if for all $x\in L_U\cup\{\tau\}$ we have $s\not\overset{\!\!\!x}{\rightarrow}$ in $\yS$.
When the model is clear from the context we may write only $\yde_T$ instead of $\yde_T^\yS$. 
Further, given an IOLTS $\yS=(S,s_0,L_I,\yLUmd,T)\in\yltscT{L_I,\yLUmd}$, Definition~9~\cite{tret-model-2008} creates a new model $\yS^\yde=\yiolts{S}{s_0}{L_I}{\yLUpd}{T\cup T_\yde}$, where $T_\yde = \{(s,\yde,s)\yst \yde^\yS_T(s)\}$.
We designate this new class of IOLTSs by $\yltscTd{L_I,\yLUmd}=\{\yS^\yde\yst \yS\in\yltscT{L_I,\yLUmd}\}$
to stress that these models were constructed from IOLTSs whose output alphabet did not contain $\yde$
(so $\yLUmd$ in the notation), but the output alphabet of the extended model always contains $\yde$ (so  $\yltscTd{}$ in the notation.)
If the original model was from the class $\yltscTR{L_I,\yLUmd}$, meaning that all states are reachable from the initial state, then the new class of models will be designated by  $\yltscTRd{L_I,\yLUmd}$. 

Next proposition states a simple result. 
\begin{prop}\label{prop:class-qTt}
	$\yltscTd{L_I,\yLUmd}\subsetneq \yltscT{L_I,\yLUpd}$ and  $\yltscTRd{L_I,\yLUmd}\subsetneq \yltscTR{L_I,\yLUpd}$ (properly contained).
\end{prop} 
\begin{proof}
	Immediate from the definitions.  
\end{proof}

Now we turn to Definition~\ref{def:ioltsq} where quiescence in $\yde$-IOLTS models is introduced in this work. 
To emphasize that $\yde$ is always a symbol in the output alphabet of an $\yde$-IOLTS, in this appendix we designate this class of models  by $\yiocAdLUpd{I}{U}$ (in the main text it was designated simply by $\yiocq{I}{U}$).
Let $\yS=(S,s_0,L_I,\yLUpd,T)\in\yiocAdLUpd{I}{U} $ be a $\yde$-IOLTS.
If $q\in S$ is quiescent according to Definition~\ref{def:ioltsq}, we write  $\yde^\yS_A(q)$, and may omit the index ${\,}^\yS$ when no confusion can arise.

Proposition~\ref{prop:class-delta} is important because it gives the exact relationship between the class of all $\yde$-IOLTS models and the class of all models after they are extended in order to include quiescence, as defined in~\cite{tret-model-2008}.
\begin{prop}\label{prop:class-delta}
	Assume hypothesis~\ref{hypo:fin}, then we have that 
	\begin{enumerate}
		\item  $\yiocAdLUpd{I}{U}\subsetneq \yltscTd{L_I,\yLUmd}$  (properly contained)
		\item  $\yltscTRd{L_I,\yLUmd}= \yiocAdLUpd{I}{U}$ 
	\end{enumerate}
\end{prop}
\begin{proof} 
	To prove item (1), assume $\yS\in \yiocAdLUpd{I}{U}$.
	From the definition of the class $ \yiocAdLUpd{I}{U}$ we obtain $\yS=\yiolts{S}{s_0}{L_I}{\yLUpd}{T}$, where
	\begin{align}
		&\yS\in\yiocALUpd{I}{U}\label{eq:delta1}; \text{and} \\
		&(s,\yde,q)\in T\quad\text{if and only if (a) $s=q$; (b) when $\ytr{s}{x}$ with  $x\in L_I\cup \yLUpd\cup\{\tau\}$, then 
			$x\in L_I\cup\{\yde\}$}\label{eq:delta3}. 
	\end{align}
	From the definition of the class $\yltscTd{L_I,\yLUmd}$,  we need to show that $\yS=\yT^\yde$ for some $\yT\in \yltscT{L_I,\yLUmd}$. 
	Take $\yT=\yiolts{S}{s_0}{L_I}{\yLUmd}{R}$, where 
	\begin{align}
		&R=T-\{(s,\yde,q)\yst s,q\in S\}.\label{eq:delta4}
	\end{align}
	Since there are no $\yde$-transitions in $R$, the output alphabet of $\yT$ can be taken as $\yLUmd$.
	Since $\yS\in  \yiocAdLUpd{I}{U}$, we get $\yT\in \yioltsn{IO}_A(L_I,\yLUmd)$.
	From Proposition~\ref{prop:class-iolts} we have $\yT\in \yltscT{L_I,\yLUmd}$.
	It remains to show that $\yS=\yT^\yde$.
	
	From the definitions we have $\yT^\yde=\yiolts{S}{s_0}{L_I}{\yLUpd}{R\cup R_\yde}$, where $R_\yde=\{(s,\yde,s)\yst\yde^\yT_T(s)\}$. 
	In order to complete the proof of item (1), we need show that $R\cup R_\yde=T$.
	
	Let $(s,x,q)\in R$.
	From the Eq. (\ref{eq:delta4}) we have $(s,x,q)\in T-\{(s,\yde,q)\yst s,q\in S\}$, and then  $(s,x,q)\in T$.
	
	Let  $(s,x,q)\in R_\yde$.  
	From definition of the class $\yltscTd{L_I,\yLUmd}$ we have that $s=q$ and $x=\yde$. 
	So, $\yde^\yT_T(s)$ in $\yT$.
	Now assume $s\in S$ and $\ytr{s}{y}{}$ with $y\in L_I\cup \yLUpd\cup\{\tau\}$. 
	Since we have $\yde^\yT_T(s)$ in $\yT$, we get 
	$y\not \in \yLUmd\cup\{\tau\}$. 
	Hence $y\in L_I\cup \{\yde\}$.
	From Eq. (\ref{eq:delta3}), we obtain $(s,x,q)=(s,\yde,s)\in T$ and then we conclude that  $R\cup R_\yde\ysse T$. 
	
	In order to prove $T \ysse R\cup R_\yde$, take $(s,x,q)\in T$ with  $x\neq \yde$.
	From Eq. (\ref{eq:delta4}) we get $(s,x,q)\in R$, and then $(s,x,q)\in R\cup R_\yde$.
	Now assume that $(s,\yde,q)\in T$.
	From Eq. (\ref{eq:delta3}) we have $s=q$ and the following condition is satisfied in $s$: if $\ytr{s}{x}{}$ with $x\in L_I\cup \yLUpd\cup\{\tau\}$ then $x\in L_I\cup\{\yde\}$. 
	Next, we show that  $\yde^\yT_T(s)$ is in $\yT$, which gives $(s,\yde,s)\in R_\yde$, and then $(s,\delta,q)=(s,\yde,s)\in R\cup R_\yde$. 
	This will imply $T\ysse R\cup R_\yde$, so that $T=R\cup R_\yde$, completing the proof.
	
	So, assume that we do not have $\yde^\yT_T(s)$ in $\yT$.
	From the definition we would have $\ytr{s}{x}{}$ with $x\in \yLUmd\cup\{\tau\}$.
	Therefore, $x\in L_I\cup \yLUpd\cup\{\tau\}$.
	From  Eq. (\ref{eq:delta3}) we have $x\in L_I\cup\{\yde\}$, which  contradicts with $x\in \yLUmd\cup\{\tau\}$.
	Hence, $\yde^\yT_T(s)$ in $\yT$, and the proof is complete. 
	
	Now we turn to  item (2). 
	Let $\yS\in \yiocAdLUpd{I}{U}$. 
	From item (1) we get $\yS\in \yltscTd{L_I,\yLUmd}$.
	Since every state of $\yS$ is reachable from its initial state, it follows that $\yS\in\yltscTRd{L_I,\yLUmd}$, and we can conclude that $\yiocAd{I}{U}\ysse \yltscTRd{L_I,\yLUmd}$. 
	
	Now we prove that $ \yltscTRd{L_I,\yLUmd}\ysse\yiocAdLUpd{I}{U}$.
	Let $\yS\in\yltscTRd{L_I,\yLUmd}$ be an IOLTS. 
	We want to show that $\yS\in\yiocAdLUpd{I}{U}$.
	From the definition of $\yltscTRd{L_I,\yLUmd}$ we know that $\yS=\yT^\yde$ for some $\yT=\yiolts{S}{s_0}{L_I}{\yLUmd}{T}\in\yltscTR{L_I,\yLUmd}$.
	From the definition we also know that $\yT^\yde=\yiolts{S}{s_0}{L_I}{(\yLUmd)^{+\yde}}{T\cup T_\yde}$, where 
	$T_\yde=\{(s,\yde,s)\in T\yst \yde^\yT_T(s)\}$.
	We want to show that $\yS=\yT^\yde\in\yiocAdLUpd{I}{U}$. 
	From the definition of the class $\yiocAdLUpd{I}{U}$, this is equivalent to show that 
	\begin{align}
		&\yT^\yde\in\yiocALUpd{I}{U}\label{eq:delta5}; \text{and}\\
		&(s,\yde,q)\in T\cup T_\yde\quad\text{if and only if  (a) $s=q$; (b) if $\ytr{s}{x}$ with   $x\in L_I\cup \yLUpd\cup\{\tau\}$ then $x\in L_I\cup\{\yde\}$}\label{eq:delta6}. 
	\end{align}
	
	Since $\yT=\yiolts{S}{s_0}{L_I}{\yLUmd}{T}$, and  $(\yLUmd)^{+\yde}=\yLUpd$, we have that $\yT^\yde\in \yltsn{LTSR}_T(L_I,\yLUpd)$.
	From Proposition~\ref{prop:class-iolts}(2) we obtain $\yT^\yde\in\yiocALUpd{I}{U}$ and the Eq. (\ref{eq:delta5}) holds. 
	
	It remains to prove Eq.~(\ref{eq:delta6}).
	Let $(s,\yde,q)\in T\cup T_\yde$.
	Since $\yT=\yiolts{S}{s_0}{L_I}{\yLUmd}{T}$ and $\yde\not\in \yLUmd$, it follows that $(s,\yde,q)\not\in T$ and so $(s,\yde,q)\in T_\yde$.
	We then have $s=q$ and $\yde^\yT_T(s)$ in $\yT$. 
	Therefore, Eq. (\ref{eq:delta6}a) holds. 
	Now we assume $\ytr{s}{x}$ with $x\in L_I\cup \yLUpd\cup\{\tau\}$. 
	Since  $(s,\yde,s)\in T_\yde$ and $\yde^\yT_T(s)$ in $\yT=\yiolts{S}{s_0}{L_I}{\yLUmd}{T}$, we should have $x\not\in \yLUmd\cup \{\tau\}$.
	Then, $x\in L_I\cup\{\yde\}$ and Eq. (\ref{eq:delta6}b) also holds. 
\end{proof}

Now Proposition~\ref{prop:delta} states quiescence and relate them between the approaches. 
\begin{prop}\label{prop:delta}
	Let $\yS=\yiolts{S}{s_0}{L_I}{\yLUpd}{T}\in\yiocAdLUpd{I}{U}$.
	We then have that $\yS=\yT^\yde$ where $\yT=\yiolts{S}{s_0}{L_I}{\yLUmd}{R\cup R_\yde}\in\yltscTR{L_I,\yLUmd}$.
	Moreover, for all $s\in S$ we get $\yde^\yS_A(s)$ in $\yS$ if, and only if, $\yde^\yT_T(s)$ in $\yT$.
\end{prop}
\begin{proof}
	From Proposition~\ref{prop:class-delta}(2) we have  $\yiocAdLUpd{I}{U}=\yltscTRd{L_I,\yLUmd}$.
	The definition of the class $\yltscTd{L_I,\yLUmd}$ gives $\yS=\yT^\yde$, where  $\yT=\yiolts{S}{s_0}{L_I}{\yLUmd}{R}\in\yltscTR{L_I,\yLUmd}$, with $T=R\cup R_\yde$ and  $R_\yde=\{(p,\yde,p)\yst \text{$\yde^\yT_T(p)$ in $\yT$} \}$.
	
	Assume that $\yde^\yS_A(s)$ holds in $\yS$ and $\yde^\yT_T(s)$ does not hold in $\yT$.
	Definition of $\yde^\yT_T(\cdot)$ gives that $\ytr{s}{x}{}$ in $\yT$ with $x\in\yLUmd\cup\{\tau\}$.
	So, $(s,x,p)\in R$ for some $p\in S$, and then $(s,x,p)\in T$.
	Therefore, $\ytr{s}{x}{}$ in $\yS$.
	Since we have $\yde^\yS_A(s)$ in $\yS$, the definition of $\yde^\yS_A(\cdot)$ together with $\ytr{s}{x}{}$ in $\yS$ gives  that $x\in L_I\cup\{\yde\}$, contradicting $x\in\yLUmd\cup\{\tau\}$.
	
	On the other direction, we assume $\yde^\yT_T(s)$ em $\yT$.
	From the definition of $R_\yde$ we have $(s,\yde,s)\in R_\yde$.
	Hence $(s,\yde,s)\in T$, i.e., $\ytr{s}{\yde}{}$ in $\yS$.
	From the definition of $\yde^\yS_A(\cdot)$ we also have $\yde^\yS_A(s)$, concluding the proof. 
\end{proof}

\subsection*{The $\yout$ relation} 

In~\cite{tret-model-2008},  Definition~11  introduces the $\yout$ relation, here denoted $\yout_T$, thus:
Let $\yS=\yiolts{S}{s_0}{L_I}{L_U}{T}\in\yltscT{L_I,L_U}$.
Then, for all $s\in S$ and all $Q\subseteq S$, 
$$\yout_T(s)=\{x\in L_U |  \ytr{s}{x}{} \} \cup \{\yde | \yde^\yS_T(s) \}\quad\text{and}\quad \yout_T(Q)= \bigcup \{out_T(s) | s\in Q\}.$$

In this work, we define the same relation, denoted $\yout_A$, as follows:
Let $\yS=\yiolts{S}{s_0}{L_I}{L_U}{T}\in \yiocA{I}{U}$. 
Then, for all $Q\ysse S$
$$\yout_A(Q)= \bigcup\limits_{s\in Q}\{x\in L_U\yst \ytrt{s}{x}{}\},$$
See Definition~\ref{def:out-after}. 

Next proposition shows that both definitions of $\yout$ coincide. 
\begin{prop}\label{prop:out1}
	Let  $\yS=(S,s_0,L_I,\yLUpd,T)\in \yiocAdLUpd{I}{U}$ be an IOLTS. 
	Then, we have that $\yout_A(s \yafter\!\!_A \,\ysi) = \yout_T(s \yafter\!\!_T \,\ysi)$ for all $s \in S$ and all $\ysi\in L_\tau^\star$.
\end{prop}
\begin{proof} 
	From Proposition~\ref{prop:derivas} we have $\ytrtT{}{}{}\,\,=\,\,\ytrtA{}{}{}$, and so the indexes may be omitted. 
	Likewise, we can write $\ytr{}{}{}$ instead of $\ytrT{}{}{}$ or $\ytrA{}{}{}$. 
	Also from Proposition~\ref{prop:after} we have $\yafter\!\!_A=\yafter\!\!_T$ and we write $\yafter$\!\!.
	
	Let $Q= s\yafter \ysi$.
	We want to prove that $\yout_A(Q)=\yout_T(Q)$.
	
	Let $x \in \yout_A(Q)$, $x\in \yLUpd$. 
	Since  $\yS\in \yiocAdLUpd{I}{U}$ it is clear that $\yS\in \yiocALUpd{I}{U}$.
	Hence, from the definition of $\yout_A$ we obtain some $q\in Q$ such that $\ytrt{q}{x}{}$. 
	Thus, from Definition~\ref{def:path} we must have $\mu=\tau^k x$ for some $k\geq 0$, and  
	$\ytr{q}{\mu}{}$. 
	Then, we have some $p\in S$ such that $\ytr{q}{\tau^k}{p}\ytr{}{x}$.
	Since $Q=s\yafter \ysi$ and $q\in Q$ we have $\ytrt{s}{\ysi}{q}$.
	From Definition~\ref{def:path} we have some $\rho\in L_\tau^\star$ such that $\ytr{s}{\rho}{q}$ and $h_\tau(\rho)=\ysi$.
	Composing, we get $\ytr{s}{\rho\tau^k}{p}$.
	Since $h_\tau(\rho\tau^k)=h_\tau(\rho)=\ysi$ we obtain $\ytrt{s}{\ysi}{p}$, so that $p\in s \yafter \ysi$, that is, $p\in Q$.
	From Proposition~\ref{prop:class-delta}(1) we have $ \yiocAdLUpd{I}{U}\subsetneq \yltscTd{L_I,\yLUmd}$. 
	Hence, from Proposition~\ref{prop:class-qTt} we get $\yS\in \yltscT{L_I,\yLUpd}$. 
	If $x\in L_U$, since  we already have $p\in Q$ and $\ytr{p}{x}{}$, the definition of $\yout_T$ gives $x\in \yout_T(Q)$.
	If $x=\yde$ then we must have a transition $(p,\yde,t)\in T$ for some $t\in S$.
	Then, following Definition~\ref{def:ioltsq}, we know that $p$ is quiescent, that is, we have $\yde_A(p)$. 
	Using Proposition~\ref{prop:delta} we get $\yde_T(q)$ and the 
	definition of $\yout_T$ gives again $x\in \yout_T(Q)$.
	We conclude that $\yout_A(Q)\ysse \yout_T(Q)$.
	
	Now assume that $x\in\yout_T(Q)$.
	We must have some $q\in Q$ such that: 	
	(i) $\yde_T(q)$ and $x=\yde$, or 
	(ii) $\ytr{q}{x}{}$ and $x\in L_U$. 
	First assume (i).
	Using Proposition~\ref{prop:delta} we get $\yde_A(q)$ and 
	Definition~\ref{def:ioltsq} says that we have a transition $(q,\yde,q)\in T$. 
	Hence, $\ytr{q}{\yde}{}$ and  Definition~\ref{def:path} gives  $\ytrt{q}{\yde}{}$.
	Since $q\in Q$, the definition of $\yout_A$ gives $\yde\in \yout_A(Q)$.
	Now, assume (ii). 
	From Definition~\ref{def:path} we infer  $\ytrt{q}{x}{}$.
	Again, $q\in Q$, and the definition of $\yout_A$ gives $x\in \yout_A(Q)$.
	We conclude that $\yout_T(Q)\ysse\yout_A(Q)$.
\end{proof}

\subsection*{Input-enabledness property} 

In~\cite{tret-model-2008}, Definition 7 says that the class of all input-output transition systems with inputs in $L_I$ and outputs in $L_U$ is a restricted subclass of $\yltscT{L_I\cup L_U}$. More specifically, $(S,L_I,L_U,T,s_0)\in \yltscT{L_I\cup L_U}$ is an input-output transition system if any reachable state $s$ is input-enabled, that is, there is a transition out of $s$ for all input symbols.
More formally, in Definition 5(6) we have   $der_T(p)=\{q\yst \ytrt{p}{}{q}\}$. 
Then  Definition 7  decrees that
$(S,L_I,L_U,T,s_0)\in \yltscT{L_I\cup L_U}$ is an input-output transition system when for all $p\in der_T(s_0)$ we have $\ytrt{p}{a}{}$ for all $a\in L_I$.
The class of all input-output systems is here denoted by $\yiocT{L_I,L_U}$.
The subclass of all input-output systems $\yS$ where all states are reachable from the initial states, that is when $\yS$ is in the subclass  $\yltscTR{L_I,L_U}$, will be designated by $\yiocTR{L_I,L_U}$.

In our work, Definition~\ref{def:inic-tret} says that an IOLTS $(S,s_0,L_I,L_U,T)\in\yiocA{I}{U}$ is input-enabled when $\yinp(s)=L_I$ for all $s\in S$.
In the same definition we find that $\yinp(s)=\{\ell\in L_I\yst \ytrt{s}{\ell}\}$. 
In this appendix we designate by $\yioceA{I}{U}$ the class of all input-enabled IOLTS models over input alphabet $L_I$ and output alphabet $L_U$. 
\begin{prop}\label{prop:class-iolts2}
	Under  Hypothesis~\ref{hypo:fin} we have 
	\begin{enumerate}
		\item  $\yioceA{I}{U}\subsetneq\yiocT{L_I,L_U}$ (properly contained)
		\item $\yiocTR{L_I,L_U}= \yioceA{I}{U}$ 
	\end{enumerate}
\end{prop}
\begin{proof}  
	First  note  $\yinp(p)=L_I$ if and only if $\ytrt{p}{a}{}$ for all $a\in L_I$.
	With this observation, we prove item (1) using Proposition~\ref{prop:class-iolts}(1), and prove item (2), using Proposition~\ref{prop:class-iolts}(2). 
\end{proof}


Now we are in position to compare the definitions of the {\bf ioco} relation.
In~\cite{tret-model-2008}, the definition of the {\bf ioco} relation depends on the notion  of $Straces$, as stated in its Definition~9.
It proceeds as follows:
Let $\yS=(S,s_0,L_I,L_U,T)\in\yltscT{L_I,L_U}$  and let  $\yS^\yde=\yT\in\yltscTd{L_I,\yLUpd}$
be its extension to include quiescence.
Definition~9 says that, for all $s\in S$ and all $\ysi\in(L_I \cup L_U\cup{\{\yde\}})^\star$, we have  $\ysi\in Straces(s)$  if $\ytrut{s}{\ysi}{\yT}{}$.
So, strings in $Straces(\yS)$ are just the observable traces of $\yT$, i.e., $Straces(\yS)=otr(\yT)$, the observable traces of $\yT$.
With the notion of $Straces$, Definition~12~in~\cite{tret-model-2008} specifies the $\yiocoT$ relation thus:
Let $\yI=(Q,q_0,L_I,L_U,R)\in \yiocT{L_I,L_U}\ysse\yltscT{L_I,L_U}$ and let   $\yQ=\yI^\yde\in\yiocTd{L_I,\yLUpd}$ be its extension.
Also let $\yS=(S,s_0,L_I,L_U,T)\in \yltscT{L_I,L_U}$ with $\yT=\yS^\yde\in \yltscTd{L_I,\yLUpd}$ being its extension. 
Then, $\yI \yiocoT \yS$ if, and only if,  for all $\ysi\in Straces(\yS)$, we have
$$\yout_T^\yQ(q_0 \yafterm{\yQ} \ysi)\ysse \yout_T^\yT(s_0 \yafterm{\yT} \ysi).$$ 
Here $\yout_T^\yQ$ indicates that the set $\yout_T$ is being obtained in $\yS$ and, 
similarly,  $\yafterm{\yQ}$  indicates that the set $\yafter$ is being calculated in the $\yS$, according to previous definitions.

Since $Straces(s_0)=otr(\yT)$, we rewrite the definition of $\yiocoT$ as follows. 
Assume $\yI=(Q,q_0,L_I,L_U,R)\in \yiocT{L_I,L_U}\ysse\yltscT{L_I,L_U}$ and  $\yS=(S,s_0,L_I,L_U,T)\in\yltscT{L_I,L_U}$.
We say that $\yI \yiocoT \yS$ if, and only if,  for all $\ysi\in otr(\yS^\yde)$, we have
$$\yout_T(q_0 \yafter \ysi)\ysse \yout_T(s_0 \yafter \ysi),$$ 
where $\yout_T$ and $\yafter$ are obtained over the corresponding extended models $\yI^\yde$ and $\yS^\yde$.

In this work,  when $\yS=(S,s_0,L_I,L_U,T)$ and  $\yI=(Q,q_0,L_I,L_U,R)$ are in $\yiocA{I}{U}$, we say that $\yI  \yiocoA \yS$ if, and only if, for all $\ysi\in otr(\yS)$, we have
$$\yout_A(q_0 \yafter \ysi)\ysse \yout_A(s_0 \yafter \ysi).$$ 
See Definition~\ref{def:out-after}.

Next proposition establishes the relationship between relations $\text{\bf ioco}_T$ and $\text{\bf ioco}_A$. 
\begin{prop}\label{prop:ioco2}
	Let $\yI\in\yiocTR{L_I,\yLUmd}$ and let $\yS\in  \yltscTR{L_I,\yLUmd}$, with  $\yQ=\yI^\yde$ and  $\yT=\yS^\yde$ being their corresponding extended models.
	Then,  $\yI \yiocoT \yS$ if and only if $\yQ \yiocoA \yT$.
\end{prop}
\begin{proof}
	Let $\yS=(S,s_0,L_I,\yLUmd,T)$ and  $\yI=(Q,q_0,L_I,\yLUmd,R)$.
	From the definitions we get immediately  $\yT=(S,s_0,L_I,\yLUpd,T\cup T_\yde)\in\yltscTRd{L_I,\yLUmd}$ and $\yQ=(Q,q_0,L_I,\yLUpd,R\cup R_\yde)\in\yltscTRd{L_I,\yLUmd}$, where $T_\yde=\{(s,\yde,s)|\yde^\yS_T(s)\}$ and  $R_\yde=\{(s,\yde,s)|\yde^\yI_T(s)\}$. 
	
	We first show that $\yI \yiocoT \yS$ implies that $\yQ \yiocoA \yT$.
	From Proposition~\ref{prop:class-delta}(2) we have $\yT,\yQ\in\yiocAdLUpd{I}{U}$, and so the 
	relation $\text{\bf ioco}_A$ is also defined for $\yQ$ and $\yT$.
	Assume that we have $\yI \yiocoT \yS$, but $\yQ \yiocoA \yT$ does not hold.
	From the definition we have some $\ysi\in otr(\yT)$ and some $x\in \yLUpd$ such that 
	$$x\in\yout_A^\yQ(q_0 \yafterm{\yQ} \ysi)\quad\text{and}\quad x\not\in\yout_A^\yT(s_0 \yafterm{\yT} \ysi).$$
	Since $\yQ\in \yiocAdLUpd{I}{U}$, using Proposition~\ref{prop:out1} we have 
	$\yout_A^\yQ(q_0 \yafterm{\yQ} \ysi)=\yout_T^\yQ(q_0 \yafterm{\yQ} \ysi)$ and then 
	$x\in \yout_T^\yQ(q_0 \yafterm{\yQ} \ysi)$.
	Since $\yI \yiocoT \yS$ holds we should have $x\in \yout_T^\yT(s_0 \yafterm{\yT} \ysi)$.
	Likewise, $\yT\in \yiocAdLUpd{I}{U}$ and Proposition~\ref{prop:out1} now gives $\yout_A^\yT(q_0 \yafterm{\yT} \ysi)=\yout_T^\yT(q_0 \yafterm{\yT} \ysi)$.
	Hence, $x\in\yout_A^\yT(s_0 \yafterm{\yT} \ysi)$ and we have reached a contradiction. 
	Therefore, if $\yI \yiocoT \yS$ then $\yQ \yiocoA \yT$.
	
	On the other direction, we need to show that if $\yQ \yiocoA \yT$ then $\yI \yiocoT \yS$.
	Now, assume that we have $\yQ \yiocoA \yT$, but $\yI \yiocoT \yS$ does not hold.
	The reasoning is entirely analogous and again we would reach a contradiction. 
\end{proof}

\subsection*{Determinism} 

We also show that the notion of determinism defined in~\cite{tret-model-2008} coincides with our definition of determinism.  
In \cite{tret-model-2008}, Definition 5(9) characterizes determinism as follows:  $\yS=(S,L,T,s_0)\in\yltscT{L}$ is deterministic if, for all $\ysi \in L^\star$, $s_0 \yafter \ysi$ has at most one element.
We indicate this by writing $determ_T(\yS)$.

In our work, Definition~\ref{def:lts-deterministic} says that $\yS=(S,s_0,L,T)\in\yltscA{L}$ is deterministic if $\ytrt{s_0}{\ysi}{s_1}$ and $\ytrt{s_0}{\ysi}{s_2}$ imply $s_1=s_2$, for all $s_1$, $s_2\in S$ and all $\ysi\in L^\star$.
If that is the case, we write $determ_A(\yS)$.

The next proposition shows that these notions coincide for every element in  $\yltscA{L}$.
\begin{prop}\label{prop:det}
	Let $\yS=(S,s_0,L,T)\in\yltscA{L}$.
	Then $determ_T(\yS)$ if and only if $determ_A(\yS)$. 
\end{prop}
\begin{proof}
	Note that from Proposition~\ref{prop:class-iolts}(1) we also have  $\yS\in\yltscT{L}$.
	
	Assume that $determ_T(\yS)$ holds.
	We have that $s_0 \yafter\!\!_T\, \ysi$ has at most one element if and only if $\ytrtT{s_0}{\ysi}{s_1}$ and $\ytrtT{s_0}{\ysi}{s_2}$ imply $s_1=s_2$, for all $s_1$, $s_2\in S$, and all $\ysi\in L^\star$.
	From Proposition~\ref{prop:derivas} we get that $\ytrtA{s_0}{\ysi}{s_1}$ and $\ytrtA{s_0}{\ysi}{s_2}$ imply $s_1=s_2$, so that $determ_A(\yS)$ also holds.
	
	For the other direction just reverse the argument.
\end{proof}

\subsection*{Test cases and test purposes} 

When defining the class of all test cases with inputs $L_I$ and outputs $L_U$, \cite{tret-model-2008} starts, at Definition 10, with a model $\yT=(S,s_0,\yLUmd,\yLIptmd,T)$  in the class $\yiocT{\yLUmd,\yLIptmd}$, that is, $\yT$ is already input-enabled with respect to $\yLUmd$.
Further restrictions apply, namely:
\begin{enumerate}
	\item $\yT$ is  finite state and deterministic as defined in~\cite{tret-model-2008}. We recall that according to Definition~\ref{def:lts} all our models are finite state.
	\item $S$ contains two distinct states, $\ypass$ and $\yfail$, and $\yout(\ypass)=\yout(\yfail)=\yLUptmd$.
	\item $\yT$ has no cycles, except at states $\ypass$ and $\yfail$, that is $\ytrt{s}{\ysi}{s}$ implies $s=\ypass$ or $s=\yfail$ for any $\ysi\neq\yeps$ and $\ysi\in (\yLUmd\cup\yLIptmd)^\star$. 
	\item For all state $s\in S$, we must have $\yinit_T(s)=\{x\}\cup \yLUmd$, for some $x\in \yLIptmd$.
\end{enumerate}
For any model $\yS=(S,s_0,L_I,L_U,T)$, Definition 5 in~\cite{tret-model-2008} says that for all $s\in S$
we have $\yinit_T(s)=\{x\in L_I\cup L_U\cup\{\tau\}\yst \ytr{s}{x}{}\}$.
Hence, for any model $\yT=(S,s_0,\yLUmd,\yLIptmd,T)$ condition 4 reduces to saying that  for all $s\in S$
we must have
$$\{x\}\cup \yLUmd=\{y\in  \yLIptmd\cup \yLUmd\cup\{\tau\}\yst \ytr{s}{y}{}\}\quad\text{for some $x\in \yLIptmd$}.$$
We conclude that $\ytrt{s}{x}{}$ for all $x\in\yLUmd$ and all $s\in S$.
Thus, when condition 4 holds we know that $\yT$ is already input-enabled.

In~\cite{tret-model-2008}, Definition 16, the $\yte$ symbol in test cases synchronizes with the $\yde$ symbol that signals quiescence in IUTs.  
Here we have used the same $\yde$ symbol in test purposes to synchronize with the $\yde$ symbol that 
flags quiescence in IUTs, as made explicit in Definition~\ref{def:ioltsq}.
So, specifications and IUTs are models from $\yiocALUmdLIpd{I}{U}$.
When constructing test purposes from given specifications, the input and output alphabets are interchanged.
Accordingly, a test purpose over the input alphabet $L_I$ and output alphabet $L_U$ is defined as a model 
$\yT=\yiolts{S}{s_0}{\yLUpd}{\yLImd}{T}$  in $\yiocALUpdLImd{U}{I}$.

In our Definition~\ref{def:out-after}, given any model  $\yS=\yiolts{S}{s_0}{L_I}{L_U}{T}\in \yiocA{I}{U}$
we say that $\yout_A(s)= \{x\in L_U\yst \ytrt{s}{x}{}\}$ for all $s\in S$, and in Definition~\ref{def:in-compl-out-determ} we say that $\yS$ is output-deterministic when $\vert \yout_A(s)\vert=1$ for all $s\in S$.
Also such a model $\yS$ is  input-enabled when $\yinp_A(s)=L_I$.
Hence, a model $\yT=\yiolts{S}{s_0}{\yLUpd}{\yLImd}{T}$ is output-deterministic when $\vert \{x\in \yLImd\yst \ytrt{s}{x}{}\}\vert=1$ and it is input-enabled when $\yinp_A(s)=\yLUpd$, for all $s\in S$.
Moreover, since in our models we substitute $\yde$ for $\yte$, conditions (2), (3) and (4) should read as follows: 
\begin{enumerate}
	\item[($2'$)] $S$ contains two distinct states, $\ypass$ and $\yfail$, and $\yout(\ypass)=\yout(\yfail)=\yLUpd$.
	\item[($3'$)] $\yT$ has no cycles, except at states $\ypass$ and $\yfail$, that is $\ytrt{s}{\ysi}{s}$ implies $s=\ypass$ or $s=\yfail$ for any $\ysi\neq\yeps$ and $\ysi\in (\yLUpd\cup\yLImd)^\star$. 
	\item[($4'$)]  For all state $s\in S$, we must have $\yinit_T(s)=\{x\}\cup \yLUpd$, for some $x\in \yLImd$.
\end{enumerate}
Equivalently, condition ($4'$) can be written as
$$\{x\}\cup \yLUpd=\{y\in  \yLImd\cup \yLUpd\cup\{\tau\}\yst \ytr{s}{y}{}\}\quad\text{for some $x\in \yLImd$}.$$
The next result says that certain models $\yT\in\yiocALUpdLImd{U}{I}$ satisfy conditions (1) and (4) above.
\begin{prop}\label{prop:tp-tc}
	Let $\yT=\yiolts{S}{s_0}{\yLUpd}{\yLImd}{T}\in \yiocALUpdLImd{U}{I}$ be deterministic, input-enabled and output-deterministic.
	Then $\yT$  satisfies conditions (1) and ($4'$)  given above. 
\end{prop}
\begin{proof}
	Since $\yT$ is deterministic, using Proposition~\ref{prop:det} we conclude that $\yT$ is also deterministic in the sense defined in~\cite{tret-model-2008}.
	Thus, condition ($1$) holds.
	
	By the preceding discussion, if $\yT$ satisfies condition (4) then is is already input-enabled.
	Also from the text above, it remains to show that for all $s\in S$ we have 
	$$\{x\}\cup \yLUpd=\{y\in  \yLImd\cup \yLUpd\cup\{\tau\}\yst \ytr{s}{y}{}\}\quad\text{for some $x\in \yLImd$}.$$
	So, fix some $s\in S$.
	Since $\vert \yout_A(s)\vert=1$, we may choose $x\in\yLImd$ such that $\yout_A(s)=\{x\}$, and now we have to
	show that 
	$$\yout_A(s)\cup \yLUpd=\{y\in  \yLImd\cup \yLUpd\cup\{\tau\}\yst \ytr{s}{y}{}\}.$$ 
	So, let $y\in \yLImd\cup \yLUpd\cup\{\tau\}$ with $\ytr{s}{y}{}$.
	Since $\yT$ is deterministic, Proposition~\ref{prop:lts-deterministic} says that $\yT$ has no $\tau$-labeled transitions. 
	Hence, $y\in \yLImd\cup \yLUpd$.
	If $y\in \yLUpd$ then $y\in \yout_A(s)\cup \yLUpd$.
	Now assume that $y\in \yLImd$ with $\ytr{s}{y}{}$, so that $y\in\yout_A(s)$.
	Since $\vert \yout_A(s)\vert=1$ we have $\yout_A(s)=\{y\}$, and again $y\in \yout_A(s)\cup \yLUpd$.
	We have shown that $\{y\in  \yLImd\cup \yLUpd\cup\{\tau\}\yst \ytr{s}{y}{}\}\ysse \yout_A(s)\cup \yLUpd$.
	
	Now let $y\in \yout_A(s)\cup \yLUpd$.
	If $y\in \yout_A(s)$ then $y\in\yLImd$ and $\ytrt{s}{y}{}$. Since $\yT$ has no $\tau$-labeled transitions we get $\ytr{s}{y}{}$.
	Hence $y \in \{y\in  \yLImd\cup \yLUpd\cup\{\tau\}\yst \ytr{s}{y}{}\}$.
	Now assume $y\in\yLUpd$.
	Since $\yinp_A(s)=\yLUpd$ we get $y\in\yinp_A(s)$, so that $\ytrt{s}{y}{}$. 
	Because $\yT$ has no $\tau$-labeled transitions we have $\ytr{s}{y}{}$.
	Hence, $y\in\yLUpd$ and $\ytr{s}{y}{}$ give $y \in \{y\in  \yLImd\cup \yLUpd\cup\{\tau\}\yst \ytr{s}{y}{}\}$.
	We now have $\yout_A(s)\cup \yLUpd\ysse \{y\in  \yLImd\cup \yLUpd\cup\{\tau\}\yst \ytr{s}{y}{}\}$ and the proof is complete.
\end{proof}

In Theorem~\ref{theo:all-reqs} we showed that test purposes satisfying a number of restrictions can be constructed for any given specification. 
We can now verify that those test purpose models are also test cases, in the sense used in~\cite{tret-model-2008}.

\begin{prop}\label{prop:tps-are-tcs}
	Let $\yS$ be a specification in $\yiocALUmdLIpd{I}{U}$, and let $m\geq 1$.
	Then the set $TP$ of test purposes constructed in Theorem~\ref{theo:all-reqs} is {\bf ioco}-complete for
	$\yS$ with respect to any implementation with at most $m$ states.
	Moreover, any test purpose  in $TP$ satisfies conditions (1), ($2'$), ($3'$) and ($4'$) listed above.
\end{prop}
\begin{proof}
	By Theorem~\ref{theo:all-reqs} we know that $TP $ is $m$-{\bf ioco}-complete for $\yS$.
	
	Let $\yT=\yiolts{S}{s_0}{\yLUpd}{\yLImd}{T}\in \yiocALUpdLImd{U}{I}$ be a test purpose constructed in
	$TP$.
	By Theorem~\ref{theo:all-reqs} we know that $\yT$ is already deterministic, input-enabled and output-deterministic.
	So, by Proposition~\ref{prop:tp-tc} conditions (1) and ($4'$) are satisfied.
	
	By Theorem~\ref{theo:all-reqs}, $\yT$ has two distinct $\ypass$ and $\yfail$ states, and it is acyclic
	except for self-loops at these states.
	The proof of  Theorem~\ref{theo:all-reqs} starts with test purposes constructed at Proposition~\ref{prop:m-complete}.
	A simple examination of the proof of Proposition~\ref{prop:m-complete} shows that we explicitly
	add self-loops $(\yfail,\ell,\yfail)$ and $(\ypass,\ell,\ypass)$ to $\yT$, for all $\ell \in\yLUpd$.
	Hence, conditions ($2'$) and ($3'$) are also satisfied.
\end{proof}

\end{appendices}


\end{document}